\documentclass[final,3p,times,12pt]{elsarticle}
\usepackage[T1]{fontenc}
\usepackage[utf8]{inputenc}
\usepackage{microtype}
\usepackage{amsmath,amssymb,amsthm,mathtools}
\usepackage{enumitem}
\usepackage{booktabs}
\usepackage{aliascnt}
\usepackage{needspace}
\usepackage[hidelinks,bookmarksnumbered=true]{hyperref}
\usepackage[nameinlink,noabbrev]{cleveref}
\usepackage{tikz}
\usepackage{tocloft}

\biboptions{numbers,sort&compress}

\allowdisplaybreaks
\setlist{topsep=2pt plus 1pt minus 1pt,itemsep=1pt,parsep=0pt,partopsep=0pt}

\makeatletter
\def\thm@space@setup{%
  \thm@preskip=4pt plus 2pt minus 1pt
  \thm@postskip=4pt plus 2pt minus 1pt}
\makeatother

\AtBeginDocument{%
  \setlength{\parskip}{0pt}%
  \setlength{\abovedisplayskip}{5pt plus 2pt minus 2pt}%
  \setlength{\belowdisplayskip}{5pt plus 2pt minus 2pt}%
  \setlength{\abovedisplayshortskip}{2pt plus 1pt minus 1pt}%
  \setlength{\belowdisplayshortskip}{3pt plus 1pt minus 1pt}%
}

\newcommand{\Vsub}{\mathcal V_{\leq 1}}
\newcommand{\Vone}{\mathcal V_1}
\newcommand{\Vext}{\mathcal V_{\mathsf{ext}}}
\newcommand{\M}{\mathsf M}
\newcommand{\FVA}{\omega\mathbf{FVA}}
\newcommand{\FAC}{\mathcal{F}(\mathbf{C})}
\newcommand{\FS}{\mathbf{FS}}
\newcommand{\DCPO}{\mathbf{DCPO}}
\newcommand{\BC}{\mathbf{BC}}
\newcommand{\RB}{\mathbf{RB}}
\newcommand{\stle}{\leq_{\mathrm{st}}}
\newcommand{\id}{\operatorname{id}}
\newcommand{\supp}{\operatorname{supp}}
\newcommand{\Max}{\operatorname{Max}}
\newcommand{\Min}{\operatorname{Min}}
\newcommand{\Mon}{\operatorname{Mon}}
\newcommand{\Up}{\operatorname{Up}}
\newcommand{\cone}{\operatorname{cone}}
\newcommand{\Int}{\operatorname{int}}
\newcommand{\R}{\mathbb R}
\newcommand{\Rplus}{\mathbb R_{\geq 0}}
\newcommand{\N}{\mathbb N}
\newcommand{\Z}{\mathbb Z}
\newcommand{\up}{\mathord{\uparrow}}
\newcommand{\down}{\mathord{\downarrow}}
\newcommand{\waybelow}{\ll}

\newcommand{\Dirac}{\delta}
\newcommand{\etaV}{\eta}
\newcommand{\eps}{\varepsilon}

\theoremstyle{plain}
\newtheorem{theorem}{Theorem}[section]
\newaliascnt{proposition}{theorem}
\newtheorem{proposition}[proposition]{Proposition}
\aliascntresetthe{proposition}
\newaliascnt{lemma}{theorem}
\newtheorem{lemma}[lemma]{Lemma}
\aliascntresetthe{lemma}
\newaliascnt{corollary}{theorem}
\newtheorem{corollary}[corollary]{Corollary}
\aliascntresetthe{corollary}

\theoremstyle{definition}
\newaliascnt{definition}{theorem}
\newtheorem{definition}[definition]{Definition}
\aliascntresetthe{definition}
\newaliascnt{example}{theorem}
\newtheorem{example}[example]{Example}
\aliascntresetthe{example}
\newaliascnt{remark}{theorem}
\newtheorem{remark}[remark]{Remark}
\aliascntresetthe{remark}

\crefname{definition}{Definition}{Definitions}
\crefname{example}{Example}{Examples}
\crefname{theorem}{Theorem}{Theorems}
\crefname{proposition}{Proposition}{Propositions}
\crefname{lemma}{Lemma}{Lemmas}
\crefname{corollary}{Corollary}{Corollaries}
\crefname{remark}{Remark}{Remarks}
\crefname{section}{Section}{Sections}
\crefname{subsection}{Section}{Sections}
\crefname{equation}{Equation}{Equations}

\begin{document}

\begin{frontmatter}

\title{Finite-valuation approximable structures: a solution to the Jung--Tix problem of probabilistic powerdomains\tnoteref{t1}}
\tnotetext[t1]{Research supported by NSF of China (Nos. 12471439, 12231007).}

\author[addr1]{Yuxu Chen}
\address[addr1]{School of Mathematics, Sichuan University, Chengdu 610065, China}
\ead{chenyuxu@scu.edu.cn}

\author[addr1]{Hui Kou}
\ead{kouhui@scu.edu.cn}

\author[addr1]{Zhenchao Lyu}
\ead{zhenchaolyu@scu.edu.cn}

\begin{abstract}
We introduce the category \(\FVA\) of finite-valuation approximable domains, a full subcategory of continuous domains contained in the category of pointed countably based FS-domains. We prove that \(\FVA\) is Cartesian closed and closed under both the subprobabilistic and probabilistic valuation powerdomains. Hence the valuation monads \(\Vsub\) and \(\Vone\) restrict to \(\FVA\), yielding a positive answer to the category-existence form of the Jung--Tix problem, a long-standing open problem in domain theory since the 1990s. 
In particular, we develop a new   factorization-approximation framework for constructing objects from a given class of known
objects or structures. Applying this method to the class of subprobabilistic powerdomains over finite posets, we construct the class $\FVA$ and show that it is closed under Scott-continuous retracts, finite products, function spaces, \(\Vsub\) and \(\Vone\) monads.

\end{abstract}


\end{frontmatter}

\clearpage
\begingroup
\setlength{\parskip}{0pt}
\tableofcontents
\endgroup
\clearpage

\section{Introduction}
\subsection{Domain theory and finite approximation}
\label{subsec:intro-finite-approximation}

Domain theory originated in Dana Scott's order-theoretic approach \cite{Scott1970,Scott1972,Scott1993} to computation 
and in the Scott--Strachey programme for denotational semantics
\cite{ScottStrachey1971}.  Its basic idea is to order partial objects by information content: $x\leq y$ means that $y$ contains at least the information present in $x$. Directed suprema describe limits of compatible approximations, Scott-continuous maps preserve those limits, and least fixed points interpret recursive definitions. Suitable
Cartesian closed categories of domains then provide interpretations of
higher-order function types. Domain theory thereby connects order, topology, fixed-point theory, and the semantics of programming languages within a single mathematical setting 
\cite{AbramskyJung1994,AmadioCurien1998,GierzEtAl2003,
Goubault2013,Gunter1992,Winskel1993,Zhang1991}.

Finite order structure has played an organizing role from the beginning.
Scott's universal domain
$\mathcal P\omega$ represents data types by retracts of an algebraic
domain whose compact elements are finite~\cite{Scott1976}. 
Plotkin's universal domain $\mathbb T^\omega$ and the
embedding--projection method of Smyth and Plotkin similarly reconstruct infinite domains and solutions of recursive domain equations from controlled approximation stages
\cite{Plotkin1978,SmythPlotkin1982}. A fundamental
requirement for interpreting higher-order types is Cartesian closure.
Jung's systematic study of
Cartesian closed categories of domains placed bifinite domains, their Scott-continuous retracts, the RB-domains, and the broader class of FS-domains at the centre of this finite-approximation programme~\cite{Jung1989,AbramskyJung1994}. 
In particular, pointed FS-domains form one of the  Cartesian closed full subcategories of the category of pointed domains with
Scott-continuous maps as morphisms. The structure of these classes and their order-topological foundations were subsequently developed further
by Lawson, Scott, and their collaborators~\cite{GierzEtAl2003}.

\subsection{Probabilistic powerdomains, Jung--Tix problem and related work}
\label{subsec:intro-jung-tix}

The introduction of probabilistic computation revealed a persistent
obstruction to this programme. The intellectual origins of probabilistic choice and its semantic treatment in domain theory trace back to the foundational contributions of Saheb-Djahromi \cite{Saheb-Djahromi1980}. Later Jones and Plotkin introduced the probabilistic powerdomain in order to model probabilistic choice \cite{JonesPlotkin1989,Jones1990}.

For a dcpo $D$, let $\Vsub(D)$ and $\Vone(D)$ denote, respectively,
the dcpos of continuous subprobabilistic and probabilistic valuations on $D$,
ordered pointwise on Scott-open sets. Continuous valuations provide the
standard domain-theoretic representation of probabilistic choice.
Although the probabilistic powerdomain construction preserves continuity \cite{Jones1990},
the category of all continuous domains and Scott-continuous functions is not Cartesian closed \cite{Jung1989}.
On the other hand, among the familiar full Cartesian closed subcategories with strong finite-approximation properties, the classes of RB-domains and FS-domains stand out as two particularly appropriate candidates for semantic applications.
We have shown that RB-domains are not closed under probabilistic powerdomain construction recently \cite{ChenKouLyu2026}.
It remains unknown whether the subprobabilistic powerdomain $\Vsub$
preserves FS-domains. 

Jung and Tix proved that $\Vsub(P)$ is an RB-domain for any finite rooted tree $P$ and an FS-domain for any finite reversed rooted tree, but their result does not extend to arbitrary finite posets beyond these two classes~\cite{JungTix1998}.
The resulting Jung--Tix compatibility problem asks
whether there exists a suitably well-behaved class of continuous dcpos
that is both Cartesian closed and closed under probabilistic powerdomains. As Jung and Tix emphasized, however, their results
illustrate the difficulty of the problem rather than provide a satisfactory general solution.

Considerable progress in higher-order probabilistic semantics has
nevertheless been achieved by changing the semantic setting rather than
settling this compatibility problem in its original form. Probabilistic
coherence spaces provide a fully abstract model of probabilistic
PCF~\cite{EhrhardPaganiTasson2018}. Quasi-Borel predomains support adequate
semantics for languages with higher-order types, recursion, and continuous
probability distributions~\cite{VakarKammarStaton2019}, while a direct
domain-theoretic semantics for statistical programming languages over
dcpos was developed in~\cite{GoubaultLarrecqJiaTheron2023}. Commutative
probabilistic monads on dcpos yield a sound and adequate semantics for the
Probabilistic FixPoint Calculus%
~\cite{JiaLindenhoviusMisloveZamdzhiev2021}, and such a monad supplies the
classical probabilistic component in the semantics of variational quantum
programming~\cite{JiaEtAl2022}. In another direction, Goubault-Larrecq used
the call-by-push-value paradigm to separate the two semantic requirements
between distinct classes of types: probabilistic choice resides at value
types, while the passage from value types to computation types incorporates
demonic nondeterminism~\cite{Goubault2019}. Di Gianantonio and Edalat instead
work directly with random variables over standard probability spaces,
obtaining probabilistic semantics in an enriched category of Scott
domains~\cite{DiGianantonioEdalat2024}.

Although these alternative approaches have yielded several interesting semantic
models for higher-order probabilistic languages, they do so by modifying
the ambient semantic setting, separating the relevant type disciplines, or
replacing valuations with random variables. 
Consequently, these approaches do not resolve the original Jung--Tix
problem, which calls for a solution within the standard domain-theoretic
setting of continuous dcpos and Scott-continuous maps. A solution in this setting would provide the conceptually
simplest and mathematically most elegant domain-theoretic reconciliation
of higher-order function spaces with probabilistic choice.

Several important advances have clarified why the original problem is difficult. Goubault-Larrecq introduced $\omega$-QRB-domains and proved that they are preserved by the probabilistic powerdomain, finite products, retracts, and expanding bilimits, but the resulting category is not Cartesian closed~\cite{GoubaultLarrecq2012}. Goubault-Larrecq and Jung subsequently proved that QRB-domains coincide with QFS-domains and with Lawson-compact quasicontinuous dcpos, and established the corresponding
probabilistic closure theorem without the earlier countability and pointedness restrictions
\cite{GoubaultLarrecqJung2014}. Lyu and Kou also studied the probabilistic powerdomain from a topological viewpoint \cite{LK2018}. 
Passing from continuous to quasicontinuous domains does not remove this obstruction within the full-subcategory setting.  Jia, Jung, Kou, Li, and Zhao proved that every full Cartesian closed subcategory of the category of quasicontinuous domains and Scott-continuous maps consists entirely of continuous domains~\cite{JiaJungKouLiZhao2015}. Thus a full Cartesian closed solution cannot be obtained merely by enlarging the object class from continuous to quasicontinuous domains.

The obstruction is already visible on finite posets. Even for finite posets, it is difficult to check if their probabilistic powerdomains are FS-domains.  For nearly three decades, the Jung--Tix problem has become one of the central and technically most difficult open problems in domain theory. Its resolution is fundamental to the development of a satisfactory domain-theoretic foundation for higher-order probabilistic  denotational semantics.

\subsection{Main contribution and proof roadmap}
\label{subsec:intro-roadmap}

In this paper, we give a positive answer to the Jung--Tix problem by introducing $\FVA$, a Cartesian closed full subcategory consisting of continuous domains that is closed under probabilistic powerdomains.

The proof strategy is divided into two stages. First, we treat finite posets. We prove that the subprobabilistic powerdomain of every finite poset is an FS-domain. 
The second stage is to use the finite valuation spaces themselves as building blocks for general domains. A domain is called
\emph{finite-valuation approximable} when its identity is the pointwise
supremum of an increasing sequence of maps, each factoring through
$\Vsub(P_n)$ for some finite poset $P_n$. We write $\FVA$ for the resulting
full subcategory. Here ``finite'' refers to $P_n$; the factorization object
$\Vsub(P_n)$ is usually infinite, but its order is controlled by finitely many
upper-set coordinates. Thus the definition replaces finite-image approximation by factorization through finite valuation spaces.

Our main result shows that every object of $\FVA$ is a pointed countably based
FS-domain, and that $\FVA$ contains the terminal dcpo and is closed under
Scott-continuous retracts, finite products, function spaces, and both
$\Vsub$ and $\Vone$. Consequently, $\FVA$ is a full Cartesian closed
subcategory of $\DCPO$, and the subprobabilistic and probabilistic valuation monads restrict to it. This gives a nontrivial finite-structure solution to the category-existence form of the Jung--Tix problem. 
The class of finite-valuation approximable domains contains all countably based bc-domains, while it is incomparable with the class of RB-domains.

The proof proceeds through a sequence of interfaces between domain theory and finite-dimensional analysis. Section~\ref{sec:finite-poset-powerdomains} establishes the finite-poset foundation by constructing explicit FS approximate identities on \(\Vsub(P)\). Section~\ref{sec:factorization-classes} then develops the factorization-approximation framework and specializes it to the generators \(\Vsub(P)\), reducing the main closure questions to finite-dimensional constructions involving ordered convex polytopes, Hasse-cone generators, monotone couplings, randomized grid rounding, interior contraction, and finite stochastic kernels. 
These kernels are subsequently used in two directions: Section~\ref{sec:valuation-closure} applies Choquet integration and Kleisli extension to valuation powerdomains, while Section~\ref{sec:function-spaces-ccc} uses finite sampling and barycentric reconstruction to establish the required function-space closure and Cartesian closedness. Finally, Section~\ref{sec:category-size} compares the resulting class with countably based bc-domains and RB-domains.

\section{Preliminaries}\label{sec:finite-valuations}

\subsection{Domain-theoretic background}
\label{subsec:prelim-domain}

We recall the domain-theoretic notions used below and then fix the
finite-dimensional notation for valuations. For general background on domains
and continuous valuations, see~\cite{AbramskyJung1994,GierzEtAl2003}.
Throughout, $\N=\{0,1,2,\ldots\}$, and directed sets are understood to be
nonempty.

A subset $E$ of a poset is \emph{directed} if every two elements of $E$ have an upper bound in $E$. A \emph{dcpo} is a poset in which every directed subset has a supremum. A dcpo is \emph{pointed} if it has a least element, denoted by $\bot$. Let $P$ be a dcpo. A subset $U$ of $P$ is called \emph{Scott open}, if $U$ is an upper set and $\sup D\in U$ implies $U\cap D\not=\emptyset$ for every directed subset $D$ of $P$. All Scott open subsets of $P$ form a topology, denoted by $\sigma(P)$, which is called the \emph{Scott topology} of $P$. 
 
 For elements $a,b$ of a dcpo, we write $a\waybelow b$ and say that $a$ is \emph{way below} $b$ if, whenever $b\leq\sup E$ for a directed set
$E$, there is $e\in E$ with $a\leq e$. In particular, $a$ is called compact if $a\ll a$. A dcpo $D$ is \emph{continuous} (resp. \emph{algebraic}) if,
for every $x\in D$, the set $\{a\in D:a\waybelow x\}$ (resp. $\{a\in D:a\waybelow a\leq x\}$ ) is directed and has
supremum $x$. For simplicity, a \emph{domain} means  a continuous dcpo in this paper. A subset $B\subseteq D$ is a \emph{basis} if
\(
        B_x=\{b\in B:b\waybelow x\}
\)
is directed with supremum $x$ for every $x\in D$. A  domain is
\emph{countably based} if it has a countable basis.

\begin{lemma}\label{lem:cont-equal}\cite[Exercise II-1.35]{GierzEtAl2003}
Let $D$ be a dcpo. Then the following conditions are equivalent:
\begin{enumerate}
\item[(1)] $D$ is continuous;
\item[(2)] $D$ has a basis;
\item[(3)] for every Scott-open subset $U$ of $D$ and every $x\in U$, there is $a\in D$ such that
\[x\in\operatorname{int}_{\sigma(D)}(\uparrow a)
  \subseteq {\uparrow a} \subseteq U.\]
\end{enumerate}
\end{lemma}

\begin{lemma}\label{lem:cont-property}
A continuous dcpo $E$ has the following properties:
\begin{enumerate}
\item[(1)] \cite[Theorem I-1.9(ii)]{GierzEtAl2003} Interpolation: if $x\ll y$, then $x\ll z\ll y$ for some $z\in E$.
\item[(2)] \cite[Proposition II-1.10(ii)]{GierzEtAl2003} For $x\in E$, set
\(
  \mathord{\Uparrow}x=\{a\in E:x\ll a\}.
\)
Then $\mathord{\Uparrow}x$ is Scott open, and $\{\mathord{\Uparrow} x: \ x\in E\}$ forms a basis of the Scott topology.
\item[(3)] \cite[Theorem III-4.5]{GierzEtAl2003} $E$ is countably based if and only if its Scott topology has a countable basis.
\end{enumerate}
\end{lemma}

 A map between dcpos is
\emph{Scott-continuous} if and only if it is monotone and preserves directed
suprema, i.e., it is continuous respect to the Scott topologies.  For dcpos $D$ and $E$, we write $[D\to E]$ for the dcpo of
Scott-continuous maps from $D$ to $E$, ordered pointwise:
\[
  f\leq g
  \quad\Longleftrightarrow\quad
  f(x)\leq g(x)\text{ for every }x\in D,
\]
called the function space. 
Directed suprema in $[D\to E]$ are computed pointwise. Thus, for a
directed family $(f_i)$,
\[
  \left(\sup_i f_i\right)(x)=\sup_i f_i(x).
\]
This observation will be used repeatedly when a construction on valuations
is lifted pointwise to a function space. We write $\DCPO$ for the category
of dcpos and Scott-continuous maps.

\begin{lemma}\cite[Lemma 1.21]{Jung1989}\label{lem:dcpo-Cartesian}
Let ${\bf A}$ be a full subcategory of $\DCPO$. Then ${\bf A}$ is Cartesian closed iff its terminal object, binary products, and exponentials all exist and inherited from $\DCPO$, that is, the terminal object is order-isomorphic to the singleton dcpo, the binary products of any two objects $A_1,A_2\in{\bf A}$ is order-isomorphic to the cartesian product $A_1\times A_2$ and the exponential object $A_2^{A_1}$ is order-isomorphic to $[A_1\to A_2]$.
\end{lemma}

In general, the function space of two continuous dcpos need not be continuous~\cite{Jung1989}. Therefore, the category of all continuous dcpos is not Cartesian closed.

\begin{definition}\label{def:FS}
\cite[Definition~II-2.15]{GierzEtAl2003}.  Let $D$ be a dcpo. 
 \begin{enumerate}
     \item[(1)] 
 A Scott-continuous map $f:D\to D$ is
\emph{finitely separated from the identity} if there is a finite subset
$M\subseteq D$ such that, for every $x\in D$, some $m\in M$ satisfies
\(
        f(x)\leq m\leq x.
\)
The set $M$ is called a \emph{finite separator} for $f$.
 \item[(2)] 
A family $(f_i)_{i\in I}$ of Scott-continuous self-maps on $D$, directed in the
pointwise order, is called an \emph{FS approximate identity} if every $f_i$ is
finitely separated from the identity and $\sup_i f_i=\id_D$ pointwise. A
dcpo admitting an FS approximate identity is called an \emph{FS-domain}. 
We write $\FS$ for the full subcategory of $\DCPO$ consisting of all 
FS-domains.
\end{enumerate}
\end{definition}

\begin{lemma}\cite[Lemma II-2.16]{GierzEtAl2003}\label{lem:FS-properties}
Let $D$ be a dcpo.
\begin{enumerate}
\item[(1)] If a Scott-continuous map $f:D\to D$ is finitely separated from the identity, then $f(x)\ll x$ for every $x\in D$.
\item[(2)] Every FS-domain is continuous.
\end{enumerate}
\end{lemma}

The following important result is also due to Jung \cite{Jung1989,AbramskyJung1994}.

\begin{theorem}\label{thm:FS-cate}\cite[Theorem 4.2.11]{AbramskyJung1994}
The full subcategory of pointed FS-domains and Scott-continuous maps is Cartesian closed.
\end{theorem}

In the classical setting of the category of pointed continuous dcpos and Scott-continuous maps, the two maximal Cartesian closed full subcategories are the FS-domains and the continuous L-domains \cite{Jung1989,Jung1990}. However, the valuation powerdomain of a L-domain need not be an L-domain in general. Thus, among the two classical maximal Cartesian closed classes of domains, FS-domains provide the natural setting in which to seek a positive solution to the Jung--Tix problem.

\subsection[Continuous valuations and Kleisli extension]{Continuous valuations, integration, and Kleisli extension}
\label{subsec:prelim-valuations}

For a real number $r$, write $r^+=\max\{r,0\}$.

\begin{definition}\cite{JonesPlotkin1989}
Let $X$ be a topological space and let $\mathcal O(X)$ be its lattice of open sets. A \emph{continuous valuation} on $X$ is a Scott-continuous map
$\nu:\mathcal O(X)\to[0,\infty]$ satisfying
\[
  \nu(\varnothing)=0,
  \qquad
  \nu(U)+\nu(V)=\nu(U\cup V)+\nu(U\cap V).
\]
It is a
\emph{subprobabilistic valuation} if $\nu(X)\leq1$, and a 
\emph{probabilistic valuation} if $\nu(X)=1$. Here, $[0,\infty]$ is endowed with the Scott topology.

These valuations are ordered pointwise and form dcpos denoted by $\mathcal V(X)$, $\Vsub(X)$, and $\Vone(X)$, respectively. When $X$ is a dcpo endowed with its Scott topology, we call $\Vsub(X)$ the \emph{subprobabilistic powerdomain} over $X$ and $\Vone(X)$ the \emph{probabilistic powerdomain} over $X$.

\end{definition}

Here, $\nu(U)$ is interpreted as the mass assigned to the observable event
$U$, and the displayed modularity equation is the finite-additivity law in
its form appropriate to open sets. Scott continuity of $\nu$ means that the
mass of a directed union of open sets is the supremum of their masses.

Continuous valuations provide an order-theoretic interpretation of probabilistic choice that is compatible with recursion and higher-order functions. They have been used in denotational semantics for probabilistic PCF, statistical probabilistic programming with continuous distributions, and mixed classical--quantum programming languages \cite{JonesPlotkin1989,Jones1990,EhrhardPaganiTasson2018,VakarKammarStaton2019,GoubaultLarrecqJiaTheron2023,JiaLindenhoviusMisloveZamdzhiev2021,JiaEtAl2022}. 
In particular, Goubault-Larrecq, Jia, and Th\'eron use domain-theoretic valuation methods to give semantics to a statistical language with higher-order functions, continuous distributions, conditioning, and recursion \cite{GoubaultLarrecqJiaTheron2023}. 

In this paper, all valuation powerdomains are taken over dcpos endowed
with their Scott topologies. If $D$ is continuous, then $\Vsub(D)$ is
continuous \cite[Corollary 5.4]{Jones1990}; if $D$ is countably based, then so is $\Vsub(D)$
\cite[Corollary 5.5]{Jones1990}. For pointed $D$,
\cref{lem:valuation-retract} will exhibit $\Vone(D)$ as a
Scott-continuous retract of $\Vsub(D)$.

For a Scott-continuous map $f:D\to E$ between two dcpos $D$ and $E$, define its pushforward 
\[
  \Vsub(f):\Vsub(D)\longrightarrow\Vsub(E),
  \qquad
  \Vsub(f)(\nu)(U)=\nu\bigl(f^{-1}(U)\bigr)
  \quad(U\in\sigma(E)).
\]
The unit is $\eta_D(x)=\delta_x$, where $\delta_x$ is the Dirac valuation concentrated at $x$. We use the same symbol $\eta_D$ for its corestriction to $\Vone(D)$ when the codomain is clear.

The following is the local continuity of the continuous-valuation monad.  

\begin{lemma}\label{lem:V-local}
Let \((f_i)_{i\in I}\) be a directed family of Scott-continuous maps
\[
  f_i:D\longrightarrow E,
\]
and let \(f=\sup_i f_i\) pointwise.  Then
\[
  \Vsub(f)=\sup_i\Vsub(f_i)
\]
pointwise in \([\Vsub(D)\to\Vsub(E)]\).
In particular, if \(D=E\), \(f_i\leq\id_D\) for every \(i\), and
\(\sup_i f_i=\id_D\), then
\[
  \Vsub(f_i)\leq\id_{\Vsub(D)}
  \qquad\text{and}\qquad
  \sup_i\Vsub(f_i)=\id_{\Vsub(D)}.
\]
\end{lemma}

\begin{proof}
Fix \(\nu\in\Vsub(D)\) and a Scott-open set \(U\subseteq E\).  Since the
family \((f_i)_i\) is directed, the family
\((f_i^{-1}(U))_i\) is directed under inclusion.  Moreover,
\[
  f^{-1}(U)=\bigcup_i f_i^{-1}(U).
\]
Indeed, the inclusion from right to left follows from \(f_i\leq f\) and the
upperness of \(U\).  Conversely, if \(f(x)=\sup_i f_i(x)\in U\), Scott
openness gives \(f_i(x)\in U\) for some \(i\).  Hence the Scott continuity
of \(\nu\) yields
\[
\begin{aligned}
  \bigl(\Vsub(f)(\nu)\bigr)(U)
  &=
  \nu\bigl(f^{-1}(U)\bigr)\\
  &=
  \sup_i\nu\bigl(f_i^{-1}(U)\bigr)\\
  &=
  \sup_i\bigl(\Vsub(f_i)(\nu)\bigr)(U).
\end{aligned}
\]
Directed suprema of continuous valuations are computed pointwise on
Scott-open sets~\cite{JonesPlotkin1989,Jones1990}, so the first assertion
follows.

Now assume \(f_i\leq\id_D\).  If \(f_i(x)\in U\), then \(x\in U\) because
\(U\) is upper.  Thus \(f_i^{-1}(U)\subseteq U\), and hence
\[
  \Vsub(f_i)(\nu)(U)
  =
  \nu\bigl(f_i^{-1}(U)\bigr)
  \leq
  \nu(U).
\]
Therefore \(\Vsub(f_i)\leq\id_{\Vsub(D)}\).  The last assertion follows by
applying the first part to \(\sup_i f_i=\id_D\).
\end{proof}


Equip \([0,1]\) with its usual order and Scott topology. Let
\(\nu\in\Vsub(D)\), and let
\(
    h:D\longrightarrow [0,1]
\)
be Scott-continuous. We use the standard integration with respect to
continuous valuations, defined by the layer-cake formula
\[
  \int_D h\,d\nu
  =
  \int_0^1
  \nu\bigl(h^{-1}((t,1])\bigr)\,dt .
\]
The integral on the right-hand side is the ordinary Riemann integral.
More generally, for a lower semicontinuous map
\(
    h:D\longrightarrow\overline{\mathbb R}_{+},
\)
we define
\[
  \int_D h\,d\nu
  =
  \sup_{r>0}
  \int_0^r
  \nu\bigl(h^{-1}((t,\infty])\bigr)\,dt .
\]
This is the usual \emph{Choquet integral} associated with continuous
valuations~\cite{Tix1995,GoubaultLarrecq2022}.
For $[0,1]$-valued Scott-continuous functions, it agrees with
the simple-function definition of integration used by
Jones and Plotkin in \cite{JonesPlotkin1989}.

The following theorem is the directed form of monotone convergence for integration with respect to continuous valuations, see \cite[Theorem~3.1]{JonesPlotkin1989} or \cite[Satz 4.4]{Tix1995}.

\begin{theorem}\label{thm:jones-monotone-convergence}
Let $D$ be a dcpo, let $\nu\in\Vsub(D)$, and let $(h_i)_{i\in I}$ be a directed family of Scott-continuous maps $h_i:D\to[0,1]$, ordered pointwise.
Then
\[
  \int_D\sup_{i\in I}h_i\,d\nu
  =
  \sup_{i\in I}\int_Dh_i\,d\nu.
\]
\end{theorem}

In particular, for every Scott-open set \(U\subseteq D\),
\(
    \int_D\chi_U\,d\nu=\nu(U),
\)
and hence
\(
    \int_D1\,d\nu=\nu(D).
\)
Now let
\(
    k:D\longrightarrow\Vsub(E)
\)
be Scott-continuous. Such a map is called a \emph{valuation kernel}.
For every Scott-open set \(U\subseteq E\), the evaluation map
\[
    x\longmapsto k(x)(U)
\]
is Scott-continuous from \(D\) to \([0,1]\). The \emph{Kleisli extension} of
\(k\) is the map
\[
k^\dagger:\Vsub(D)\longrightarrow\Vsub(E)
\]
defined by
\begin{equation}\label{eq:kleisli-definition}
    (k^\dagger\nu)(U)
    =
    \int_D k(x)(U)\,d\nu(x).
\end{equation} 

The following lemma can be found in \cite[Lemma 3.1, Propositionn 3.2]{GoubaultLarrecqJiaTheron2023}.

\begin{lemma}
\label{lem:valuation-kleisli-laws}
Let $C,D,E,F$ be dcpos, and let
\[
k:D\longrightarrow\Vsub(E),
  \qquad
  \ell:E\longrightarrow\Vsub(F)
\]
be Scott-continuous kernels. The following properties hold.
\begin{enumerate}[label=\textup{(\roman*)}]
\item The map $k^\dagger:\Vsub(D)\to\Vsub(E)$ defined by
\eqref{eq:kleisli-definition} is Scott-continuous.
~Kleisli extension is also monotone in the kernel: if $k\leq k'$,
then $k^\dagger\leq(k')^\dagger$.
\item The two unit laws are
\begin{equation}\label{eq:kleisli-unit-laws}
  \eta_D^\dagger=\id_{\Vsub(D)},
  \qquad
  k^\dagger\circ\eta_D=k.
\end{equation}
\item Kleisli extension satisfies the associativity law
\begin{equation}\label{eq:kleisli-associativity}
  (\ell^\dagger\circ k)^\dagger
  =\ell^\dagger\circ k^\dagger.
\end{equation}
\item For every Scott-continuous map $u:D\to E$,
\begin{equation}\label{eq:pushforward-kleisli}
  \Vsub(u)=(\eta_E\circ u)^\dagger.
\end{equation}
\end{enumerate}

Consequently, for Scott-continuous maps $v:C\to D$ and $w:E\to F$,
\begin{equation}\label{eq:kleisli-pushforward-composition}
\begin{aligned}
  (k\circ v)^\dagger&=(k^\dagger\circ\eta_D\circ v)^\dagger=k^\dagger\circ(\eta_D\circ v)^\dagger=k^\dagger\circ\Vsub(v),\\
  (\Vsub(w)\circ k)^\dagger&=((\eta_F\circ w)^\dagger\circ k)^\dagger=(\eta_F\circ w)^\dagger\circ k^\dagger=\Vsub(w)\circ k^\dagger.
\end{aligned}
\end{equation}
\end{lemma}

The corresponding multiplication of the valuation monad is obtained by
taking the barycentre of a valuation of valuations,
\[
    \mu_D=(\id_{\Vsub(D)})^\dagger:
    \Vsub(\Vsub(D))\longrightarrow\Vsub(D).
\]
The same constructions restrict to probabilistic valuations. We use the
canonical Scott-continuous inclusion
\[
    \Vone(D)\hookrightarrow\Vsub(D).
\]
Indeed, if \(f:D\to E\) and \(\nu\in\Vone(D)\), then
\[
    \Vone(f)(\nu)(E)
    =
    \nu(f^{-1}(E))
    =
    \nu(D)
    =
    1.
\]
Moreover, if
\(
    k:D\to\Vone(E)
\)
and \(\nu\in\Vone(D)\), then
\[
    (k^\dagger\nu)(E)
    =
    \int_D k(x)(E)\,d\nu(x)
    =
    \int_D1\,d\nu
    =
    1.
\]
Hence the Kleisli extension restricts to
\(
\Vone(D)\longrightarrow\Vone(E).
\)

\subsection[Finite-poset valuation domains]{Finite-poset valuation domains and stochastic order}
\label{subsec:prelim-finite-posets}

Let $P$ be a finite poset. A subset $U\subseteq P$ is called an \emph{upper set} if
$x\in U$ and $x\leq y$ imply $y\in U$. We use
\[
        \up x=\{y\in P:x\leq y\},
        \qquad
        \down A=\{x\in P:x\leq a\text{ for some }a\in A\}.
\]
Every directed subset of a finite poset has a greatest element, so the
Scott-open subsets of $P$ are exactly its upper sets.
Scott continuity on $\sigma(P)$ is automatic for every
monotone map.  For $x\in P$, the Dirac valuation is
$\delta_x(U)=1$ if $x\in U$, and $\delta_x(U)=0$ otherwise.

On a finite poset every subprobabilistic valuation is simple. More precisely,
there are unique coefficients $p_x\geq0$ with $\sum_x p_x\leq1$ such that
\begin{equation}\label{eq:atomic-representation}
        \nu=\sum_{x\in P}p_x\delta_x,
        \qquad
        p_x=\nu(\up x)-\nu(\up x\setminus\{x\}).
\end{equation}
This is the standard finite-space description of the valuation powerdomain;
see~\cite[Page 5]{JungTix1998}. We therefore identify $\Vsub(P)$ with the
Euclidean simplex
\[
        \Delta_{\leq1}(P)
        =\left\{p\in\Rplus^P:\sum_{x\in P}p_x\leq1\right\}.
\]
For $A\subseteq P$, write $p(A)=\sum_{x\in A}p_x$. The stochastic order on
valuations is
\begin{equation}\label{eq:upper-set-order}
        p\stle q
        \quad\Longleftrightarrow\quad
        p(U)\leq q(U)
        \quad\text{for every upper set }U\subseteq P.
\end{equation}
Thus $q$ is above $p$ when every upward-closed observation receives at least
as much mass under $q$ as under $p$. On a finite poset this order permits
both adding mass and moving existing mass upward; it is generally different
from the coordinatewise order on the atomic coefficients.
The positive support of $p$ is
\[
        \supp(p)=\{x\in P:p_x>0\}.
\]
The standard characterizations of stochastic order on partially ordered
spaces~\cite{KamaeKrengelOBrien1977} give the following finite form.  The
subprobability version is obtained from the probability version by adjoining a
fresh least point and placing the missing mass there.

\begin{lemma}\label{lem:monotone-integrals}
For $p,q\in\Vsub(P)$, the following conditions are equivalent:
\begin{enumerate}[label=\textup{(\roman*)}]
\item $p\stle q$;
\item for every nonnegative monotone map $g:P\to\R$,
\[
  \sum_{x\in P}p_xg(x)\leq\sum_{x\in P}q_xg(x).
\]
\end{enumerate}
\end{lemma}
\begin{proof}
The implication \textup{(ii)}$\Rightarrow$\textup{(i)} follows by taking
$g=\mathbf 1_U$ for every upper set $U$. Conversely, let
$0<t_1<\cdots<t_m$ be the distinct positive values of a nonnegative
monotone map $g:P\to\R$, and put
\[
  U_j=\{x\in P:g(x)\geq t_j\}.
\]
Each $U_j$ is an upper set. With $t_0=0$,
\[
  g=\sum_{j=1}^m(t_j-t_{j-1})\mathbf 1_{U_j}.
\]
If $p\stle q$, applying the upper-set inequalities to this nonnegative
linear combination gives
\[
  \sum_{x\in P}p_xg(x)\leq\sum_{x\in P}q_xg(x).
\]
\end{proof}

A finite poset is a continuous dcpo with a finite basis. Jones's results
therefore imply that $\Vsub(P)$ is a countably based domain; a countable basis
is obtained from simple valuations with rational
coefficients~\cite[Corollaries~5.4 and~5.5]{Jones1990}.

Directed suprema in the valuation powerdomain are computed pointwise on
Scott-open sets~\cite{JonesPlotkin1989,Jones1990}. Thus, if
$D\subseteq\Delta_{\leq1}(P)$ is directed and $p=\sup D$, then
\begin{equation}\label{eq:directed-supremum}
        p(U)=\sup_{q\in D}q(U)
        \qquad\text{for every upper set }U\subseteq P.
\end{equation}

\section{Subprobabilistic powerdomains of finite posets}
\label{sec:finite-poset-powerdomains}

\subsection{Construction of the approximating maps}\label{sec:construction}

This section constructs the basic approximation used throughout the paper.
For a valuation $p$ on a finite poset, we repeatedly remove mass from the
maximal elements of its positive support.  The resulting one-parameter family
$(\Phi_t)_{t\geq0}$ moves every valuation downward and converges back to it
as $t\to 0$. The main issue, addressed in the next subsection, is to choose the scale so that the maps also preserve the stochastic
order.

We first introduce the frontiers, erosion rates, and the recursive flow used in the construction.
\label{subsec:construction-flow}
Throughout this section, \(P\) is nonempty.  Put
\[
  n=|P|,
  \qquad
  K_P=n(n+1)^{n-1}.
\]
For $p\in\Delta_{\leq1}(P)$ with $p\neq0$, put
\[
        A(p)=\Max(\supp(p))
\]
to be the set of maximal elements of $\supp(p)$, which we also called the \textit{frontier}. A subset of a poset is an
\emph{antichain} if no two distinct elements are comparable. Thus $A(p)$ is
a nonempty antichain. For every nonempty antichain $A\subseteq P$, define
\begin{equation}\label{eq:coefficient}
        c(A)=(n+1)^{\,n-|\down A|}.
\end{equation}
The exact formula is chosen to enforce the comparison in
\cref{lem:coefficient-comparison}: if one frontier lies strictly below
another, then the lower frontier is eroded at a substantially larger rate.
This separation of rates is what later prevents two ordered trajectories
from crossing.

\begin{lemma}\label{lem:coefficient-comparison}
Let $A$ and $B$ be nonempty antichains of $P$. If
$A\subseteq\down B$ and $A\neq B$, then
\[
        c(A)\geq(n+1)c(B)>n c(B).
\]
\end{lemma}

\begin{proof}
For every antichain $E$, the maximal elements of $\down E$ are exactly the
points of $E$. Hence $A\subseteq\down B$ implies
$\down A\subseteq\down B$. Equality of these two lower sets would imply
$A=B$, so the inclusion is strict. Therefore
$|\down A|\leq|\down B|-1$, and~\eqref{eq:coefficient} gives
\[
        \frac{c(A)}{c(B)}
        =(n+1)^{|\down B|-|\down A|}
        \geq n+1.
\]
\end{proof}

Fix $p\in\Delta_{\leq1}(P)$. We define
\[
        \phi_p:[0,\infty)\longrightarrow\Delta_{\leq1}(P)
\]
recursively. During one stage, all coordinates on the current frontier are
decreased at the common rate $c(A)$, while all other coordinates are kept
fixed. The stage ends when the first frontier coordinate reaches zero; the
frontier is then recomputed from the smaller support.

Set $q_0=p$ and $s_0=0$. Suppose that $q_j$ and $s_j$ have been
defined. If $q_j=0$, define $\phi_p(s)=0$ for all $s\geq s_j$ and terminate
the recursion. If $q_j\neq0$, set
\[
        A_j=A(q_j),
        \qquad
        c_j=c(A_j),
        \qquad
        \tau_j=\min_{a\in A_j}\frac{(q_j)_a}{c_j}.
\]
The set $A_j$ is finite and nonempty, and $(q_j)_a>0$ for every $a\in A_j$,
hence $\tau_j>0$. For $0\leq u\leq\tau_j$, define
\begin{equation}\label{eq:recursive-interval}
        \phi_p(s_j+u)
        =q_j-u c_j\sum_{a\in A_j}\delta_a.
\end{equation}
Equivalently,
\[
        \bigl(\phi_p(s_j+u)\bigr)_x
        =
        \begin{cases}
        (q_j)_x-u c_j,&x\in A_j,\\
        (q_j)_x,&x\notin A_j.
        \end{cases}
\]
The definition of $\tau_j$ implies that every coordinate remains
nonnegative. The total mass does not increase, so
$\phi_p(s_j+u)\in\Delta_{\leq1}(P)$ for $0\leq u\leq\tau_j$. Define
\[
        s_{j+1}=s_j+\tau_j,
        \qquad
        q_{j+1}=\phi_p(s_{j+1}).
\]
At least one coordinate indexed by $A_j$ is zero in $q_{j+1}$. Coordinates outside \(A_j\) are unchanged, while the remaining
coordinates indexed by \(A_j\) stay nonnegative. Consequently,
\[
        \supp(q_{j+1})\subsetneq\supp(q_j).
\]
Thus the recursion has at most $|\supp(p)|\leq n$ nonzero steps and eventually
reaches the zero vector.

Formula~\eqref{eq:recursive-interval} shows that $\phi_p$ is continuous on
each interval $[s_j,s_{j+1}]$. At the common endpoint of two consecutive
intervals, both definitions have value $q_{j+1}$. If $q_m=0$, the extension
$\phi_p(s)=0$ for $s\geq s_m$ also agrees at $s_m$. Hence $\phi_p$ is
continuous on $[0,\infty)$. For $t\geq0$, define
\begin{equation}\label{eq:def_of_phi}
        \Phi_t:\Delta_{\leq1}(P)\longrightarrow\Delta_{\leq1}(P),
        \qquad
        \Phi_t(p)=\phi_p(t).
\end{equation}

\begin{remark}
To illustrate the definition of the above function, we give an explicit computation on a finite poset $B_5$, including its Hasse diagram in Example~\ref{ex:erosion-B5} in \ref{app:erosion-B5}.
\end{remark}

\vskip 3mm

We now establish the semigroup law and monotonicity in the time parameter.
\label{subsec:construction-semigroup}
We first show some basic properties of this map.
\begin{lemma}\label{lem:basic-family}
For every $p\in\Delta_{\leq1}(P)$ and $s,t\geq0$, the following hold.
\begin{enumerate}[label=\textup{(\roman*)}]
\item $\Phi_t(p)\leq p$ coordinatewise, and hence $\Phi_t(p)\stle p$;
\item $\Phi_0=\id$ and $\Phi_{s+t}=\Phi_s\circ\Phi_t$;
\item if $s\geq t$, then $\Phi_s(p)\leq\Phi_t(p)$ coordinatewise.
\end{enumerate}
\end{lemma}

\begin{proof}
On each stage interval \([s_j,s_{j+1}]\), we have
\[
    \phi_p(s_j+u)
    =
    q_j-u c_j\sum_{a\in A_j}\delta_a
    \qquad
    (0\leq u\leq\tau_j).
\]
Hence, for every \(x\in P\),
\[
    \bigl(\phi_p(s_j+u)\bigr)_x
    =
    \begin{cases}
        (q_j)_x-u c_j, & x\in A_j,\\[1mm]
        (q_j)_x,       & x\notin A_j.
    \end{cases}
\]
Thus no coordinate increases during a stage. Since the endpoint of one stage
is the initial state of the next, it follows inductively that
\(
    \phi_p(t)\leq p
\)
coordinatewise for every \(t\geq0\). Therefore
\[
    \Phi_t(p)=\phi_p(t)\leq p
\]
coordinatewise.
Hence \(\Phi_t(p)\stle p\). This proves \textup{(i)}.

The equality $\Phi_0=\id$ follows from $\phi_p(0)=p$. We prove the semigroup
identity. Fix $t\geq0$ and put
\[
        r=\Phi_t(p)=\phi_p(t).
\]
Here $\phi_r$ denotes the map obtained from the same recursion with initial
value $r$. We claim that
\begin{equation}\label{eq:translated-tail}
        \phi_r(u)=\phi_p(t+u)
        \qquad(u\geq0).
\end{equation}
If $r=0$, both sides are zero. Assume that \(r\neq0\).  Since $r=\phi_p(t)$, there is a unique $j$ such that
$s_j\leq t<s_{j+1}$. Write $t=s_j+v$, where $0\leq v<\tau_j$. Then
\[
        r=q_j-vc_j\sum_{a\in A_j}\delta_a.
\]
No coordinate has become zero between $s_j$ and $t$, so
$\supp(r)=\supp(q_j)$ and $A(r)=A_j$. The coefficient in the first recursive
step starting from $r$ is therefore $c_j$, and the length of that step is
\[
\begin{aligned}
        \min_{a\in A_j}\frac{r_a}{c_j}
        =\min_{a\in A_j}
          \left(\frac{(q_j)_a}{c_j}-v\right)
        =\tau_j-v.
\end{aligned}
\]
Hence, for $0\leq u\leq\tau_j-v$,
\[
\begin{aligned}
        \phi_r(u)
        =r-u c_j\sum_{a\in A_j}\delta_a
        =q_j-(v+u)c_j\sum_{a\in A_j}\delta_a
        =\phi_p(t+u).
\end{aligned}
\]
At $u=\tau_j-v$, both sides equal $q_{j+1}$. If $q_{j+1}=0$, both maps
remain zero. If $q_{j+1}\neq0$, both recursions compute the same set
$A(q_{j+1})$, the same coefficient $c(A(q_{j+1}))$, and the same number
\[
        \min_{a\in A(q_{j+1})}
        \frac{(q_{j+1})_a}{c(A(q_{j+1}))}.
\]
They therefore agree on the next recursive interval and again have the same
endpoint. Repeating this argument over the finitely many remaining intervals
proves
\eqref{eq:translated-tail}. Consequently,
\[
        \Phi_s(\Phi_t(p))
        =\phi_r(s)
        =\phi_p(t+s)
        =\Phi_{s+t}(p),
\]
which proves~\textup{(ii)}.

If $s\geq t$, write $s=t+u$ with $u\geq0$. By~\textup{(ii)} and then
\textup{(i)},
\[
        \Phi_s(p)=\Phi_u(\Phi_t(p))\leq\Phi_t(p)
\]
coordinatewise. This proves~\textup{(iii)}.
\end{proof}

The construction above gives a decreasing semigroup below the identity, but
it does not yet show that each $\Phi_t$ is monotone for the stochastic order.
Since that order is determined by upper-set masses, the next subsection studies
the evolution of $\Phi_t(p)(U)$ for each upper set $U$.

\subsection{Upper-set inequalities and order preservation}\label{sec:upper-set}

The purpose of this subsection is twofold. First, we obtain uniform upper and
lower bounds on the rate at which an upper-set mass decreases. Second, we use
those local rate comparisons in a first-contact argument to prove that two
initially ordered trajectories cannot cross. This will yield the order
preservation needed for Scott continuity and finite separation.

We first compute the local decay rates of upper-set coordinates.
\label{subsec:upper-local}

\begin{lemma}\label{lem:max-support-meets-upper}
Let $p\in\Delta_{\leq1}(P)$ and let $U\subseteq P$ be an upper set. If
$p(U)>0$, then
\(
        A(p)\cap U\neq\varnothing.
\)
\end{lemma}

\begin{proof}
Choose $x\in\supp(p)\cap U$. Since $\supp(p)$ is finite, there is a maximal
element $a$ of $\supp(p)$ with $x\leq a$. Then $a\in A(p)$. Since $U$ is an
upper set, $a\in U$.
\end{proof}

For an upper set $U\subseteq P$ and $p\in\Delta_{\leq1}(P)$, define
\begin{equation}\label{eq:gamma}
        \gamma_U(p)=
        \begin{cases}
        c(A(p))\,|A(p)\cap U|,&p\neq0,\\
        0,&p=0.
        \end{cases}
\end{equation}

\begin{lemma}\label{lem:local-upper-formula}
Let $p\in\Delta_{\leq1}(P)$, let $U$ be an upper set, and let $s\geq0$. There
is $\eta>0$ such that
\begin{equation}\label{eq:local-upper-formula}
        \Phi_{s+h}(p)(U)
        =\Phi_s(p)(U)-h\gamma_U(\Phi_s(p))
        \qquad(0\leq h\leq\eta).
\end{equation}
\end{lemma}
\begin{proof}
Put \(q=\Phi_s(p)\). If \(q=0\), then
\(\Phi_{s+h}(p)=0\) for all \(h\geq0\). Since
\(\gamma_U(0)=0\), equation~\eqref{eq:local-upper-formula}
holds for every \(\eta>0\).

Suppose that \(q\neq0\). By Lemma~\ref{lem:basic-family}(ii),
\(
        \Phi_{s+h}(p)=\Phi_h(q).
\)
Choose
\[
        \eta
        =\min_{a\in A(q)}
          \frac{q_a}{c(A(q))}
        >0.
\]
For every \(a\in A(q)\) and every \(0\leq h\leq\eta\), we have
\[
        q_a-hc(A(q))
        \geq q_a-\eta c(A(q))
        \geq0.
\]
Hence the first recursive stage starting from \(q\) is valid throughout
\([0,\eta]\), and therefore
\[
        \Phi_h(q)
        =
        q-hc(A(q))
        \sum_{a\in A(q)}\delta_a
        \qquad(0\leq h\leq\eta).
\]
Since \(\delta_a(U)=1\) precisely when \(a\in U\), it follows that
\[
\begin{aligned}
        \Phi_h(q)(U)
        &=
        q(U)
        -hc(A(q))
         \sum_{a\in A(q)}\delta_a(U)\\
        &=
        q(U)
        -hc(A(q))\,|A(q)\cap U|\\
        &=
        q(U)-h\gamma_U(q).
\end{aligned}
\]
Moreover,
\[
        q(U)-hc(A(q))\,|A(q)\cap U|
        =
        \sum_{x\in U\setminus A(q)}q_x
        +
        \sum_{a\in A(q)\cap U}
        \bigl(q_a-hc(A(q))\bigr) \geq0.
\]
Using \(q=\Phi_s(p)\) and
\(\Phi_{s+h}(p)=\Phi_h(q)\), we obtain
\[
        \Phi_{s+h}(p)(U)
        =
        \Phi_s(p)(U)
        -h\gamma_U(\Phi_s(p))
\]
for every \(0\leq h\leq\eta\). The formula remains valid at
\(h=\eta\).
\end{proof}

The local formulas are next upgraded to uniform two-sided estimates.
\label{subsec:upper-uniform}

\begin{lemma}\label{lem:upper-set-estimate}
For every upper set $U\subseteq P$, every
$p\in\Delta_{\leq1}(P)$, and every $t\geq0$,
\begin{equation}\label{eq:upper-set-estimate}
        \Phi_t(p)(U)\leq\bigl(p(U)-t\bigr)^+ = \max\{p(U)-t,0\}.
\end{equation}
\end{lemma}

\begin{proof}
Put
\(H(s)=\Phi_s(p)(U) \) 
for every \(s\geq0\).
If \(t=0\), then \(\Phi_0(p)=p\), and hence
\(
        \Phi_0(p)(U)=p(U)=\bigl(p(U)-0\bigr)^+.
\)
Thus we may assume that \(t>0\).
If \(H(t)=0\), then
\eqref{eq:upper-set-estimate} holds naturally. 
Suppose therefore that
\(
        H(t)>0.
\)
By Lemma~\ref{lem:basic-family}(iii), for every \(0\leq s\leq t\),
\(
        \Phi_t(p)\leq\Phi_s(p)
\)
coordinatewise. Evaluating both sides on \(U\) gives
\[
        H(s)=\Phi_s(p)(U)
        \geq\Phi_t(p)(U)
        =H(t)>0.
\]

Consider the recursive construction of \(\Phi_r(p)\). Recall that
\[
        q_j=\Phi_{s_j}(p),\qquad
        A_j=A(q_j),
\]
and, for \(0\leq u\leq\tau_j\),
\[
        \Phi_{s_j+u}(p)
        =
        q_j-u\,c(A_j)\sum_{a\in A_j}\delta_a.
\]
Since \(\Phi_t(p)(U)>0\), the construction has not reached the zero
valuation before time \(t\). Thus there is an index \(k\) such that
\[
        s_k\leq t\leq s_{k+1},
        \qquad
        s_{j+1}=s_j+\tau_j
        \quad(0\leq j<k).
\]

We first consider a complete recursive interval
\([s_j,s_{j+1}]\), where \(0\leq j<k\). Since
\[
        q_j(U)
        =\Phi_{s_j}(p)(U)
        \geq\Phi_t(p)(U)
        >0,
\]
Lemma~\ref{lem:max-support-meets-upper} gives
\(
        A_j\cap U\neq\varnothing.
\)
Consequently,
\[
        |A_j\cap U|\geq1.
\]
Moreover, by the definition of \(c\),
\[
        c(A_j)
        =(n+1)^{\,n-|\down A_j|}
        \geq1.
\]
Evaluating the recursive formula at \(u=\tau_j\), we obtain
\begin{align*}
    \Phi_{s_{j+1}}(p)(U)
        &=
        \Phi_{s_j}(p)(U)
        -\tau_j c(A_j)
          \sum_{a\in A_j}\delta_a(U)\\
          &=
        \Phi_{s_j}(p)(U)
        -\tau_j c(A_j)|A_j\cap U|.
\end{align*}
Therefore
\[
\begin{aligned}
        \Phi_{s_j}(p)(U)
        -\Phi_{s_{j+1}}(p)(U)
        =
        \tau_j c(A_j)|A_j\cap U|
        \geq\tau_j
        =s_{j+1}-s_j.
\end{aligned}
\]

It remains to consider the last interval from \(s_k\) to \(t\). Put
\[
        u=t-s_k,
        \qquad 0\leq u\leq\tau_k.
\]
Again,
\[
        q_k(U)
        =\Phi_{s_k}(p)(U)
        \geq\Phi_t(p)(U)
        >0,
\]
so Lemma~\ref{lem:max-support-meets-upper} yields
\(
        A_k\cap U\neq\varnothing.
\)
Hence
\[
        c(A_k)|A_k\cap U|\geq1.
\]
Using the recursive formula once more,
\[
\begin{aligned}
        \Phi_t(p)(U)
        &=
        \Phi_{s_k}(p)(U)
        -(t-s_k)c(A_k)|A_k\cap U|,
\end{aligned}
\]
and therefore
\[
        \Phi_{s_k}(p)(U)-\Phi_t(p)(U)
        \geq t-s_k.
\]
Adding the preceding inequalities gives
\begin{align*}
        p(U)-\Phi_t(p)(U)
        &=
        \sum_{j=0}^{k-1}
        \bigl(
          \Phi_{s_j}(p)(U)
          -\Phi_{s_{j+1}}(p)(U)
        \bigr)
        +\Phi_{s_k}(p)(U)-\Phi_t(p)(U)\\
        &\geq
        \sum_{j=0}^{k-1}(s_{j+1}-s_j)
        +(t-s_k)
        =t.
\end{align*}
Thus
\(
        \Phi_t(p)(U)\leq p(U)-t.
\)
Since \(\Phi_t(p)(U)>0\), this also implies \(p(U)-t>0\), and hence
\[
        \Phi_t(p)(U)
        \leq p(U)-t
        =\bigl(p(U)-t\bigr)^+.
\]
\end{proof}

Recall that we set $K_P=n(n+1)^{n-1}$.

\begin{lemma}\label{lem:lower-set-estimate}
For every nonempty upper set $U\subseteq P$, every
$p\in\Delta_{\leq1}(P)$, and every $t\geq0$,
\begin{equation}\label{eq:lower-set-estimate}
  \Phi_t(p)(U)\geq\bigl(p(U)-K_Pt\bigr)^+.
\end{equation}
\end{lemma}

\begin{proof}
If \(t=0\), then
\(
        \Phi_0(p)(U)=p(U)
        =\bigl(p(U)-K_Pt\bigr)^+,
\)
so the assertion is immediate. 
We henceforth assume that \(t>0\).

Recall the recursive construction of \(\Phi_t(p)\). For every nonzero
recursive stage, we have
\(
        q_j=\Phi_{s_j}(p),
        \
        A_j=A(q_j),
\)
and, for \(0\leq u\leq\tau_j\),
\[
        \Phi_{s_j+u}(p)
        =
        q_j-u\,c(A_j)\sum_{a\in A_j}\delta_a.
\]
Evaluating this identity on \(U\) gives
\[
\begin{aligned}
        \Phi_{s_j+u}(p)(U)
        &=
        q_j(U)
        -u\,c(A_j)
          \sum_{a\in A_j}\delta_a(U)\\
        &=
        q_j(U)
        -u\,c(A_j)|A_j\cap U|.
\end{aligned}
\]
Since \(q_j=\Phi_{s_j}(p)\), we therefore have
\begin{equation}\label{eq:upper-mass-change}
        \Phi_{s_j}(p)(U)
        -
        \Phi_{s_j+u}(p)(U)
        =
        u\,c(A_j)|A_j\cap U|.
\end{equation}

We now give a uniform upper bound for the coefficient on the right.
Since \(q_j\neq0\), the set \(A_j\) is nonempty. Hence
\(
        |\down A_j|\geq1.
\)
By the definition
\(
        c(A_j)
        =(n+1)^{\,n-|\down A_j|},
\)
it follows that
\[
        c(A_j)\leq(n+1)^{n-1}.
\]
Also,
\(
        |A_j\cap U|
        \leq |A_j|
        \leq n.
\)
Consequently,
\[
        c(A_j)|A_j\cap U|
        \leq n(n+1)^{n-1}
        =K_P.
\]
Thus \eqref{eq:upper-mass-change} yields
\begin{equation}\label{eq:upper-mass-change-bound}
        \Phi_{s_j}(p)(U)
        -
        \Phi_{s_j+u}(p)(U)
        \leq K_Pu
        \qquad(0\leq u\leq\tau_j).
\end{equation}

Suppose first that the recursion has not reached the zero valuation before
 \(t\). Then there is an index \(k\) such that
\(
        s_k\leq t\leq s_{k+1}.
\)
For every \(j<k\), applying
\eqref{eq:upper-mass-change-bound} with \(u=\tau_j\) gives
\[
        \Phi_{s_j}(p)(U)
        -
        \Phi_{s_{j+1}}(p)(U)
        \leq K_P\tau_j
        =K_P(s_{j+1}-s_j).
\]
On the last recursive interval, put
\(
        u=t-s_k.
\)
Since \(0\leq u\leq\tau_k\), we similarly obtain
\[
        \Phi_{s_k}(p)(U)-\Phi_t(p)(U)
        \leq K_P(t-s_k).
\]
Adding these inequalities gives
\[
\begin{aligned}
        p(U)-\Phi_t(p)(U)
        &=
        \sum_{j=0}^{k-1}
        \bigl(
          \Phi_{s_j}(p)(U)
          -
          \Phi_{s_{j+1}}(p)(U)
        \bigr)
        +\Phi_{s_k}(p)(U)-\Phi_t(p)(U)\\
        &\leq
        K_P
        \sum_{j=0}^{k-1}(s_{j+1}-s_j)
        +K_P(t-s_k)
        =K_Pt.
\end{aligned}
\]

It remains to consider the case in which the recursion reaches the zero
valuation at some time \(s_N\leq t\). Then
\(
        \Phi_r(p)=0
\) for \( r\geq s_N.
\)
Applying the preceding estimate to the complete recursive intervals
\([s_j,s_{j+1}]\), \(0\leq j<N\), gives
\[
\begin{aligned}
        p(U)-\Phi_{s_N}(p)(U)
        &\leq K_Ps_N.
\end{aligned}
\]
Since \(\Phi_{s_N}(p)=0\) and \(s_N\leq t\),
\[
        p(U)-\Phi_t(p)(U)
        =p(U)
        \leq K_Ps_N
        \leq K_Pt.
\]
Thus in all cases
\[
        p(U)-\Phi_t(p)(U)\leq K_Pt,
\]
Finally, since \(\Phi_t(p)(U)\geq0\), we obtain
\[
        \Phi_t(p)(U)
        \geq
        \max\{p(U)-K_Pt,0\}
        =
        \bigl(p(U)-K_Pt\bigr)^+.
\]
This proves \eqref{eq:lower-set-estimate}.
\end{proof}

Combining Lemma~\ref{lem:upper-set-estimate} and Lemma~\ref{lem:lower-set-estimate}, we have the following result.

\begin{corollary}\label{cor:finite-two-sided}
For every nonempty upper set $U\subseteq P$,
\begin{equation}\label{eq:finite-two-sided}
  \bigl(p(U)-K_Pt\bigr)^+
  \leq\Phi_t(p)(U)
  \leq\bigl(p(U)-t\bigr)^+
  \qquad(p\in\Vsub(P),\ t\geq0).
\end{equation}
\end{corollary}

We now prove order preservation by a boundary comparison and first-contact argument.
\label{subsec:upper-contact}

We start to show that $\Phi_t$ is order-preserving. Before that, we need the following lemma.

\begin{lemma}\label{lem:max-support-comparison}
If $p\stle q$ and $p\neq0$, then
\(
        A(p)\subseteq\down A(q).
\)
\end{lemma}

\begin{proof}
Let $a\in A(p)$. Since $p_a>0$, we have $p(\up a)>0$, and therefore
$q(\up a)>0$. Choose $x\in\supp(q)\cap\up a$. Since $\supp(q)$ is finite,
there is a maximal element $b$ of $\supp(q)$ with $x\leq b$. Then
$b\in A(q)$ and $a\leq x\leq b$.
\end{proof}

\begin{lemma}\label{lem:boundary-comparison}
Let $p\stle q$, and let $U$ be an upper set such that $p(U)=q(U)$. Then
\(
        \gamma_U(p)\geq\gamma_U(q).
\)
\end{lemma}

\begin{proof}
If $p(U)=q(U)=0$, then neither support meets $U$, and hence $\gamma_U(p)=\gamma_U(q) = 0$. Suppose that $p(U)=q(U) > 0$. By
Lemma~\ref{lem:max-support-meets-upper}, both $A(p)\cap U$ and
$A(q)\cap U$ are nonempty. Lemma~\ref{lem:max-support-comparison} gives
$A(p)\subseteq\down A(q)$.

If $A(p)=A(q)$, then the two values of $\gamma_U$ are equal. If
$A(p)\neq A(q)$, Lemma~\ref{lem:coefficient-comparison} gives
$c(A(p))>n c(A(q))$, and hence
\[
        \gamma_U(p)
        \geq c(A(p))
        >n c(A(q))
        \geq |A(q)\cap U|c(A(q))
         =\gamma_U(q).
\]
\end{proof}

\begin{theorem}\label{thm:order-preservation}
For every $t\geq0$, the map $\Phi_t$ is order preserving, i.e.,
\(
        p\stle q
        \ \Longrightarrow\
        \Phi_t(p)\stle\Phi_t(q).
\)
\end{theorem}

\begin{proof}
Fix \(p\stle q\). For every upper set \(U\subseteq P\), define
\[
        d_U(s)
        =
        \Phi_s(q)(U)-\Phi_s(p)(U)
        \qquad(s\geq0).
\]
By the continuity of the maps
\(s\mapsto\Phi_s(p)\) and \(s\mapsto\Phi_s(q)\), each \(d_U\) is continuous.
Moreover, since \(p\stle q\),
\[
        d_U(0)
        =
        q(U)-p(U)
        \geq0
\]
for every upper set \(U\).

Suppose, towards a contradiction, that
\(
        d_U(s)<0
\)
for some upper set \(U\) and some \(s\geq0\). Let
\[
        \sigma
        =
        \inf
        \bigl\{
          s\geq0:
          d_U(s)<0
          \text{ for some upper set }U
        \bigr\}.
\]
The set in braces is nonempty by assumption, so \(\sigma\) is well defined.

We first show that
\(
        d_U(\sigma)\geq0
        \
        \text{for every upper set }U.
\)
Indeed, by the definition of \(\sigma\), for every \(s<\sigma\) and every
upper set \(U\),
\[
        d_U(s)\geq0.
\]
If \(\sigma>0\), letting \(s\to\sigma\) and using the continuity of
\(d_U\) gives
\(
        d_U(\sigma)\geq0.
\)
If \(\sigma=0\), the same conclusion follows from
\(
        d_U(0)=q(U)-p(U)\geq0.
\)
Thus, in all cases,
\(
        d_U(\sigma)\geq0
        \
        \text{for every upper set }U.
\)
Equivalently,
\[
        \Phi_\sigma(p)(U)
        \leq
        \Phi_\sigma(q)(U)
        \qquad
        \text{for every upper set }U,
\]
and hence
\(
        \Phi_\sigma(p)\stle\Phi_\sigma(q).
\)

We now apply Lemma~\ref{lem:local-upper-formula} at time \(\sigma\).
For each upper set \(U\), there are positive numbers
\(\eta_{p,U}\) and \(\eta_{q,U}\) such that
\[
\begin{aligned}
        \Phi_{\sigma+h}(p)(U)
        &=
        \Phi_\sigma(p)(U)
        -
        h\gamma_U(\Phi_\sigma(p)),\\
        \Phi_{\sigma+h}(q)(U)
        &=
        \Phi_\sigma(q)(U)
        -
        h\gamma_U(\Phi_\sigma(q))
\end{aligned}
\]
whenever
\[
        0\leq h
        \leq
        \min\{\eta_{p,U},\eta_{q,U}\}.
\]
Since \(P\) is finite, it has only finitely many upper sets. We may
therefore choose a single \(\eta>0\) such that both formulas hold for every
upper set \(U\) and every \(0\leq h\leq\eta\).

Subtracting the two formulas gives
\[
\begin{aligned}
        d_U(\sigma+h)
        &=
        d_U(\sigma)
        +
        h\Bigl(
          \gamma_U(\Phi_\sigma(p))
          -
          \gamma_U(\Phi_\sigma(q))
        \Bigr)
\end{aligned}
\]
for every upper set \(U\) and every \(0\leq h\leq\eta\).

Consider first an upper set \(U\) such that
\(
        d_U(\sigma)=0.
\)
Then
\[
        \Phi_\sigma(p)(U)
        =
        \Phi_\sigma(q)(U).
\]
Together with
\(
        \Phi_\sigma(p)\stle\Phi_\sigma(q),
\)
Lemma~\ref{lem:boundary-comparison} gives
\[
        \gamma_U(\Phi_\sigma(p))
        \geq
        \gamma_U(\Phi_\sigma(q)).
\]
Hence
\[
        \gamma_U(\Phi_\sigma(p))
        -
        \gamma_U(\Phi_\sigma(q))
        \geq0,
\]
and therefore
\[
        d_U(\sigma+h)
        \geq0
        \qquad(0\leq h\leq\eta).
\]

Now consider an upper set \(U\) such that
\(
        d_U(\sigma)>0.
\)
By continuity of \(d_U\), there exists \(\eta_U>0\) such that
\[
        d_U(\sigma+h)>0
        \qquad(0\leq h\leq\eta_U).
\]
Again, since there are only finitely many upper sets, we may decrease
\(\eta\), if necessary, so that
\[
        d_U(\sigma+h)\geq0
\]
for every upper set \(U\) and every \(0\leq h\leq\eta\).

Thus no upper set \(U\) satisfies
\(
        d_U(s)<0
\)
for any
\(
        s\in[\sigma,\sigma+\eta].
\)
This contradicts the definition of \(\sigma\). 

Therefore
\(
        d_U(s)\geq0
\)
for every upper set \(U\subseteq P\) and every \(s\geq0\). Hence
\[
        \Phi_s(p)(U)
        \leq
        \Phi_s(q)(U)
        \qquad
        \text{for every upper set }U,
\]
and consequently
\(
        \Phi_s(p)\stle\Phi_s(q)
        \) for every \(s\geq0.
\)
\end{proof}

\subsection[Finite-poset valuation powerdomains]{Probabilistic powerdomains of finite posets are FS-domains}
\label{sec:finite-separation}

The preceding subsections constructed an order-preserving semigroup.
To turn the maps $\Phi_t$ into an FS approximate identity, two
additional properties are required. Each $\Phi_t$ must preserve directed
suprema, and for $t>0$ it must admit a finite separator. The upper-set
estimate supplies both: it compares $\Phi_t(p)$ with a simultaneous finite
approximation of $p$, and it allows every atomic coordinate to be rounded
down to a fixed rational grid.

We next prove Scott continuity of the erosion maps.
\label{subsec:finite-FS-scott}

\begin{proposition}\label{prop:Phi-scott}
For every $t\geq0$, the map
\(
        \Phi_t:\Delta_{\leq1}(P)\longrightarrow\Delta_{\leq1}(P)
\)
is Scott-continuous.
\end{proposition}

\begin{proof}
By Theorem~\ref{thm:order-preservation}, $\Phi_t$ is order preserving. It
remains to prove that it preserves directed suprema.

Let $D\subseteq\Delta_{\leq1}(P)$ be directed, put $p=\sup D$, and let
\(
        r=\sup_{q\in D}\Phi_t(q).
\)
The image $\Phi_t(D)$ is directed because $\Phi_t$ is order preserving. Since
$q\stle p$ for every $q\in D$, order preservation also gives
$\Phi_t(q)\stle\Phi_t(p)$, and hence
\begin{equation}\label{eq:scott-one-side}
        r\stle\Phi_t(p).
\end{equation}
Fix $\eps>0$. By~\eqref{eq:directed-supremum}, for every upper set
$U$ there is $q_U\in D$ such that
\(
        q_U(U)>p(U)-\eps.
\)
There are only finitely many upper sets. Directedness therefore provides one  $q_\eps\in D$ above all the finitely many $q_U$, and hence \begin{equation}\label{eq:simultaneous-approximation}
        q_\eps(U)>p(U)-\eps
        \qquad\text{for every upper set }U.
\end{equation}
Lemma~\ref{lem:upper-set-estimate} and
\eqref{eq:simultaneous-approximation} imply
\[
        \Phi_\eps(p)(U)
        \leq\bigl(p(U)-\eps\bigr)^+
        \leq q_\eps(U)
\]
for every upper set $U$. Indeed, the second inequality is immediate when
$p(U)\leq\eps$, and otherwise it follows from
\eqref{eq:simultaneous-approximation}. Hence
$\Phi_\eps(p)\stle q_\eps$. Applying $\Phi_t$ and using order
preservation and Lemma~\ref{lem:basic-family}(ii), we obtain
\begin{equation}\label{eq:time-shift-comparison}
        \Phi_{t+\eps}(p)
        =\Phi_t(\Phi_\eps(p))
        \stle\Phi_t(q_\eps)
        \stle r.
\end{equation}
For fixed $p$, the map $s\mapsto\Phi_s(p)$ is continuous. 
For each
upper set $U$, letting
$\eps\rightarrow0$ in~\eqref{eq:time-shift-comparison} gives $\Phi_t(p)(U)\leq r(U)$. 
Since the finite family of upper-set coordinates determines the stochastic order, this ordinary coordinate limit yields $\Phi_t(p)\stle r$. Together with~\eqref{eq:scott-one-side}, this yields
\[
        \Phi_t(p)=r=\sup_{q\in D}\Phi_t(q).
\]
Therefore $\Phi_t$ preserves directed suprema and is Scott-continuous.
\end{proof}

We then establish finite separation and obtain the FS approximate identity.
\label{subsec:finite-FS-separation}

\begin{proposition}\label{prop:finite-separation}
For every $t>0$, the Scott-continuous map $\Phi_t$ is finitely separated from
$\id_{\Delta_{\leq1}(P)}$.
\end{proposition}

\begin{proof}
Choose \(N\geq1\) such that \(n/N<t\), and let
\[ M_N
    =
    \left\{
        m\in\Delta_{\leq1}(P):
        Nm_x\in\N
        \text{ for every }x\in P
    \right\}.
\]
For every \(m\in M_N\), write \(m_x=k_x/N\), where
\(k_x\in\N\). Since
\(
    \sum_{x\in P}m_x\leq1,
\)
we have
\(
    \sum_{x\in P}k_x\leq N.
\)
In particular, \(0\leq k_x\leq N\) for every \(x\in P\). Thus each
coordinate \(m_x\) belongs to the finite set
\(
    \left\{0,\frac1N,\ldots,1\right\}.
\)
Since \(P\) is finite, only finitely many such vectors exist. Hence \(M_N\)
is finite.
For $p\in\Delta_{\leq1}(P)$, round each coordinate
down by setting
\[
        \lambda_x^p=\frac{\lfloor Np_x\rfloor}{N}.
\]
For every $x$, $\lambda_x^p\leq p_x$, so
\[
        \sum_{x\in P}\lambda_x^p
        \leq\sum_{x\in P}p_x
        \leq1.
\]
Thus $\lambda^p\in M_N$ and $\lambda^p\leq p$ coordinatewise; in
particular, $\lambda^p\stle p$. For every upper set $U$,
\begin{equation}\label{eq:rounding-error}
        0\leq p(U)-\lambda^p(U)=\sum_{x\in U}p_x-\sum_{x\in U}\frac{\lfloor Np_x\rfloor}{N}=\sum_{x\in U}\frac{Np_x-\lfloor Np_x\rfloor}{N}
        \leq\frac{|U|}{N}
        \leq\frac nN<t.
\end{equation}
If $p(U)\leq t$, Lemma~\ref{lem:upper-set-estimate} gives
$\Phi_t(p)(U)=0\leq\lambda^p(U)$. If $p(U)>t$, then
\eqref{eq:rounding-error} gives $p(U)-t<\lambda^p(U)$, while
Lemma~\ref{lem:upper-set-estimate} gives
\[
        \Phi_t(p)(U)\leq p(U)-t<\lambda^p(U).
\]
Thus $\Phi_t(p)\stle\lambda^p\stle p$ for every $p$, and $M_N$ is a finite
separator for $\Phi_t$.
\end{proof}

\begin{theorem}\label{thm:subprobability-FS}
For every finite poset $P$, the powerdomain $\Vsub(P)$ is a countably based
FS-domain.
\end{theorem}

\begin{proof}
If $P=\varnothing$, then $\Vsub(P)$ is a singleton. Assume that $P$ is
nonempty. By the standard results recalled in Section~\ref{sec:finite-valuations},
$\Vsub(P)\cong\Delta_{\leq1}(P)$ is a countably based domain. It remains
to construct an FS approximate identity.

For $k\in\N$, put $t_k=2^{-k}$. Since $t_{k+1}<t_k$, Lemma~\ref{lem:basic-family}(iii) gives, for every $p$,
\begin{equation}\label{eq:directed-approximation}
        \Phi_{t_k}(p)\stle\Phi_{t_{k+1}}(p)\stle p.
\end{equation}
Each $\Phi_{t_k}$ is Scott-continuous by
Proposition~\ref{prop:Phi-scott} and finitely separated by
Proposition~\ref{prop:finite-separation}.

Fix $p\in\Delta_{\leq1}(P)$. The map $t\mapsto\Phi_t(p)$ is continuous and
$\Phi_0(p)=p$, so for every upper set $U$,
\[
        \lim_{k\to\infty}\Phi_{t_k}(p)(U)=p(U).
\]
The sequence in~\eqref{eq:directed-approximation} is directed. If $r$ denotes
its supremum, then~\eqref{eq:directed-supremum} gives
\[
        r(U)
        =\sup_k\Phi_{t_k}(p)(U)
        =\lim_{k\to\infty}\Phi_{t_k}(p)(U)
        =p(U)
\]
for every upper set $U$. Hence $r=p$. Thus
$\sup_k\Phi_{t_k}=\id$ pointwise, and $(\Phi_{t_k})_{k\in\N}$ is an FS
approximate identity.
\end{proof}

Thus the maps $(\Phi_{2^{-k}})_{k\in\N}$ form an explicit increasing FS
approximate identity on $\Vsub(P)$. This completes the subprobability
analysis for finite posets and provides the finite generators used in the
second half of the paper.


We finally consider normalized probabilistic valuations on finite posets.
\label{subsec:finite-FS-normalized}

Let $Q$ be a finite nonempty poset. Its normalized probabilistic powerdomain is
\[
        \Vone(Q)=\Delta_1(Q)
        =\left\{p\in\Rplus^Q:\sum_{x\in Q}p_x=1\right\},
\]
ordered by the same upper-set inequalities. Recall that a \emph{least}
element is below every point, whereas a \emph{minimal} element merely has no
strictly smaller point. This distinction is decisive here. If $Q$ has a
least element, that point can store the missing mass and the normalized case
reduces to the subprobability case. If $Q$ has no least element, the
probability vectors supported on $\Min(Q)$ form an infinite family of
minimal elements of $\Vone(Q)$, which is incompatible with a finite
separator.

\begin{proposition}\label{prop:probability-positive}
If $Q$ has a least element $\bot$, then restriction of coordinates gives an
order isomorphism
\[
        \Vone(Q)\cong\Vsub(Q\setminus\{\bot\}).
\]
Consequently, $\Vone(Q)$ is a countably based FS-domain.
\end{proposition}

\begin{proof}
This is the finite-poset instance of Edalat's lifting trick
\cite[Lemma~6.1]{GoubaultLarrecq2022}. 
 Put $R=Q\setminus\{\bot\}$. Restriction sends
$p\in\Vone(Q)$ to $p|_R\in\Vsub(R)$. Its inverse sends
$\mu\in\Vsub(R)$ to the probabilistic valuation $\widehat\mu$ defined by
\[
  \widehat\mu_x=\mu_x\quad(x\in R),
  \qquad
  \widehat\mu_\bot=1-\mu(R).
\]
 An upper set of $Q$ not containing $\bot$ is exactly an upper set of $R$,
whereas an upper set containing $\bot$ is necessarily $Q$ itself, on which
every probabilistic valuation has value $1$. Hence restriction and its inverse
preserve and reflect the stochastic order. They are therefore inverse order
isomorphisms of dcpos. The final assertion follows from
\cref{thm:subprobability-FS}; when $Q=\{\bot\}$, both sides are singletons.
\end{proof}

\begin{proposition}\label{prop:probability-negative}
If $Q$ has no least element, then no self-map $f:\Vone(Q)\to\Vone(Q)$ admits a
finite set $F\subseteq\Vone(Q)$ such that, for every $p\in\Vone(Q)$, some
$m\in F$ satisfies
\(
        f(p)\stle m\stle p.
\)
In particular, $\Vone(Q)$ is not an FS-domain.
\end{proposition}

\begin{proof}
Let $M=\Min(Q)$. Since $Q$ is finite, every element lies above a minimal
element. Thus, if $M$ were a singleton, its unique element would be below
every point of $Q$ and would be a least element. Hence $|M|\geq2$.

Let $p\in\Vone(Q)$ be supported on $M$, and suppose that $q\stle p$. For
every $m\in M$, the set $Q\setminus\{m\}$ is upper, and therefore
\[
  1-q_m=q(Q\setminus\{m\})
  \leq p(Q\setminus\{m\})=1-p_m.
\]
Thus $q_m\geq p_m$ for every $m\in M$. Since
$\sum_{m\in M}p_m=1$ and $q$ also has total mass one, it follows that
$q_m=p_m$ for every $m\in M$ and $q(Q\setminus M)=0$. Hence $q=p$, so every
probabilistic valuation supported on $M$ is a minimal element of $\Vone(Q)$.
Then
\[
  F_M=\{p\in\Vone(Q):\supp(p)\subseteq M\}
\]
is infinite because $|M|\geq2$.

If a finite set $F$ separated a self-map $f$ from the identity, then for each
$p\in F_M$ there would be $m_p\in F$ with
$f(p)\stle m_p\stle p$. Minimality of $p$ would give $m_p=p$, forcing the
finite set $F$ to contain the infinite set $F_M$, a contradiction.
\end{proof}

\begin{corollary}\label{cor:probability-characterization}
For every finite nonempty poset $Q$,
\[
        \Vone(Q)\text{ is an FS-domain}
        \quad\Longleftrightarrow\quad
        Q\text{ has a least element}.
\]
In the positive case, $\Vone(Q)$ is countably based.
\end{corollary}

The finite-poset analysis is now complete: subprobabilistic valuations always
form an FS-domain, whereas normalized probabilistic valuations do so exactly
when the underlying finite poset has a least element. We now use the
subprobability spaces $\Vsub(P)$ as finite-dimensional factorization objects
for a class of general domains.

\section{Factorization-approximation: a new method to generate domain categories}
\label{sec:factorization-classes}

In this section, we develop a new method for constructing objects from a given class of known objects or structures. Unlike standard constructions such as products, coproducts, and subobjects, our method transfers essential properties of the given class to newly constructed objects by factoring morphisms through members of that class. 
Applying this method, we construct a class $\omega{\bf FVA} $ of objects whose categorical properties settle the generalized Jung–Tix problem.
The underlying principle of the method is not specific to the present setting and may find applications in other areas of mathematics.

\subsection{Factorization approximation}

More precisely, from a prescribed class ${\bf C}$ we generate a full subcategory by requiring the identity to be approximated by maps factoring through objects of ${\bf C}$. This construction transfers continuity, countable bases, finite separation, and closure properties from the generators to the generated objects. For the Jung--Tix problem, the natural generators are the finite-dimensional probabilistic domains $\Vsub(P)$, where $P$ is finite; the method therefore incorporates probabilistic structure at every approximation stage while retaining finite order-theoretic data.

\begin{definition}\label{def:C-factorization}
Let ${\bf C}$ be a nonempty class of dcpos.
\begin{enumerate}
\item[(1)] Let $D,E$ be dcpos. A Scott-continuous map $a:D\to E$ is \emph{${\bf C}$-factorable} if there are $C\in{\bf C}$ and Scott-continuous maps
\[
  D\xrightarrow{p}C\xrightarrow{e}E
\]
such that $a=e\circ p$.
\item[(2)] A \emph{${\bf C}$-factorization approximate identity} on a dcpo $D$ is an increasing sequence of ${\bf C}$-factorable maps
\[
  a_n=e_np_n:D\xrightarrow{p_n}C_n\xrightarrow{e_n}D,
  \qquad C_n\in{\bf C},
\]
such that
\[
  a_n\leq a_{n+1}\leq\id_D,
  \qquad
  \sup_{n\in\N}a_n=\id_D
\]
pointwise. A dcpo admitting such an approximate identity is called \emph{${\bf C}$-factorization approximable}. We write $\FAC$ for the full subcategory of $\DCPO$ consisting of all such dcpos.
\end{enumerate}
\end{definition}

The constant sequence $\id_C$ shows that every $C\in{\bf C}$ belongs to $\FAC$. More generally, we have the following elementary retract property.

\begin{proposition}\label{def:C-contain}
The category $\FAC$ contains every Scott-continuous retract of an object of ${\bf C}$. In particular, ${\bf C}\subseteq\FAC$.
\end{proposition}

\begin{proof}
If $D$ is a Scott-continuous retract of $C\in{\bf C}$, with $e:D\to C$, $r:C\to D$, and $re=\id_D$, then the constant sequence $\id_D=re$ is a ${\bf C}$-factorization approximate identity on $D$.
\end{proof}

\subsection{Continuity and countable bases}

 \begin{proposition}\label{prop:C-continuous}
Let ${\bf C}$ be a class of continuous dcpos. Then every object of $\FAC$ is continuous. Moreover, if every object of ${\bf C}$ is countably based, then every object of $\FAC$ is countably based.
\end{proposition}

\begin{proof}
Let \(D\in\FAC\), and choose a
 \({\bf C}\)-factorization approximate identity
 \[
   a_n=e_np_n:D\xrightarrow{p_n}C_n\xrightarrow{e_n}D, \qquad C_n\in{\bf C} 
   \qquad(n\in\N)
 \]
  such that
 \[
   a_n\leq a_{n+1}\leq\id_D,
   \qquad
   \sup_{n\in\N}a_n=\id_D
 \]
 pointwise.  
 
 We first record the following consequence of \(a_n=e_np_n\leq\id_D\):
 for every \(x\in D\) and every \(c\in C_n\),
 \begin{equation}\label{eq:C-way-below-transfer}
   c\ll_{C_n}p_n(x)
   \quad\Longrightarrow\quad
   e_n(c)\ll_Dx.
 \end{equation}
 Indeed, let \(S\subseteq D\) be directed and suppose that $
   x\leq\sup S$.  Since \(p_n\) is Scott-continuous,
 \[
   p_n(x)
   \leq p_n\left(\sup S\right)
   =\sup_{s\in S}p_n(s).
 \]
 If \(c\ll_{C_n}p_n(x)\), then there is some \(s\in S\) such that $   c\leq p_n(s)$.  Applying \(e_n\), we obtain
 \[
   e_n(c)
   \leq e_np_n(s)
   =a_n(s)
   \leq s.
 \]
 This proves \eqref{eq:C-way-below-transfer}.

 We next establish a local approximation property in \(D\).  Let \(x\in D\)
 and let \(U\subseteq D\) be Scott open with \(x\in U\).  Since  \(
   \sup_{n\in\N} a_n(x) = x   \), there exists 
  \(n\in\N\) such that  \(
   a_n(x)=e_np_n(x)\in U  \).  
   Thus  \(
   p_n(x)\in e_n^{-1}(U)\in\sigma(C_n)  \). 
  Since \(C_n\) is  continuous,
 \[
   p_n(x)
   =
   \sup\{c\in C_n:c\ll_{C_n}p_n(x)\}.
 \]
 Therefore there is some \(c\in C_n\) such that  $
   c\ll_{C_n}p_n(x)$  and $
   e_n(c)\in U$. Set 
 \[
   \mathord{\Uparrow}_{C_n}c
   =
   \{z\in C_n:c\ll_{C_n}z\}.
 \]
 By interpolation in the continuous dcpo \(C_n\), the set
 \(\mathord{\Uparrow}_{C_n}c\) is Scott open.  Hence 
 \(
   V=p_n^{-1}\bigl(\mathord{\Uparrow}_{C_n}c\bigr)
 \) 
 is Scott open in \(D\).  Since \(c\ll_{C_n}p_n(x)\), we have \(x\in V\). 
 If \(y\in V\), then  \(
   c\ll_{C_n}p_n(y)\),  and in particular \(c\leq p_n(y)\).  Hence
 \[
   e_n(c)
   \leq e_np_n(y)
   =a_n(y)
   \leq y.
 \]
 It follows that
 \[
   x\in V
   \subseteq\uparrow e_n(c)
   \subseteq U,
 \]
 where the last inclusion follows from \(e_n(c)\in U\) and the fact that
 every Scott-open set is an upper set.  Consequently,
 \begin{equation}\label{eq:C-local-principal}
   x\in
   \operatorname{int}_{\sigma(D)}
   \bigl(\uparrow e_n(c)\bigr)
   \subseteq
   \uparrow e_n(c)
   \subseteq U.
 \end{equation}
So $D$ is  continuous by \cref{lem:cont-equal}.

 Suppose now that every member of \({\bf C}\) is countably based.  For every
 \(n\in\N\), choose a countable basis \(B_n\) of \(C_n\), and put
 \[
   B=\bigcup_{n\in\N}e_n[B_n].
 \]
 Then \(B\subseteq D\) is countable. By \eqref{eq:C-local-principal}, the following  family
 $$\mathcal B=\bigcup_{n\in \N} \{\mathord{\Uparrow} e_n(b): \ b\in B_n\}$$
 is a countable base of the Scott topology of $D$. Hence, $D$ is countably based by \cref{lem:cont-property} \textup{(3)}.  
 \end{proof}

\subsection{Cartesian closedness and probability monad of the generated class}
In this subsection, we consider a subclass of FS-domains.

\begin{lemma}\label{lem:square}
Let $D$ be a continuous dcpo and let $(a_n)_{n\in\N}$ be an increasing sequence of Scott-continuous maps such that
\[
  a_n\leq\id_D,
  \qquad
  \sup_na_n=\id_D.
\]
For every $n$, let $(h_{n,k})_{k\in\N}$ be an increasing sequence of maps finitely separated from the identity such that
\[
  \sup_kh_{n,k}=a_n.
\]
Set
\[
  \mathcal S=\{h_{n,k}^2:n,k\in\N\}.
\]
Then:
\begin{enumerate}
\item[(1)] every member of $\mathcal S$ is finitely separated; if $h_{n,k}$ factors through an object $B$, then $h_{n,k}^2$ factors through the same object;
\item[(2)] $\mathcal S$ is directed and $\sup\mathcal S=\id_D$, i.e., $D$ is an FS-domain.
\item[(3)] $\mathcal S$ contains an increasing cofinal sequence.
\end{enumerate}
\end{lemma}

\begin{proof}
If $M$ separates a map $h$, then
\[
  h^2(x)\leq h(x)\leq m\leq x
\]
for a suitable $m\in M$, so $M$ also separates $h^2$. If $h=ep$ factors through $B$, then $h^2=(h\circ e)p$.

Let $q_1^2,\ldots,q_r^2$ be finitely many members of $\mathcal S$, and let $M_l$ be a finite separator for $q_l$. By \cref{lem:FS-properties}\textup{(1)}, $q_l(m)\ll m$ for every $m\in M_l$. Repeated interpolation gives
\[
  q_l(m)\ll r_{l,m}\ll s_{l,m}\ll t_{l,m}\ll m.
\]
There are only finitely many pairs $(l,m)$. Since $a_n(z)$ increases to $z$, choose $N$ so that
\[
  t_{l,m}\leq a_N(m),
  \qquad
  r_{l,m}\leq a_N(s_{l,m})
\]
for all relevant pairs. Since $h_{N,k}(z)$ increases to $a_N(z)$, choose $K$ so that
\[
  s_{l,m}\leq h_{N,K}(m),
  \qquad
  q_l(m)\leq h_{N,K}(s_{l,m})
\]
for all $(l,m)$. Put $H=h_{N,K}$. If $q_l(x)\leq m\leq x$ with $m\in M_l$, then
\[
  q_l^2(x)\leq q_l(m)\leq H(s_{l,m})\leq H^2(m)\leq H^2(x).
\]
Thus $\mathcal S$ is directed.

For fixed $n$, the diagonal family $(h_{n,k}^2)_k$ is cofinal among the composites $h_{n,k}h_{n,l}$, and hence
\[
  \sup_kh_{n,k}^2=a_n^2.
\]
Similarly, $\sup_na_n^2=\id_D$, so $\sup\mathcal S=\id_D$.

Finally, enumerate $\mathcal S$ as $(s_j)_{j\in\N}$. Choose $t_0=s_0$ and, recursively, choose $t_{j+1}\in\mathcal S$ above both $t_j$ and $s_{j+1}$. Then $(t_j)_j$ is increasing and cofinal in $\mathcal S$.
\end{proof}

From now on, we assume that ${\bf C}$ is a nonempty class of pointed countably based FS-domains.

\begin{proposition}\label{prop:C-saturation}
Let ${\bf C}$ be a nonempty class of pointed countably based FS-domains.
We have
\[
  {\bf C}\subseteq\FAC\subseteq\FS,
  \qquad
  \mathcal{F}(\mathcal{F}(\mathbf{C}))=\FAC.
\]
Moreover, every object of $\FAC$ is pointed and countably based, and $\FAC$ is closed under Scott-continuous retracts.
\end{proposition}

\begin{proof}
We first recall that every countably based FS-domain admits an increasing sequential FS approximate identity. Indeed, choose a countable basis $B$ and enumerate the pairs $(b,c)\in B\times B$ with $b\ll c$. From an arbitrary directed FS approximate identity, choose for the $n$th pair a map sending $c$ above $b$, and recursively choose an increasing sequence dominating all maps chosen so far. The basis and interpolation properties then show that the supremum of this sequence is the identity.

Let $D\in\FAC$, and choose
\[
a_n=e_np_n:D\xrightarrow{p_n}C_n\xrightarrow{e_n}D,
  \qquad C_n\in{\bf C},
\]
with $\sup_na_n=\id_D$. By \cref{prop:C-continuous}, $D$ is continuous and countably based. For every $n$, choose an increasing FS approximate identity $(g_{n,k})_k$ on $C_n$ and set
\[
  h_{n,k}=e_ng_{n,k}p_n.
\]
If $M_{n,k}$ separates $g_{n,k}$, then $e_n[M_{n,k}]$ separates $h_{n,k}$, and $\sup_kh_{n,k}=a_n$. Applying \cref{lem:square} and taking the increasing cofinal sequence supplied there gives an FS approximate identity on $D$ whose members are ${\bf C}$-factorable. Hence $D\in\FS$. Since $C_n$ has a least element,
\[
  e_n(\bot_{C_n})\leq e_np_n(x)=a_n(x)\leq x
\]
for every $x\in D$, so $D$ is pointed.

Now suppose that $D$ has a $\FAC$-factorization approximate identity
\[
  a_n=e_np_n:D\xrightarrow{p_n}B_n\xrightarrow{e_n}D,
  \qquad B_n\in\FAC.
\]
By the preceding paragraph, each $B_n$ has an increasing approximate identity $(q_{n,k})_k$ consisting of finitely separated ${\bf C}$-factorable maps. Put
\[
  h_{n,k}=e_n q_{n,k} p_n.
\]
These maps are $\bf C$-factorable, increase to $a_n$, and are finitely separated: if $M$ separates $q_{n,k}$, then $e_n[M]$ separates $h_{n,k}$. The square refinement again yields a ${\bf C}$-factorization approximate identity on $D$. Thus $\mathcal{F}(\mathcal{F}(\mathbf{C}))=\FAC$.

The inclusion ${\bf C}\subseteq\FAC$ follows from the constant identity sequence. Finally, if $D$ is a Scott-continuous retract of $B\in\FAC$, with $e:D\to B$, $r:B\to D$, and $re=\id_D$, then $(ra_ne)_n$ is a ${\bf C}$-factorization approximate identity on $D$. Hence $\FAC$ is retract closed.
\end{proof}

\begin{proposition}\label{prop:C-product}

Suppose that for every $C_1,C_2 \in {\bf C}$, we have $C_1\times C_2\in\FAC$.
Then $D_1\times D_2\in\FAC$ for all $D_1,D_2\in\FAC$.
\end{proposition}

\begin{proof}
Choose ${\bf C}$-factorization approximate identities $(a_n^1)_n$ and $(a_n^2)_n$ on $D_1$ and $D_2$, with factorization objects $C_n^1,C_n^2\in{\bf C}$. The maps
\[
  A_n=a_n^1\times a_n^2
\]
form an increasing sequence below the identity and have supremum $\id_{D_1\times D_2}$. Each $A_n$ factors through $C_n^1\times C_n^2\in\FAC$. Applying \cref{prop:C-saturation} gives $D_1\times D_2\in\FAC$.
\end{proof}

\begin{proposition}\label{prop:C-function}
Suppose that
\(
  [C_1\to C_2]\in\FAC
  \)  for all \( C_1,C_2\in{\bf C},
\)
then $[D_1\to D_2]\in\FAC$ for all $D_1,D_2\in\FAC$.
\end{proposition}

\begin{proof}
Choose ${\bf C}$-factorization approximate identities
\[
  a_n^1=e_n^1p_n^1\quad\text{on }D_1,
  \qquad
  a_n^2=e_n^2p_n^2\quad\text{on }D_2.
\]
Define
\[
  P_n(f)=p_n^2\circ f\circ e_n^1,
  \qquad
  E_n(g)=e_n^2\circ g\circ p_n^1.
\]
Then $A_n=E_nP_n$ satisfies
\[
  A_n(f)=a_n^2\circ f\circ a_n^1
\]
and factors through $[C_n^1\to C_n^2]\in\FAC$. The sequence $(A_n)_n$ is increasing and lies below the identity. For every $f\in[D_1\to D_2]$, $x\in D_1$, and $k\in\N$,
\[
  f(a_k^1(x))
  =\sup_{n\geq k}a_n^2(f(a_k^1(x)))
  \leq\sup_nA_n(f)(x).
\]
Taking the supremum over $k$ and using Scott continuity of $f$ gives $f(x)\leq\sup_nA_n(f)(x)$. The reverse inequality is immediate. Hence $\sup_nA_n=\id_{[D_1\to D_2]}$, and \cref{prop:C-saturation} completes the proof.
\end{proof}

\begin{proposition}\label{prop:V-close}
Let $\mathcal V$ be either $\Vsub$ or $\Vone$. If
\(
  \mathcal V(C)\in\FAC
  \) for all \(C\in{\bf C},
\)
then $\mathcal V(D)\in\FAC$ for every $D\in\FAC$.
\end{proposition}

\begin{proof}
Let $(a_n=e_np_n)_n$ be a ${\bf C}$-factorization approximate identity on $D$. Functoriality gives
\[
  \mathcal V(a_n)=\mathcal V(e_n)\mathcal V(p_n),
\]
which factors through $\mathcal V(C_n)\in\FAC$. The maps $\mathcal V(a_n)$ form an increasing sequence below the identity, and local continuity gives
\[
  \sup_n\mathcal V(a_n)
  =\mathcal V(\sup_na_n)
  =\id_{\mathcal V(D)}.
\]
Applying \cref{prop:C-saturation} proves the result.
\end{proof}

The singleton is a retract of every pointed dcpo.  Applying \cref{lem:dcpo-Cartesian,prop:C-product,prop:C-function}, we obtain the following criterion.

\begin{theorem}\label{thm:C-Cartesian closed}
Let ${\bf C}$ be a non-empty class consisting of some pointed countably based FS-domains. Then the following conditions are equivalent:
\begin{enumerate}
\item[(1)] $\FAC$ is Cartesian closed;
\item[(2)] for all $C_1,C_2\in{\bf C}$,
\[
  C_1\times C_2\in\FAC,
  \qquad
  [C_1\to C_2]\in\FAC.
\]
\end{enumerate}
\end{theorem}

\begin{proof}
Necessity follows from ${\bf C}\subseteq\FAC$. Conversely, \cref{prop:C-saturation} shows that every object of $\FAC$ is pointed, so the singleton belongs to $\FAC$ by retract closure. The result now follows from \cref{lem:dcpo-Cartesian,prop:C-product,prop:C-function}.
\end{proof}

\subsection{$\omega{\bf FVA}$: a new category generated by subprobabilistic powerdomains of finite posets}\label{subsec:FVA}

\vskip 3mm
The preceding results in \Cref{sec:finite-poset-powerdomains} show that the subprobabilistic powerdomains of finite posets are FS-domains. We now use these spaces as factorization objects to define a new category in order to handle the generalized Jung--Tix problem.

A Scott-continuous map $a:D\to E$ is called \emph{finite-valuation factorable} if there are a finite poset $P$ and Scott-continuous maps
\[
  D\xrightarrow{p}\Vsub(P)\xrightarrow{e}E
\]
such that $a=e\circ p$. If $D=E$ and $a\leq\id_D$, then $a$ is called a \emph{finite-valuation approximant} of $D$.
\begin{definition}\label{def:FVA}
Let $D$ be a dcpo. 
\begin{enumerate}
\item[(1)] A \emph{finite-valuation approximate identity} on a dcpo $D$ is a sequence

\[
  a_n=e_np_n:D\xrightarrow{p_n}\Vsub(P_n)\xrightarrow{e_n}D
  \qquad(n\in\N),
\]
where each $P_n$ is a finite poset and
\begin{equation}\label{eq:fvai}
  a_n\leq a_{n+1}\leq\id_D,
  \qquad
  \sup_{n\in\N}a_n=\id_D
\end{equation}
pointwise. 

\item[(2)] A dcpo admitting such an approximate identity is called \emph{finite-valuation approximable}. We write $\FVA$ for the full subcategory of $\DCPO$ consisting of all finite-valuation approximable dcpos. 
\end{enumerate}
\end{definition}

\begin{remark}\label{rem:finite-poset-RB}
The approximating self-maps need not have finite image. It is different from the finite-image deflations used for RB-domains. 
For a domain $D$, the following conditions are equivalent (see \cite{Jung1989}):
\begin{enumerate}[label=\textup{(\roman*)}]
\item There are finite posets $P_n$ and Scott-continuous maps
\[
  D\xrightarrow{p_n}P_n\xrightarrow{e_n}D
  \qquad(n\in\N)
\]
such that, with $a_n=e_np_n$,
\[
  a_n\leq a_{n+1}\leq\id_D,
  \qquad
  \sup_n a_n=\id_D;
\]
\item $D$ is a countably based RB-domain.
\end{enumerate}

Thus replacing $\Vsub(P_n)$ by $P_n$ gives precisely the countably based RB-domains. No equation $p_ne_n=\id_{P_n}$ is required.

If sequences are replaced throughout by arbitrary directed families, the same finite-poset factorization condition characterizes all RB-domains.
\end{remark}

\begin{theorem}\label{thm:FVA-property}
\label{prop:generator-product}
\label{lem:generators}
\label{lem:pointed}
  The category $\FVA$ has the following properties:
\begin{enumerate}
\item[(1)] it contains $\Vsub(P)$ for every finite poset $P$, and every object of $\FVA$ is a pointed countably based FS-domain;
\item[(2)] it is saturated,
\(
  \mathcal{F}(\FVA)=\FVA,
\)
and is closed under Scott-continuous retracts;
\item[(3)] for finite posets $P,Q$,
\(
\Vsub(P)\times\Vsub(Q)\in\FVA.
\)
Consequently, $\FVA$ is closed under finite products.
\end{enumerate}
\end{theorem}

\begin{proof}
Assertions (1) and (2) follow from \cref{thm:subprobability-FS,prop:C-saturation}, applied to \(
\{\Vsub(P): P \text{ is a finite poset} \}.
\)

For (3), adjoin fresh least elements and write
\[
  \widehat P=P_\bot,
  \qquad
  \widehat Q=Q_\bot.
\]
By \cref{prop:probability-positive}, adding the missing mass at the new least element gives order isomorphisms
\[
  \Vsub(P)\cong\Vone(\widehat P),
  \qquad
  \Vsub(Q)\cong\Vone(\widehat Q).
\]
Let
\[
  R=(\widehat P\times\widehat Q)\setminus\{(\bot,\bot)\}.
\]
A further application of \cref{prop:probability-positive} gives
\[
  \Vsub(R)\cong\Vone(\widehat P\times\widehat Q).
\]
Define
\[
  E(\mu,\nu)=\mu\otimes\nu,
  \qquad
  M(\xi)=\bigl(\Vone(\pi_1)(\xi),\Vone(\pi_2)(\xi)\bigr).
\]
The marginals of a product distribution are the original factors, so $M\circ E=\id$. The map $M$ is monotone. To prove that $E$ is monotone, let $\mu\stle\mu'$ and $\nu\stle\nu'$, and let $W\subseteq\widehat P\times\widehat Q$ be upper. For $x\in\widehat P$, put
\[
  W_x=\{y\in\widehat Q:(x,y)\in W\}.
\]
Then $x\mapsto\nu(W_x)$ is nonnegative and monotone, and \cref{lem:monotone-integrals} gives
\[
  (\mu\otimes\nu)(W)
  \leq(\mu'\otimes\nu)(W)
  \leq(\mu'\otimes\nu')(W).
\]
Both maps are Euclidean-continuous. By \cref{lem:euclidean-scott}, they are Scott-continuous. Transporting this retraction across the displayed order isomorphisms makes $\Vsub(P)\times\Vsub(Q)$ a Scott-continuous retract of $\Vsub(R)$. Hence it belongs to $\FVA$ by (2). Global finite-product closure follows from \cref{prop:C-product}.
\end{proof}

By \cref{thm:FVA-property}, the saturation and finite-product parts of the construction are already established. 
The main remaining task is to prove that $\FVA$ contains the function spaces and the subprobabilistic and probabilistic powerdomains of its generators, namely
\[
  [\Vsub(Q)\to\Vsub(P)],
  \qquad
  \Vsub(\Vsub(P)),
  \qquad
  \Vone(\Vsub(P))
\]
for finite posets $P,Q$. Once these generator-level statements are proved, \cref{thm:C-Cartesian closed,prop:V-close} lift them to all objects of $\FVA$. Thus $\FVA$ is Cartesian closed and closed under both probabilistic powerdomains, and hence provides a solution to the category-existence form of the Jung--Tix problem.

\section{$\omega{\bf FVA}$ as a solution to the generalized Jung--Tix problem}
\label{sec:FVA-solution}

\subsection{Finite monotone-valuation polytopes}
\label{sec:monotone-polytopes}

In this subsection, we study the finite-dimensional function spaces arising at the finite stages of the approximation. For finite posets $A$ and $P$, the function space $[A\to\Vsub(P)]$ is both a domain-theoretic function space and
a compact convex polytope in a finite-dimensional Euclidean space. These two descriptions will be used in parallel. The pointwise order provides the order-theoretic structure needed for monotonicity and approximation from below, whereas the Euclidean realization permits finite grids, randomized
rounding, and reconstruction by finite convex combinations.

We first realize the finite monotone-valuation spaces as polytopes in order coordinates.
\label{subsec:polytope-realization}

Let $A$ and $P$ be finite posets and let
\(
  \M(A,P)
\) be the set of monotone maps from $A$ to $\Vsub(P)$,
ordered pointwise. Thus an element $x\in\M(A,P)$ is a monotone map
\(
  x:A\longrightarrow\Vsub(P).
\)
For each $a\in A$, the value $x(a)$ is a subprobabilistic valuation on $P$.
Since every monotone map from a finite dcpo is Scott-continuous,
\[
  \M(A,P)=[A\to\Vsub(P)].
\]
If $A=\varnothing$ or $P=\varnothing$, this is the one-point dcpo. Assume that both are nonempty and identify $\Vsub(P)$ with
\(
  \Delta_{\leq1}(P)
  =\left\{v\in\mathbb R_{\geq0}^{P}:\sum_{p\in P}v_p\leq1\right\}.
\)
For $x \in \M (A,P)$,
writing
\[
  x(a)=\sum_{p\in P}x_{a,p}\delta_p,
\]
the whole map $x$ is represented by the vector
$(x_{a,p})_{(a,p)\in A\times P}\in\mathbb R^{A\times P}$.
Here $x_{a,p}$ is the atomic mass assigned to $p$ by the valuation $x(a)$.

A \emph{polyhedron} in a finite-dimensional real vector space is an
intersection of finitely many closed affine half-spaces, equivalently a set
defined by finitely many linear inequalities. A bounded polyhedron is called
a \emph{polytope}. The set $\M(A,P)$ consists exactly of the vectors
$(x_{a,p})$ satisfying
\[
  x_{a,p}\geq0,
  \qquad
  \sum_{p\in P}x_{a,p}\leq1,
\]
and
\[
  \sum_{p\in U}x_{a,p}
  \leq
  \sum_{p\in U}x_{b,p}
\]
whenever $a\leq b$ in $A$ and $U$ is an upper set of $P$. The first two
families express that every $x(a)$ is a subprobabilistic valuation, and the
last family expresses the monotonicity $x(a)\stle x(b)$. Since $A$ and $P$
are finite, only finitely many inequalities occur. They are linear, so the
set is convex, and $0\leq x_{a,p}\leq1$, so it is bounded. It is also
closed; hence, by the finite-dimensional Heine--Borel theorem,
$\M(A,P)$ is a compact convex polytope in $\mathbb R^{A\times P}$.

Let $\Up(P)$ be the finite family of upper sets of $P$. For $a\in A$ and $\varnothing\neq U\in\Up(P)$, define the linear functional
\[
\lambda_{a,U}:\M(A,P)\longrightarrow\mathbb R,
  \qquad
  \lambda_{a,U}(x)=x(a)(U)
  =\sum_{p\in U}x_{a,p}.
\]
The same formula defines a linear functional on the whole ambient space
$\mathbb R^{A\times P}$, and we use the same symbol for that extension. The
coordinate corresponding to the empty upper set is identically zero and is
omitted. The pointwise order on $\M(A,P)$ therefore has the equivalent
characterization
\[
  x\leq y
  \quad\Longleftrightarrow\quad
  \lambda_{a,U}(x)\leq\lambda_{a,U}(y)
  \text{ for all }a\in A,\ U\in\Up(P).
\]
For clarity, write
\[
  \Phi_t^P:\Vsub(P)\longrightarrow\Vsub(P)
\]
for the map constructed in
\cref{eq:def_of_phi} for the poset $P$. We lift it pointwise to the finite function space by
\begin{equation}\label{eq:Psi-definition}
  \Psi_t:\M(A,P)\longrightarrow\M(A,P),
  \qquad
  \Psi_t(x)=\Phi_t^P\circ x;
\end{equation}
equivalently,
\[
  \Psi_t(x)(a)=\Phi_t^P(x(a)).
\]
Since $\Phi_t^P$ is order preserving by
\cref{thm:order-preservation}, the composite $\Phi_t^P\circ x$ is monotone
whenever $x$ is monotone. Thus $\Psi_t$ is a self-map of $\M(A,P)$.
Directed suprema in $\M(A,P)$ are computed pointwise, and
$\Phi_t^P$ is Scott-continuous by \cref{prop:Phi-scott}; hence, for every
directed family $(x_i)$ and every $a\in A$,
\[
  \Psi_t\left(\sup_i x_i\right)(a)
  =\Phi_t^P\left(\sup_i x_i(a)\right)
  =\sup_i\Phi_t^P(x_i(a))
   =\sup_i\Psi_t(x_i)(a).
\]
Therefore $\Psi_t$ is Scott-continuous. If $0\leq s\leq t$, then
\cref{lem:basic-family} gives $\Phi_t^P\leq\Phi_s^P$, and consequently
$\Psi_t\leq\Psi_s$.

Finally, for every nonempty upper set $U\subseteq P$,
\[
  \lambda_{a,U}(\Psi_t(x))
  =\Psi_t(x)(a)(U)
  =\Phi_t^P(x(a))(U).
\]
Applying \cref{cor:finite-two-sided} to the valuation $x(a)$ yields
\begin{equation}\label{eq:Psi-estimate}
  \bigl(\lambda_{a,U}(x)-K_Pt\bigr)^+
  \leq \lambda_{a,U}(\Psi_t(x))
  \leq \bigl(\lambda_{a,U}(x)-t\bigr)^+.
\end{equation}
This estimate will convert the Euclidean error of the grid rounding into an
order-theoretic error measured by the flow parameter $t$.

\begin{lemma}\label{lem:M-domain}
For all finite posets $A,P$, the dcpo $\M(A,P)$ is an FS-domain.
\end{lemma}

\begin{proof}
Every finite poset is an FS-domain, and $\Vsub(P)$ is an FS-domain by
\cref{thm:subprobability-FS}. Since every monotone map from the finite dcpo
$A$ is Scott-continuous,
\(
  \M(A,P)=[A\to\Vsub(P)].
\)
The claim follows from the Cartesian closedness of $\FS$.
\end{proof}

We next describe the pointwise order by its Hasse cone.
\label{subsec:polytope-cone}

We give a finite generating set for the order directions. A
\emph{convex cone} in a real vector space is a subset closed under addition
and multiplication by nonnegative scalars. For a set $S$ of vectors, its
\emph{conic hull} is
\[
  \cone(S)
  =\left\{\sum_{i=1}^m\alpha_i s_i:
    m\geq0,\ \alpha_i\geq0,\ s_i\in S\right\}.
\]
A cone $C$ is \emph{pointed} if $C\cap(-C)=\{0\}$; this condition ensures
that $x\leq_Cy$ defined by $y-x\in C$ is antisymmetric.

For $v=(v_p)_{p\in P}\in\mathbb R^P$ and $U\subseteq P$, write
$v(U)=\sum_{p\in U}v_p$. Define
\begin{equation}\label{eq:CP}
  C_P=
  \left\{v\in\R^P:v(U)\geq0
  \text{ for every upper set }U\subseteq P\right\}.
\end{equation}
Thus, for valuations $\mu,\nu$ on $P$,
\[
  \mu\stle\nu
  \quad\Longleftrightarrow\quad
  \nu-\mu\in C_P.
\]
For $p\in P$, the \emph{standard basis vector} $e_p\in\mathbb R^P$ has
$p$-coordinate $1$ and all other coordinates $0$. The direction $e_p$
adds mass at $p$, while $e_q-e_p$ moves one unit of mass from $p$ to $q$.
We write $p\prec q$ when $q$ \emph{covers} $p$, meaning that $p<q$ and
there is no $r$ with $p<r<q$. The cover relations form the edges of the
Hasse graph of $P$. The dual-cone and Hasse-network description of monotone
cones is classical; see~\cite{Ubhaya2001}.

The compact polytope $\M(A,P)$ is the state space of monotone valuation-valued maps, whereas $C_{A,P}$ is the cone of admissible order increments in the ambient vector space. They are not equal and play different roles in the constructions below.

\begin{lemma}\label{lem:cone-generators}
 If $p\prec q$ denotes a cover relation in $P$, then
\begin{equation}\label{eq:cone-generators}
  C_P=\cone\Bigl(
      \{e_p:p\in P\}
      \cup
      \{e_q-e_p:p\prec q\}
  \Bigr).
\end{equation}
The cone $C_P$ is pointed. Consequently the pointwise order cone of
$\M(A,P)$ is
\[
  C_{A,P}= \prod_{a \in A} C_P= C_P^A,
\]
and $C_{A,P}$ is pointed.
\end{lemma}

\begin{proof}
Adjoin a fresh least element and put $\widehat P=P_\bot$. Let
\[
  H_{\widehat P}
  =\left\{h\in\R^{\widehat P}:\sum_{x\in\widehat P}h_x=0\right\}.
\]
Deleting the $\bot$-coordinate defines a linear isomorphism
\[
  T:H_{\widehat P}\longrightarrow\R^P.
\]
Its inverse sends $v\in\R^P$ to the vector $\widehat v\in H_{\widehat P}$
defined by
\[
  \widehat v_p=v_p\quad(p\in P),
  \qquad
  \widehat v_\bot=-\sum_{p\in P}v_p.
\]
Let
\[
  \widehat C_{\widehat P}
  =\{h\in H_{\widehat P}:h(U)\geq0
      \text{ for every upper set }U\subseteq\widehat P\}.
\]
An upper set of $\widehat P$ not containing $\bot$ is exactly an upper set
of $P$, whereas the only upper set containing $\bot$ is $\widehat P$, on
which every $h\in H_{\widehat P}$ has sum zero. Hence
\begin{equation}\label{eq:T-cone}
  T(\widehat C_{\widehat P})=C_P.
\end{equation}

Put
\[
  G=\cone\{e_y-e_x:x\prec y\text{ in }\widehat P\}
  \subseteq H_{\widehat P}.
\]
For $g\in\R^{\widehat P}$,
\[
  \langle g,e_y-e_x\rangle=g(y)-g(x).
\]
Thus the dual cone $G^*$ consists precisely of the monotone real-valued functions on $\widehat P$: it is enough to impose the inequalities on cover
relations. We now identify $G^{**}$. Let $h\in\widehat C_{\widehat P}$ and
let $g$ be monotone. Since $h(\widehat P)=0$, subtracting the minimum value of
$g$ does not change $\langle g,h\rangle$. We may therefore assume that
$g\geq0$. If $0<t_1<\cdots<t_m$ are its distinct positive values and
$U_j=\{x:g(x)\geq t_j\}$, then each $U_j$ is upper and
\[
  g=\sum_{j=1}^m(t_j-t_{j-1})\mathbf 1_{U_j},
  \qquad t_0=0.
\]
Consequently,
\[
  \langle g,h\rangle
  =\sum_{j=1}^m(t_j-t_{j-1})h(U_j)\geq0.
\]
Conversely, if a vector $h$ has nonnegative pairing with every monotone
function, then the constant functions show that $h(\widehat P)=0$, and
choosing $g=\mathbf 1_U$ for an upper set $U$ gives $h(U)\geq0$. Hence
\[
  G^{**}=\widehat C_{\widehat P}.
\]
The cone $G$ is finitely generated and therefore closed. The finite-dimensional
bipolar theorem now gives
\[
  \widehat C_{\widehat P}=G.
\]
This is the standard dual-cone description of the Hasse-network cone; see
also~\cite{Ubhaya2001} and ~\cite[Theorem 14.1 and Theorem 19.1]{Rockafellar1970}.

Under $T$, a cover direction $e_m-e_\bot$, with $m$ minimal in $P$, becomes
$e_m$, and every other cover direction becomes $e_q-e_p$ for a cover
$p\prec q$ in $P$. Thus $C_P$ is generated by $e_m$ for minimal $m$ and by
the cover roots $e_q-e_p$. For an arbitrary $p\in P$, choose a saturated
chain
\[
  m=x_0\prec x_1\prec\cdots\prec x_k=p
\]
from a minimal element $m$. The telescoping identity
\[
  e_p=e_m+\sum_{i=1}^k(e_{x_i}-e_{x_{i-1}})
\]
shows that every $e_p$ lies in this cone, proving
\eqref{eq:cone-generators}.

To prove pointedness, let $v\in C_P\cap(-C_P)$. Then $v(U)=0$ for every
upper set $U$. Starting with maximal elements and proceeding downward, the
equality
\[
  0=v(\up p)=v_p+\sum_{q>p}v_q
\]
shows inductively that every coordinate $v_p$ is zero. Thus $C_P$ is
pointed.

Finally, the order on $\M(A,P)$ is pointwise. Hence
\[
  x\leq y
  \quad\Longleftrightarrow\quad
  y(a)-x(a)\in C_P\text{ for every }a\in A
  \quad\Longleftrightarrow\quad
  y-x\in C_P^A.
\]
Therefore $C_{A,P}=C_P^A$, and a finite product of pointed cones is pointed.
\end{proof}

We then compare the Euclidean and Scott structures and identify suitable interior points.
\label{subsec:polytope-interface}

The next lemma connects the Euclidean constructions below with domain
theory. The randomized coefficients will first be shown continuous in the
ordinary Euclidean topology. On the finite-dimensional ordered polytopes at
hand, monotonicity then upgrades Euclidean continuity to Scott continuity.

\begin{lemma}\label{lem:euclidean-scott}
Let $K$ and $L$ be finite products of polytopes of the form
$\M(A,P)$ and finite valuation spaces $\Vsub(Q)$ or $\Vone(Q)$. Every monotone
Euclidean-continuous map $F:K\to L$ is Scott-continuous.
\end{lemma}

\begin{proof}
Let $(x_i)_{i\in I}$ be a directed family in $K$ with supremum $x$. Consider first one component $\M(A,P)$. Directed suprema are computed pointwise on upper sets, so for every $a\in A$ and every $p\in P$,
\[
x_i(a)(\up p)\longrightarrow x(a)(\up p),
  \qquad
  x_i(a)(\up p\setminus\{p\})
  \longrightarrow x(a)(\up p\setminus\{p\}).\]
The atomic coordinates satisfy
\[
  x_i(a)_p
  =x_i(a)(\up p)-x_i(a)(\up p\setminus\{p\}).
\]
Hence every atomic coordinate of $x_i$ converges to the corresponding
coordinate of $x$. Since there are only finitely many coordinates,
$x_i\to x$ in the Euclidean topology. The same argument applies to
$\Vsub(Q)$ and $\Vone(Q)$, and then componentwise to finite products.

Euclidean continuity gives $F(x_i)\to F(x)$. Since $F$ is monotone, the family $(F(x_i))_{i\in I}$ is directed. Let $y=\sup_iF(x_i)$. For each
upper-set coordinate of $L$, directed-supremum computation and Euclidean convergence give
\[
  y(U)=\sup_iF(x_i)(U)=F(x)(U).
\]
These finitely many coordinates determine the order and the underlying valuation, so $y=F(x)$. Thus $F$ preserves directed suprema and is Scott-continuous.
\end{proof}

The later rounding construction will need a point uniformly separated from the
boundary. The \emph{Euclidean interior} $\Int K$ consists of those points that contain a sufficiently small Euclidean ball inside $K$.

\begin{lemma}\label{lem:interior-point}
If $A$ and $P$ are nonempty, then the compact convex polytope $\M(A,P)$ has
nonempty Euclidean interior.
\end{lemma}

\begin{proof}
The subprobability simplex
\(
  \Delta_{\leq1}(P)
  =
  \left\{
    \nu\in\mathbb R_{\geq0}^{P}:\nu(P)\leq1
  \right\}
\)
has nonempty interior. Choose
\(
  \rho\in\Int\Delta_{\leq1}(P),
\)
so that every atomic coordinate of $\rho$ is positive and
$\rho(P)<1$.

Since $A$ is finite, there exists a strictly order-preserving map
\(
  c:A\longrightarrow(0,1).
\)
Define
\(
  u(a)=c(a)\rho.
\)
Thus the image of $u$ lies on the open line segment
\(
  \{t\rho:0<t<1\}
  \subseteq\Int\Delta_{\leq1}(P).
\)
In particular, every atomic coordinate of every $u(a)$ is positive and
$u(a)(P)<1$.

Moreover, if $a<b$ in $A$ and $U\subseteq P$ is a nonempty upper set,
then $\rho(U)>0$ and hence
\(
  u(a)(U)
  =
  c(a)\rho(U)
  <
  c(b)\rho(U)
  =
  u(b)(U).
\)
Therefore $u$ satisfies strictly every nontrivial linear inequality
defining $\M(A,P)$. Hence $u$ lies in the Euclidean interior of
$\M(A,P)$.
\end{proof}

We have therefore represented $\M(A,P)$ as a compact ordered polytope with a
finitely generated pointed order cone and a nonempty interior.

\subsection{Monotone randomized grid rounding}
\label{sec:random-rounding}

The aim of this subsection is to replace each point of
$K=\M(A,P)$ by a probability distribution on finitely many nearby grid
points. A deterministic floor map is discontinuous at grid boundaries and need not preserve the cone order. Therefore, we introduce the random translation to remove the discontinuity after taking probabilities, while additional translations along the cover-root directions yield an explicit monotone coupling.

We first introduce cone-adapted random rounding and its error zonotope.
\label{subsec:rounding-setup}

Fix nonempty finite posets $A,P$ and put
$K=\M(A,P)\subseteq\R^d$, where $d=|A||P|$. Coordinates of
$\mathbb R^{A\times P}$ are indexed by pairs $(a,p)$. Let $e_{a,p}$ be the
standard basis vector with value $1$ in coordinate $(a,p)$ and $0$
elsewhere, and enumerate all these vectors as $e_1,\ldots,e_d$. 
Enumerate the cover-root directions
\(
  e_{a,q}-e_{a,p}
  \ (a\in A,\ p\prec q)
\) as $\xi_1,\ldots,\xi_r$. 
By \cref{lem:cone-generators},
\(
  C_{A,P}=C_P^A.
\)
For each $a\in A$, the copy of $C_P$ in the $a$-th component is
generated by
\(
  e_{a,p}
  \ (p\in P)
\)
and
\(
  e_{a,q}-e_{a,p}
  \ (p\prec q).
\)
Consequently,
\[
  C_{A,P}
  =
  \cone\Bigl(
    \{e_{a,p}:a\in A,\ p\in P\}
    \cup
    \{e_{a,q}-e_{a,p}:a\in A,\ p\prec q\}
  \Bigr).
\]
After enumerating these two finite families as
$e_1,\ldots,e_d$ and $\xi_1,\ldots,\xi_r$, respectively, the vectors
$e_i$ and $\xi_j$ generate $C_{A,P}$.

For a vector $v$, write
\[
  [0,v]=\{tv:0\leq t\leq1\}
\]
for the line segment from $0$ to $v$. For subsets $B_1,\ldots,B_m$ of a
vector space, their \emph{Minkowski sum} is
\[
  B_1+\cdots+B_m
  =\{b_1+\cdots+b_m:b_i\in B_i\}.
\]
Define the bounded set
\begin{equation}\label{eq:zonotope}
  Z_{A,P}
  =\sum_{i=1}^d[0,2e_i]
   +\sum_{j=1}^r[0,\xi_j]
  \subseteq C_{A,P}.
\end{equation}
This finite Minkowski sum of line segments is called a zonotope. Its role is
to contain every possible rounding error as we shall see; the inclusion in
$C_{A,P}$ will ensure that every rounded grid point lies below the input in
the cone order.

Let $U_1,\ldots,U_d,S_1,\ldots,S_r$ be independent random variables,
each uniformly distributed on $[0,1)$, and write
$U=(U_1,\ldots,U_d)$. For a real vector $w$, $\lfloor w\rfloor$ denotes
coordinatewise floor. 
For $z\in\mathbb R^d$, define
\begin{equation}\label{eq:random-rounding}
  Q_z=
  \left\lfloor
  z-U-\sum_{j=1}^rS_j\xi_j
  \right\rfloor
  \in\mathbb Z^d.
\end{equation}

Let $\pi_z$ be its probability distribution, that is,
\[
  \pi_z(\{\ell\})
  =
  \mathbb P(Q_z=\ell)
  \qquad(\ell\in\mathbb Z^d).
\]
When the random inputs need to be displayed
explicitly, we write $Q_z(U,S)$ for the same random vector. In one dimension
and without the $S_j$ terms,
$\lfloor z-U\rfloor$ equals $\lfloor z\rfloor$ with probability equal to
the fractional part of $z$ and equals $\lfloor z\rfloor-1$ otherwise.
Thus the individual floor map is discontinuous, but the two probabilities
vary continuously with $z$.

We next prove finite support and continuity of the rounding probabilities.
\label{subsec:rounding-continuity}

\begin{lemma}\label{lem:rounding-law}
For every $z\in\R^d$,  $\pi_z$ has finite support.  Moreover,
\begin{enumerate}[label=\textup{(\roman*)}]
\item if $\pi_z(\{\ell\})>0$, then
\begin{equation}\label{eq:rounding-support}
  z-\ell\in Z_{A,P};
\end{equation}
\item for every $\ell\in\Z^d$, the function
\(
  z\longmapsto\pi_z(\{\ell\})
\)
is continuous.
\end{enumerate}
\end{lemma}

\begin{proof}
On the event $Q_z=\ell$, put
\[
  \theta=z-U-\sum_jS_j\xi_j-\ell\in[0,1)^d.
\]
Then
\[
  z-\ell=U+\theta+\sum_jS_j\xi_j\in Z_{A,P},
\]
which proves~\textup{(i)}.  
It also shows that
\[
  \supp(\pi_z)\subseteq(z-Z_{A,P})\cap\Z^d,
\]
and the set on the right is finite.

For~\textup{(ii)}, let $z_n\to z$ in $\mathbb R^d$ and fix
$\ell\in\mathbb Z^d$. On $\Omega=[0,1)^{d+r}$, write
\[
  q_w(u,s)
  =
  \left\lfloor
    w-u-\sum_{j=1}^r s_j\xi_j
  \right\rfloor.
\]
Then
\[
  \pi_w(\{\ell\})
  =
  \int_\Omega
  \mathbf 1_{\{q_w=\ell\}}
  \,d\omega.
\]
Since
\(
  q_w(u,s)=\ell
\)
if and only if
\[
  \ell_i
  \leq
  w_i-u_i-\sum_{j=1}^r s_j(\xi_j)_i
  <
  \ell_i+1
  \qquad (1\leq i\leq d),
\]
the map
\(
  w\longmapsto \mathbf 1_{\{q_w(u,s)=\ell\}}
\)
is locally constant at $w=z$ unless, for some coordinate $i$,
\[
  z_i-u_i-\sum_{j=1}^r s_j(\xi_j)_i
  \in\{\ell_i,\ell_i+1\}.
\]

Let
\[
  B_i
  =
  \left\{
    (u,s):
    z_i-u_i-\sum_{j=1}^r s_j(\xi_j)_i
    \in\{\ell_i,\ell_i+1\}
  \right\},
\]
and put
\(
  B=\bigcup_{i=1}^d B_i.
\)
We claim that $B$ has Lebesgue measure zero. Indeed, write
\[
  B_i=B_{i,0}\cup B_{i,1},
\]
where, for $\eps\in\{0,1\}$,
\[
  B_{i,\eps}
  =
  \left\{
    (u,s):
    z_i-u_i-\sum_{j=1}^r s_j(\xi_j)_i
    =
    \ell_i+\eps
  \right\}.
\]
After all variables except $u_i$ have been fixed, the defining equality
for $B_{i,\eps}$ determines $u_i$ uniquely, namely
\[
  u_i
  =
  z_i-\ell_i-\eps
  -\sum_{j=1}^r s_j(\xi_j)_i.
\]
Thus every $u_i$-section of $B_{i,\eps}$ contains at most one
point and hence has one-dimensional Lebesgue measure zero. By Fubini's
theorem, the Lebesgue measure of $B_{i,\eps}$ is 
\(
  \mathcal{L}^{d+r}(B_{i,\eps})=0.
\)
Consequently,
\(
  \mathcal{L}^{d+r}(B_i)=0,
\)
and, since there are only finitely many coordinates,
\[
  \mathcal{L}^{d+r}(B)
  \leq
  \sum_{i=1}^d \mathcal{L}^{d+r}(B_i)
  =0.
\]
For \(w\in\mathbb R^d\), let
\[
  A_w^\ell
  =
  \{(u,s):q_w(u,s)=\ell\}.
\]
Thus
\[
  \mathbf 1_{A_w^\ell}(u,s)
  =
  \begin{cases}
  1,& q_w(u,s)=\ell,\\
  0,& q_w(u,s)\neq\ell.
  \end{cases}
\]

Fix now \((u,s)\) outside the null set of boundary parameters.  For each
coordinate \(i\), put
\[
  r_i(w)
  =
  w_i-u_i-\sum_j s_j(\xi_j)_i.
\]
Then
\[
  q_w(u,s)=\ell
\]
if and only if
\[
  \ell_i\leq r_i(w)<\ell_i+1
  \qquad\text{for every }i.
\]
Since \((u,s)\) is not a boundary parameter, at \(w=z\) none of the
numbers \(r_i(z)\) is equal to either \(\ell_i\) or \(\ell_i+1\).
Consequently, there exists a neighborhood \(N_{u,s}\) of \(z\) such
that, for every \(w\in N_{u,s}\), each of the above inequalities has
the same truth value as it has at \(w=z\).  Hence
\[
  q_w(u,s)=\ell
  \quad\Longleftrightarrow\quad
  q_z(u,s)=\ell
  \qquad(w\in N_{u,s}),
\]
or equivalently,
\[
  \mathbf 1_{A_w^\ell}(u,s)
  =
  \mathbf 1_{A_z^\ell}(u,s)
  \qquad(w\in N_{u,s}).
\]
Notice that this does not assert that
\(A_w^\ell=A_z^\ell\).  Rather, after \((u,s)\) has been fixed, the two indicator functions have the same value at that particular parameter point whenever \(w\) is sufficiently close to \(z\).

Therefore, if \(w_n\to z\), then for every \((u,s)\) outside the boundary
null set there exists \(n_0=n_0(u,s)\) such that
\[
  \mathbf 1_{A_{w_n}^\ell}(u,s)
  =
  \mathbf 1_{A_z^\ell}(u,s)
  \qquad(n\geq n_0).
\]
In particular,
\[
  \mathbf 1_{A_{w_n}^\ell}(u,s)
  \longrightarrow
  \mathbf 1_{A_z^\ell}(u,s)
\]
for almost every \((u,s)\).

Since
\[
  0\leq \mathbf 1_{A_{w_n}^\ell}(u,s)\leq1,
\]
the dominated convergence theorem yields
\[
\begin{aligned}
  \pi_{w_n}(\{\ell\})
  &=
  \int
  \mathbf 1_{A_{w_n}^\ell}(u,s)\,d(u,s)\\
  &\longrightarrow
  \int
  \mathbf 1_{A_z^\ell}(u,s)\,d(u,s)\\
  &=
  \pi_z(\{\ell\}).
\end{aligned}
\]
Thus \(w\mapsto\pi_w(\{\ell\})\) is continuous.
\end{proof}

We then construct monotone couplings for the cone order.
\label{subsec:rounding-coupling}

Equip $\Z^d$ with the order induced by $C_{A,P}$:
\[
  \ell\leq_Cm
  \quad\Longleftrightarrow\quad
  m-\ell\in C_{A,P}.
\]
Because $C_{A,P}$ is pointed, this is a partial order. 
A \emph{coupling} of two probability distributions $\mu$ and $\nu$
is a pair of random variables $(L,L')$ defined on the same probability
space such that $L$ has distribution $\mu$ and $L'$ has distribution
$\nu$.
It is a \emph{monotone coupling} if $L\leq_CL'$ almost surely, meaning with probability one. 
Such a coupling implies that $\mu$ is
stochastically below $\nu$. The next proposition constructs such a
coupling explicitly for the rounding distributions. Here stochastic order is taken with respect to the cone-induced order
$\leq_C$ on $\mathbb Z^d$ by comparison on all upper sets of the underlying ordered space.

\begin{proposition}\label{prop:rounding-monotone}
If $z,z'\in\mathbb R^d$ satisfy
\(
  z'-z\in C_{A,P}
\), then $\pi_z$ is stochastically below $\pi_{z'}$
for the order $\leq_C$. More precisely, there is a coupling $(L,L')$ of
$\pi_z$ and $\pi_{z'}$ such that $L\leq_CL'$ almost surely.
\end{proposition}

\begin{proof}
Choose coefficients $\alpha_i,\beta_j\geq0$ such that
\[
  z'-z
  =
  \sum_{i=1}^d\alpha_i e_i
  +
  \sum_{j=1}^r\beta_j\xi_j.
\]
Write
\(
  \beta_j=n_j+\theta_j,
  \
  n_j\in\mathbb N,\ 0\leq\theta_j<1.
\)
Using the random variables from \cref{eq:random-rounding}, define
\(
  S_j'=(S_j+\theta_j)\bmod 1
\)
and
\(
  \delta_j
  =
  \mathbf 1_{\{S_j+\theta_j\geq1\}}.
\)
Thus
\(
  S_j+\theta_j=S_j'+\delta_j.
\)
Translation modulo $1$ preserves the uniform distribution on $[0,1)$.
Hence $S_1',\ldots,S_r'$ are again independent and uniformly distributed
on $[0,1)$, and they remain independent of $U$.
Set
\[
  L=Q_z(U,S),
  \qquad
  L'=Q_{z'}(U,S'),
\]
and put
\[
  w=z-U-\sum_{j=1}^rS_j\xi_j,
  \qquad
  m=\sum_{j=1}^r(n_j+\delta_j)\xi_j.
\]
Since each $\xi_j$ is an integral vector,
\[
  m\in C_{A,P}\cap\mathbb Z^d.
\]
Moreover,
\[
\begin{aligned}
  z'-U-\sum_{j=1}^rS_j'\xi_j
  =
  z-U-\sum_{j=1}^rS_j\xi_j
  +\sum_{i=1}^d\alpha_i e_i
  +\sum_{j=1}^r(n_j+\delta_j)\xi_j
  =
  w+m+\sum_{i=1}^d\alpha_i e_i.
\end{aligned}
\]
Because $m$ is integral, coordinatewise flooring gives
\[
\begin{aligned}
  L'-L
  &=
  m+
  \sum_{i=1}^d
  \bigl(
    \lfloor w_i+\alpha_i\rfloor-\lfloor w_i\rfloor
  \bigr)e_i.
\end{aligned}
\]
For every $i$,
\(
  \lfloor w_i+\alpha_i\rfloor-\lfloor w_i\rfloor
  \in\mathbb N,
\)
since $\alpha_i\geq0$.  Hence $L'-L$ is a nonnegative linear combination
of the generators $\xi_j$ and $e_i$, and therefore
\(
  L'-L\in C_{A,P}.
\)
Thus
\(
  L\leq_C L'
\)
with probability one.

The random variables $L$ and $L'$ have distributions $\pi_z$ and
$\pi_{z'}$, respectively.  Indeed, $S'$ has the same distribution as
$S$.  Therefore $(L,L')$ is a monotone coupling of $\pi_z$ and
$\pi_{z'}$.
Finally, let $H\subseteq\mathbb Z^d$ be an upper set for $\leq_C$.
Since $L\leq_C L'$ almost surely,
\[
  L\in H
  \quad\Longrightarrow\quad
  L'\in H
\]
almost surely. Consequently,
\(
  \pi_z(H)
  =
  \mathbb P(L\in H)
  \leq
  \mathbb P(L'\in H)
  =
  \pi_{z'}(H).
\)
Hence
\(
  \pi_z\stle\pi_{z'}.
\)
\end{proof}

We next combine interior contraction with a finite-state encoder.
\label{subsec:rounding-encoder}

Recall that
\(
  K=\M(A,P)\subseteq\mathbb R^d
\)
is the compact convex polytope introduced above. The distribution $\pi_z$ is defined on the whole integer grid, but near the boundary
of $K$ a rounded point may lie outside $K$. We therefore move each input a small distance toward a fixed interior point before rounding. The grid size is chosen proportional to that inward displacement, so the whole rounding
error remains inside the available interior margin.

Choose $u\in\Int K$ as in \cref{lem:interior-point}. With respect to a fixed
Euclidean norm, write $B(u,r)$ and $\overline B(u,r)$ for the open and closed
balls of radius $r$ around $u$. Fix $r_0>0$ such that
\(
  \overline B(u,r_0)\subseteq\Int K,
\)
and put
\[
  R_Z=\max\{\|z\|:z\in Z_{A,P}\}.
\]
Choose $c>0$ with $cR_Z<r_0$. For $0<\eps<1/2$, define
\[
  J_\eps(x)=(1-\eps)x+\eps u,
  \qquad
  h_\eps=c\eps.
\]

\begin{lemma}\label{lem:affine-interior}
For every $x\in K$ and every $v\in\R^d$ with $\|v\|<\eps r_0$, one has
$J_\eps(x)+v\in\Int K$.
\end{lemma}

\begin{proof}
Write
\[
  J_\eps(x)+v
  =(1-\eps)x+\eps\left(u+\frac v\eps\right).
\]
The second point lies in $B(u,r_0)\subseteq\Int K$. A strict convex
combination of a point of $K$ and an interior point belongs to $\Int K$;
see \cite[Theorem~6.1]{Rockafellar1970}.
\end{proof}

Define
\[
  \Lambda_\eps
  =\{\ell\in\Z^d:h_\eps\ell\in K\}.
\]
Since $K$ is compact, the rescaled set $h_\eps^{-1}K$ is bounded and contains
only finitely many integer points. Restricting $\leq_C$ to
$\Lambda_\eps$ therefore gives a finite poset, denoted by $L_\eps$. For $x\in K$, put
\begin{equation}\label{eq:p-eps}
  p_\eps(x)
  =\pi_{J_\eps(x)/h_\eps}
  =\sum_{\ell\in\Lambda_\eps}
    \pi_{J_\eps(x)/h_\eps}(\{\ell\})\Dirac_\ell.
\end{equation}

\begin{proposition}\label{prop:p-kernel}
The map
\(
p_\eps:\M(A,P)\longrightarrow\Vone(L_\eps)\subseteq\Vsub(L_\eps)
\)
is well defined and Scott-continuous. If the coefficient of $\Dirac_\ell$
in $p_\eps(x)$ is nonzero, then
\begin{equation}\label{eq:grid-below}
  h_\eps\ell\leq_C J_\eps(x).
\end{equation}
\end{proposition}

\begin{proof}
If the coefficient at $\ell$ is nonzero, then $\pi_{J_\eps(x)/h_\eps}(\{\ell\}) > 0$, and  
\cref{eq:rounding-support} yields a $z \in Z_{A,P}$  such that
\(
    J_\eps(x)/h_\eps - \ell = z
\). Consequently,
\(
  J_\eps(x)-h_\eps\ell=h_\eps z,
\)
which proves \cref{eq:grid-below}. Moreover,
\[
  \|h_\eps z\|\leq c\eps R_Z<\eps r_0,
\]
so \cref{lem:affine-interior} applied to $v=-h_\eps z$ gives
$h_\eps\ell\in\Int K$. Thus the whole distribution is supported on $L_\eps$.

If $x\leq y$, then
\[
  \frac{J_\eps(y)-J_\eps(x)}{h_\eps}
  =\frac{1-\eps}{h_\eps}(y-x)\in C_{A,P}.
\]
The monotone coupling from \cref{prop:rounding-monotone} has both marginals supported on $L_\eps$, hence $p_\eps(x)\stle p_\eps(y)$. Every coefficient in \cref{eq:p-eps} is Euclidean-continuous by \cref{lem:rounding-law}. Since the target is finite-dimensional, \cref{lem:euclidean-scott} gives Scott continuity.
\end{proof}

We have obtained the finite probabilistic encoding
\[
  p_\eps:\M(A,P)\longrightarrow\Vone(L_\eps).
\]
It is Scott-continuous and order preserving, and every grid point occurring
with nonzero probability lies below the contracted input. The next subsection
adds a reconstruction label to each grid point and arranges the resulting
kernels into an increasing approximation of the Dirac unit.

\subsection[Finite stochastic kernels]{Finite Stochastic Kernels on Monotone Valuation Polytopes}
\label{sec:finite-kernels}

The randomized grid map
\[
  p_\eps:K=\M(A,P)\longrightarrow\Vone(L_\eps)
\]
encodes an element of \(K\) by a probability valuation on a finite poset, but
it does not yet return points of \(K\).  We now attach to every grid state
\(\ell\in L_\eps\) a label \(y_\eps(\ell)\in K\).  The labels are obtained by
applying the erosion map to the scaled grid points.  They produce both a
finite probability kernel
\[
  \kappa_\eps:K\longrightarrow\Vone(K)
\]
and its barycentric approximation
\[
  d_\eps:K\longrightarrow K.
\]
The purpose of this subsection is to choose a decreasing sequence of scales for
which these maps form increasing approximations of the Dirac unit and the
identity, respectively.

We now introduce erosion labels and prove the local sandwich estimate.

Recall that \(K_P\geq1\) is the constant in the estimates
\[
  \bigl(\lambda_{a,U}(v)-K_Pt\bigr)^+
  \leq
  \lambda_{a,U}(\Psi_t(v))
  \leq
  \bigl(\lambda_{a,U}(v)-t\bigr)^+
\]
for \(v\in K\), \(a\in A\), every nonempty \(U\in\Up(P)\), and \(t>0\).
Also recall that \(c>0\) was chosen so that
\(
  cR_Z<r_0,
\)
and that the grid size is
\(
  h_\eps=c\eps.
\)
Put
\[
  R_{A,P}
  =\max\{\lambda_{a,U}(z):z\in Z_{A,P},\ a\in A,
       \varnothing\neq U\in\Up(P)\},
\]
and set
\[
  C_*=1+2K_P+cR_{A,P}.
\]
The finite number \(R_{A,P}\) is a uniform bound on the change of every
order coordinate \(\lambda_{a,U}\) over the rounding-error set \(Z_{A,P}\).
Define \(y_\eps:L_\eps\to K\) by
\begin{equation}\label{eq:labels}
  y_\eps(\ell)=\Psi_{2\eps}(h_\eps\ell).
\end{equation}
Equivalently, for every \(a\in A\),
\[
  y_\eps(\ell)(a)
  =
  \Phi_{2\eps}^P\bigl((h_\eps\ell)(a)\bigr).
\]
Thus the grid point \(h_\eps\ell\in K\) is moved farther downward, in the order of \(K\), by applying \(\Psi_{2\eps}\).
This additional margin makes the constructions at successive scales comparable.

\begin{lemma}\label{lem:flow-sandwich}
The map \(y_\eps:L_\eps\to K\) is monotone.  Put
\[
  \alpha_\eps=\frac{\eps}{K_P},
  \qquad
  \beta_\eps=C_*\eps.
\]
For every \(x\in K\) and every
\(\ell\in\supp(p_\eps(x))\), one has
\begin{equation}\label{eq:label-sandwich}
  \Psi_{\beta_\eps}(x)
  \leq y_\eps(\ell)
  \leq\Psi_{\alpha_\eps}(x).
\end{equation}
\end{lemma}

\begin{proof}
Monotonicity follows from the order preservation of \(\Psi_{2\eps}\). For the
estimates, nonzero weight gives
\[
  J_\eps(x)-h_\eps\ell=h_\eps z
  \qquad(z\in Z_{A,P}).
\]
Fix $a\in A$ and a nonempty upper set $U\subseteq P$, and write
$\lambda=\lambda_{a,U}$. 
Using \[ J_\eps(x)=(1-\eps)x+\eps u, \qquad h_\eps=c\eps, \] and the linearity of $\lambda$, we obtain \[ \lambda(h_\eps\ell) = (1-\eps)\lambda(x)+\eps\lambda(u)-c\eps\lambda(z). \] Since \[ 0\leq\lambda(x),\lambda(u)\leq1, \qquad 0\leq\lambda(z)\leq R_{A,P}, \] it follows that \[ \begin{aligned} \lambda(h_\eps\ell) &\geq (1-\eps)\lambda(x) -c\eps R_{A,P} \geq \lambda(x)-\eps-c\eps R_{A,P},\\ \lambda(h_\eps\ell) &\leq (1-\eps)\lambda(x)+\eps \leq \lambda(x)+\eps. \end{aligned} \] Hence \[ \lambda(x)-\eps-c\eps R_{A,P} \leq \lambda(h_\eps\ell) \leq \lambda(x)+\eps. \]
Using \cref{eq:Psi-estimate} with $t=2\eps$,
we obtain
\[
\begin{aligned}
  \lambda(y_\eps(\ell))
  &\leq
  \bigl(\lambda(h_\eps\ell)-2\eps\bigr)^+,\\
  \lambda(y_\eps(\ell))
  &\geq
  \bigl(\lambda(h_\eps\ell)-2K_P\eps\bigr)^+.
\end{aligned}
\]
Since
\[
  \lambda(x)-\eps-c\eps R_{A,P}
  \leq
  \lambda(h_\eps\ell)
  \leq
  \lambda(x)+\eps,
\]
and since $r\mapsto r^+$ is monotone, it follows that
\[
\begin{aligned}
  \lambda(y_\eps(\ell))
  &\leq
  \bigl(\lambda(x)-\eps\bigr)^+,\\
  \lambda(y_\eps(\ell))
  &\geq
  \bigl(\lambda(x)-C_*\eps\bigr)^+,
\end{aligned}
\]
where
\[
  C_*=1+2K_P+cR_{A,P}.
\]
Now set
\[
  \alpha_\eps=\frac{\eps}{K_P},
  \qquad
  \beta_\eps=C_*\eps.
\]
Applying \cref{eq:Psi-estimate} to $x$ with parameters
$t=\alpha_\eps$ and $t=\beta_\eps$, respectively, gives
\[
  \bigl(\lambda(x)-\eps\bigr)^+
  \leq
  \lambda(\Psi_{\alpha_\eps}(x))
\]
and
\[
  \lambda(\Psi_{\beta_\eps}(x))
  \leq
  \bigl(\lambda(x)-C_*\eps\bigr)^+.
\]
Hence
\[
  \lambda(\Psi_{\beta_\eps}(x))
  \leq
  \lambda(y_\eps(\ell))
  \leq
  \lambda(\Psi_{\alpha_\eps}(x)).
\]
Since the functionals $\lambda_{a,U}$ determine the order on
$\M (A,P)$, we conclude that
\[
  \Psi_{\beta_\eps}(x)
  \leq
  y_\eps(\ell)
  \leq
  \Psi_{\alpha_\eps}(x).
\]

\end{proof}

We next construct finite kernels and their barycentric maps.

A \emph{probability kernel on \(K\)} is a Scott-continuous map
\(
  k:K\longrightarrow\Vone(K).
\)
In this paper, such a kernel is called \emph{finite} if there are a finite
poset $L$ and Scott-continuous maps
\[
K\xrightarrow{p}\Vone(L)\xrightarrow{\Vone(y)}\Vone(K)
\]
with $y:L\to K$ and $k=\Vone(y)\circ p$. 
For a finite probabilistic valuation
\[
  \rho=\sum_{\ell}r_\ell\Dirac_\ell
\]
and labels \(y_\ell\in K\), its pushforward along \(y\) is
\[
  \Vone(y)(\rho)=\sum_\ell r_\ell\Dirac_{y_\ell},
\]
and its barycentre is \(\sum_\ell r_\ell y_\ell\). Since $K$ is convex and contains the zero map, the same formula defines a point of $K$ for a subprobability vector, with the
missing mass placed at zero. Define
\begin{align}
  e_\eps&:\Vsub(L_\eps)\longrightarrow K,
   \quad e_\eps(\nu)=\sum_{\ell\in L_\eps}\nu_\ell y_\eps(\ell),
  \label{eq:e-eps}\\
  \kappa_\eps&=\Vone(y_\eps)\circ p_\eps:K\longrightarrow\Vone(K),
  \label{eq:kappa-eps}\\
  d_\eps&=e_\eps\circ p_\eps:K\longrightarrow K.
  \label{eq:d-eps}
\end{align}
Thus $\kappa_\eps(x)$ is the finite probability distribution obtained by
replacing each grid state $\ell$ by its label $y_\eps(\ell)$, and
$d_\eps(x)$ is its barycentre. In \cref{eq:e-eps}, if $\nu$ has total mass
less than one, the missing mass is placed at the zero map of $K$; this does
not change the displayed sum.

\begin{proposition}\label{prop:kernel-sandwich}
The maps in \cref{eq:e-eps,eq:kappa-eps,eq:d-eps} are Scott-continuous and,
for every \(x\in K\),
\begin{align}
  \Dirac_{\Psi_{\beta_\eps}(x)}
  &\leq\kappa_\eps(x)
  \leq\Dirac_{\Psi_{\alpha_\eps}(x)},
  \label{eq:kappa-sandwich}\\
  \Psi_{\beta_\eps}(x)
  &\leq d_\eps(x)
  \leq\Psi_{\alpha_\eps}(x).
  \label{eq:d-sandwich}
\end{align}
The kernel $\kappa_\eps$ factors through $\Vone(L_\eps)$, and $d_\eps$
factors through $\Vsub(L_\eps)$.
\end{proposition}
\begin{proof}
For every \(a\in A\) and every nonempty \(U\in\Up(P)\), the map
\(
  \ell\longmapsto \lambda_{a,U}(y_\eps(\ell))
\)
is nonnegative and monotone. Hence, by
\cref{lem:monotone-integrals}, the map \(e_\eps\) is monotone.
It is Euclidean-continuous and therefore Scott-continuous by
\cref{lem:euclidean-scott}. Consequently, the Scott continuity of
\(\kappa_\eps\) and \(d_\eps\) follows from their respective
factorizations
\[
  \kappa_\eps=\Vone(y_\eps)\circ p_\eps,
  \qquad
  d_\eps=e_\eps\circ p_\eps.
\]

We first record a simple consequence of the stochastic order. Suppose that
\[
  a\leq x_i\leq b
  \qquad (1\leq i\leq n),
\]
and let \((r_i)_{i=1}^n\) be a probability vector. If \(X\) is a random
variable satisfying \(\mathbb P(X=x_i)=r_i\), then
\(
  a\leq X\leq b
  \ \text{almost surely}.
\)
Thus \((a,X)\) and \((X,b)\) are monotone couplings, and hence
\[
  \delta_a
  \leq
  \sum_{i=1}^n r_i\delta_{x_i}
  \leq
  \delta_b.
\]

Now fix \(x\in K\). By \cref{eq:label-sandwich}, every
\(\ell\in\supp(p_\eps(x))\) satisfies
\(
  \Psi_{\beta_\eps}(x)
  \leq
  y_\eps(\ell)
  \leq
  \Psi_{\alpha_\eps}(x).
\)
Applying the preceding observation to the probability vector
\(
  \bigl(p_\eps(x)(\{\ell\})\bigr)_{\ell\in L_\eps}
\)
and the family \(\bigl(y_\eps(\ell)\bigr)_{\ell\in L_\eps}\), we obtain
\[
  \delta_{\Psi_{\beta_\eps}(x)}
  \leq
  \sum_{\ell\in L_\eps}
    p_\eps(x)(\{\ell\})\,\delta_{y_\eps(\ell)}
  \leq
  \delta_{\Psi_{\alpha_\eps}(x)}.
\]
Since
\[
  \kappa_\eps(x)
  =
  \sum_{\ell\in L_\eps}
    p_\eps(x)(\{\ell\})\,\delta_{y_\eps(\ell)},
\]
this proves \cref{eq:kappa-sandwich}.

Fix \(a\in A\) and a nonempty \(U\in\Up(P)\), and write
\[
  p_\eps(x)
  =
  \sum_{\ell\in L_\eps}r_\ell\delta_\ell,
  \qquad
  r_\ell\geq0,
  \qquad
  \sum_{\ell\in L_\eps}r_\ell=1.
\]
For every \(\ell\) with \(r_\ell>0\), the preceding sandwich and the
monotonicity of \(\lambda_{a,U}\) give
\[
  \lambda_{a,U}\bigl(\Psi_{\beta_\eps}(x)\bigr)
  \leq
  \lambda_{a,U}\bigl(y_\eps(\ell)\bigr)
  \leq
  \lambda_{a,U}\bigl(\Psi_{\alpha_\eps}(x)\bigr).
\]
Taking the convex combination with coefficients \(r_\ell\), and using
the linearity of \(\lambda_{a,U}\), yields
\[
\begin{aligned}
  \lambda_{a,U}\bigl(\Psi_{\beta_\eps}(x)\bigr)
  &=
  \sum_{\ell\in L_\eps}
  r_\ell\,
  \lambda_{a,U}\bigl(\Psi_{\beta_\eps}(x)\bigr)\\
  &\leq
  \sum_{\ell\in L_\eps}
  r_\ell\,
  \lambda_{a,U}\bigl(y_\eps(\ell)\bigr)\\
  &=
  \lambda_{a,U}\left(
    \sum_{\ell\in L_\eps}
    r_\ell y_\eps(\ell)
  \right)\\
  &=
  \lambda_{a,U}\bigl(d_\eps(x)\bigr)\\
  &\leq
  \sum_{\ell\in L_\eps}
  r_\ell\,
  \lambda_{a,U}\bigl(\Psi_{\alpha_\eps}(x)\bigr)\\
  &=
  \lambda_{a,U}\bigl(\Psi_{\alpha_\eps}(x)\bigr).
\end{aligned}
\]

Since the coordinates \(\lambda_{a,U}\), with
\(a\in A\) and \(\varnothing\neq U\in\Up(P)\), determine the order on
\(K\), it follows that
\[
  \Psi_{\beta_\eps}(x)
  \leq
  d_\eps(x)
  \leq
  \Psi_{\alpha_\eps}(x),
\]
which is \cref{eq:d-sandwich}.

Finally,
\(
  \kappa_\eps=\Vone(y_\eps)\circ p_\eps
\)
is the asserted finite-kernel factorization, while
\(
  d_\eps=e_\eps\circ p_\eps
\)
is the asserted factorization through \(\Vsub(L_\eps)\).
\end{proof}

We then arrange these kernels into an increasing sequence.

For a single value of \(\eps\), the preceding proposition gives only a
one-step approximation. We now choose a geometric sequence of scales so that
the upper bound at level $n$ lies below the lower bound at
level $n+1$. This produces genuinely increasing approximations rather than
merely approximations converging in Euclidean distance.

\begin{proposition}\label{prop:finite-kernels}
There are finite posets $L_n$ and Scott-continuous maps
\[
  p_n:K\to\Vone(L_n),
  \qquad
  y_n:L_n\to K,
  \qquad
  e_n:\Vsub(L_n)\to K
\]
such that, with
\[
  \kappa_n=\Vone(y_n)\circ p_n,
  \qquad
  d_n=e_n\circ p_n,
\]
one has
\begin{align}
  \kappa_n&\leq\kappa_{n+1}\leq\etaV_K, \quad
  \sup_n\kappa_n=\etaV_K,
  \label{eq:finite-kernel-limit}\\
  d_n&\leq d_{n+1}\leq\id_K,\quad
  \sup_nd_n=\id_K.
  \label{eq:finite-barycentric-limit}
\end{align}
Each $\kappa_n$ factors through $\Vone(L_n)$, and each $d_n$ is a
finite-valuation approximant.
\end{proposition}
\begin{proof}
Choose constants
\[
  0<\eps_0<\frac12
  \qquad\text{and}\qquad
  0<\theta<
  \min\left\{
    \frac12,\frac1{K_PC_*}
  \right\}.
\]
For every $n\geq0$, set
\[
  \eps_n=\eps_0\theta^n,
  \qquad
  L_n=L_{\eps_n},
  \qquad
  p_n=p_{\eps_n},
  \qquad
  y_n=y_{\eps_n},
  \qquad
  e_n=e_{\eps_n},
\]
and let
\[
  \kappa_n=\kappa_{\eps_n},
  \qquad
  d_n=d_{\eps_n}.
\]
Also write
\[
  \alpha_n=\frac{\eps_n}{K_P},
  \qquad
  \beta_n=C_*\eps_n.
\]
Since $\eps_{n+1}=\theta\eps_n$ and
$\theta<1/(K_PC_*)$, we have
\[
  \beta_{n+1}
  =
  C_*\eps_{n+1}
  =
  C_*\theta\eps_n
  <
  \frac{\eps_n}{K_P}
  =
  \alpha_n.
\]
The family $(\Psi_t)_{t\geq0}$ is decreasing in $t$; hence
\(
  \Psi_{\alpha_n}(x)
  \leq
  \Psi_{\beta_{n+1}}(x)
  \) for every \(x\in K\).
By
\cref{eq:kappa-sandwich,eq:d-sandwich}, for every \(n\) and every $x\in K$,
\[
  \Dirac_{\Psi_{\beta_n}(x)}
  \leq
  \kappa_n(x)
  \leq
  \Dirac_{\Psi_{\alpha_n}(x)}
  \leq
  \Dirac_x
\]
and
\[
  \Psi_{\beta_n}(x)
  \leq
  d_n(x)
  \leq
  \Psi_{\alpha_n}(x)
  \leq
  x.
\]
Combining these inequalities with
\(
  \Psi_{\alpha_n}(x)
  \leq
  \Psi_{\beta_{n+1}}(x)
\)
gives
\[
\begin{aligned}
  \kappa_n(x)
  &\leq
  \Dirac_{\Psi_{\alpha_n}(x)}
  \leq
  \Dirac_{\Psi_{\beta_{n+1}}(x)}
  \leq
  \kappa_{n+1}(x)
  \leq
  \Dirac_x,\\
  d_n(x)
  &\leq
  \Psi_{\alpha_n}(x)
  \leq
  \Psi_{\beta_{n+1}}(x)
  \leq
  d_{n+1}(x)
  \leq
  x.
\end{aligned}
\]
Thus $(\kappa_n)_n$ and $(d_n)_n$ are pointwise increasing, with
\(
  \kappa_n\leq\etaV_K
  \ \text{and}\ 
  d_n\leq\id_K
\)
for every \(n\).

We next identify their pointwise suprema. Since
\[
  \beta_n=C_*\eps_0\theta^n\longrightarrow0,
\] we have
\[
  \sup_n\Psi_{\beta_n}(x)=x
  \qquad(x\in K).
\]
Using the lower half of \cref{eq:d-sandwich}, we have
\(
  \Psi_{\beta_n}(x)
  \leq
  d_n(x)
  \leq
  x.
\)
Taking suprema over $n$ yields
\[
  x
  =
  \sup_n\Psi_{\beta_n}(x)
  \leq
  \sup_n d_n(x)
  \leq
  x.
\]
Therefore
\[
  \sup_n d_n(x)=x
  \qquad(x\in K),
\]
and hence
\(
  \sup_n d_n=\id_K
\)
pointwise.

Similarly, \cref{eq:kappa-sandwich} gives
\[
  \etaV_K(\Psi_{\beta_n}(x))
  \leq
  \kappa_n(x)
  \leq
  \etaV_K(x).
\]
Since the Dirac unit
\[
  \etaV_K:K\longrightarrow\Vone(K),
  \qquad
  x\longmapsto\Dirac_x,
\]
is Scott-continuous, it preserves the directed supremum
\(
  \sup_n\Psi_{\beta_n}(x)=x.
\)
Consequently,
\[
\begin{aligned}
  \sup_n\etaV_K(\Psi_{\beta_n}(x))
  =
  \etaV_K\left(\sup_n\Psi_{\beta_n}(x)\right)
  =
  \etaV_K(x).
\end{aligned}
\]
Taking suprema gives
\[
  \etaV_K(x)
  \leq
  \sup_n\kappa_n(x)
  \leq
  \etaV_K(x),
\]
and hence
\[
  \sup_n\kappa_n(x)=\etaV_K(x)
  \qquad(x\in K).
\]
Thus
\(
  \sup_n\kappa_n=\etaV_K
\)
pointwise.

Finally, by the definitions in
\cref{eq:kappa-eps,eq:e-eps}, each $\kappa_n$ and $d_n$ factors through a valuation space over the finite poset $L_n$:
\[
  \kappa_n
  =
  \Vone(y_n)\circ p_n,
  \qquad
  d_n
  =
  e_n\circ p_n.
\]
Hence $(\kappa_n)_n$ is an increasing finite-kernel approximation of
$\etaV_K$, and $(d_n)_n$ is an increasing finite-valuation approximate
identity on $K$.
\end{proof}

\begin{theorem}\label{thm:finite-polytope}
For all finite posets $A$ and $P$,
\(
  \M(A,P)\in\FVA.
\)
\end{theorem}

\begin{proof}
If $A=\varnothing$ or $P=\varnothing$, the dcpo is a singleton and hence
belongs to $\FVA$ by \cref{prop:C-saturation}.  Otherwise,
\cref{lem:M-domain,prop:finite-kernels} gives a finite-valuation approximate
identity $(d_n)$.
\end{proof}

The finite-kernel construction is now complete. Every finite monotone-valuation
polytope has an increasing approximate identity through valuation spaces of
finite posets, and it also carries finite kernels converging increasingly to
the Dirac unit. The latter, stronger statement is the input for the lifting
arguments in the next subsection.

\subsection[Valuation powerdomain closure]{Closure under Valuation Powerdomains}
\label{sec:valuation-closure}

The finite-kernel theorem of the preceding subsection has two distinct
applications.  In this subsection we use only the first one: a finite
probability-kernel approximation of the Dirac unit is lifted, by Kleisli
extension, to an approximate identity on a valuation powerdomain.  The
function-space application is deferred to
\cref{sec:function-spaces-ccc}.

We first record the finite-state formulas used in the valuation lifting.

For every dcpo \(Z\), let
\[
  \iota_Z:\Vone(Z)\hookrightarrow\Vsub(Z)
\]
denote the canonical inclusion.  If \(L\) is a finite poset, every
\(\rho\in\Vsub(L)\) has a unique atomic representation
\[
  \rho=\sum_{\ell\in L}\rho_\ell\Dirac_\ell,
  \qquad
  \rho_\ell\geq0,
  \qquad
  \sum_{\ell\in L}\rho_\ell\leq1.
\]
Consequently, for every monotone map \(y:L\to Z\),
\begin{equation}\label{eq:finite-pushforward}
  \Vsub(y)(\rho)
  =
  \sum_{\ell\in L}\rho_\ell\Dirac_{y(\ell)}.
\end{equation}
If \(v:L\to\Vsub(E)\) is a valuation kernel, then
\begin{equation}\label{eq:finite-kleisli-sum}
  v^\dagger(\rho)
  =
  \sum_{\ell\in L}\rho_\ell v(\ell).
\end{equation}
Thus, on a finite source, Kleisli extension is simply the barycentric
combination of the values of the kernel.  These formulas will also be used
in the function-space reconstruction of
\cref{sec:function-spaces-ccc}.

The following is the local continuity of Kleisli extension needed below.

\begin{lemma}\label{lem:kleisli-local}
                                                             Let \((k_i)_{i\in I}\) be a directed family of Scott-continuous kernels
\[
  k_i:D\longrightarrow\Vsub(E),
\]
with pointwise supremum \(k\).  Then
\[
  k^\dagger=\sup_i k_i^\dagger.
\]
If \(E=D\) and \(k_i\leq\etaV_D\) for every \(i\), then
\[
  k_i^\dagger\leq\id_{\Vsub(D)}.
\]
\end{lemma}

\begin{proof}
Fix \(\nu\in\Vsub(D)\) and a Scott-open set \(U\subseteq E\).  Put
\[
  g_i(x)=k_i(x)(U),
  \qquad
  g(x)=k(x)(U).
\]
Then \(g=\sup_i g_i\) pointwise. \Cref{thm:jones-monotone-convergence}
for integration with respect to continuous valuations
gives
\[
  (k^\dagger\nu)(U)
  =
  \int_D g\,d\nu
  =
  \sup_i\int_D g_i\,d\nu
  =
  \sup_i(k_i^\dagger\nu)(U).
\]
Since directed suprema in \(\Vsub(E)\) are computed pointwise on Scott-open
sets, \(k^\dagger\nu=\sup_i k_i^\dagger\nu\), and hence
\(k^\dagger=\sup_i k_i^\dagger\).      
                                          
If \(E=D\) and \(k_i\leq\etaV_D\), monotonicity of Kleisli extension and
the unit law in \cref{lem:valuation-kleisli-laws}\textup{(i)--(ii)} give
\[
  k_i^\dagger
  \leq
  \etaV_D^\dagger
  =
  \id_{\Vsub(D)}.
\]
\end{proof}

We now lift the finite kernels to valuation powerdomains and derive global closure under \(\Vsub\) and \(\Vone\).

\begin{theorem}\label{thm:valuation-kernel-lifting}
Let \(X\) be an FS-domain.  Suppose that there are finite posets \(L_n\)
and Scott-continuous maps
\[
  p_n:X\longrightarrow\Vone(L_n),
  \qquad
  y_n:L_n\longrightarrow X
\]
such that
\[
  \kappa_n=\Vone(y_n)\circ p_n,
  \qquad
  \kappa_n\leq\kappa_{n+1}\leq\etaV_X,
  \qquad
  \sup_n\kappa_n=\etaV_X.
\]
Then
\[
  \Vsub(X)\in\FVA.
\]
\end{theorem}

\begin{proof}
For every \(n\), put \[
  \bar p_n=\iota_{L_n}\circ p_n:
  X\longrightarrow\Vsub(L_n)
\]
and
\[
  \bar\kappa_n=\iota_X\circ\kappa_n:
  X\longrightarrow\Vsub(X).
\]
Naturality of the inclusion gives
\[
  \bar\kappa_n
  =
  \Vsub(y_n)\circ\bar p_n.
\]
Define
\[
  T_n=\bar\kappa_n^\dagger:
  \Vsub(X)\longrightarrow\Vsub(X).
\]
By \cref{lem:kleisli-local},
\[
  T_n\leq T_{n+1}\leq\id_{\Vsub(X)}
  \qquad\text{and}\qquad
  \sup_nT_n
  =
  \etaV_X^\dagger
  =
  \id_{\Vsub(X)}.
\]
\Cref{lem:valuation-kleisli-laws}\textup{(iv)}, which follows from Kleisli associativity and the unit laws, gives
\[
  T_n
  =
  \bigl(\Vsub(y_n)\circ\bar p_n\bigr)^\dagger
  =
  \Vsub(y_n)\circ\bar p_n^\dagger.
\]
Hence \(T_n\) factors as
\[
  \Vsub(X)
  \xrightarrow{\ \bar p_n^\dagger\ }
  \Vsub(L_n)
  \xrightarrow{\ \Vsub(y_n)\ }
  \Vsub(X).
\]
Since \(L_n\) is finite, this is a finite-valuation factorization.
Thus \((T_n)_n\) is a finite-valuation approximate identity on
\(\Vsub(X)\).
\end{proof}

\begin{corollary}\label{cor:finite-generator-valuation}
For every finite poset \(Q\),
\[
  \Vsub(\Vsub(Q))\in\FVA.
\]
\end{corollary}

\begin{proof}
If \(Q=\varnothing\), then \(\Vsub(Q)\) is the one-point dcpo and
\[
  \Vsub(\Vsub(Q))
  \cong
  \Vsub(\mathbf 1)
  \in\FVA
\]
by \cref{lem:generators}.  Suppose \(Q\neq\varnothing\).  Since
\[
  \Vsub(Q)\cong\M(\mathbf 1,Q),
\]
\cref{prop:finite-kernels} supplies a finite-kernel approximation of its
Dirac unit.  Apply \cref{thm:valuation-kernel-lifting}.
\end{proof}

\begin{lemma}\label{lem:valuation-retract}
    For every pointed dcpo $D$, $\Vone(D)$ is a Scott-continuous retract of $\Vsub(D)$.
\end{lemma}

\begin{proof}
Let \(\bot\) be the least element of \(D\).  Let
\[
\iota_D:\Vone(D)\hookrightarrow\Vsub(D)
\]
be the inclusion, and define
\[
  N_D:\Vsub(D)\longrightarrow\Vone(D),
  \qquad
  N_D(\nu)
  =
  \nu+\bigl(1-\nu(D)\bigr)\Dirac_\bot.
\]
The valuation \(N_D(\nu)\) has total mass one.  If
\(U\subsetneq D\) is Scott open, then \(\bot\notin U\), and therefore
\[
  N_D(\nu)(U)=\nu(U),
\]
whereas \(N_D(\nu)(D)=1\).  These formulas show that \(N_D\) is
Scott-continuous and that
\[
  N_D\circ\iota_D=\id_{\Vone(D)}.
\]
Thus \(\Vone(D)\) is a Scott-continuous retract of \(\Vsub(D)\).  
\end{proof}

Applying \cref{cor:finite-generator-valuation}, \cref{lem:valuation-retract}, \cref{prop:V-close} and \cref{thm:FVA-property}, we obtain that $\FVA$ is closed under valuation monads $\Vsub$ and $\Vone$.

\begin{theorem}\label{thm:global-Vsub}
If \(D\in\FVA\), then
\[
  \Vsub(D)\in\FVA, \ \ \Vone(D) \in\FVA.
\]
\end{theorem}
\begin{proof}
For every finite poset \(P\),
\cref{cor:finite-generator-valuation} gives
\(
\Vsub(\Vsub(P))\in\FVA.
\)
Applying \cref{prop:V-close} to the generator class
\(\{\Vsub(P): P \text{ is a finite poset}\}\) yields 
\(
\Vsub(D)\in\FVA
\)
for every \(D\in\FVA\).
By \cref{thm:FVA-property}, every such \(D\) is pointed; hence
\cref{lem:valuation-retract} makes \(\Vone(D)\) a Scott-continuous
retract of \(\Vsub(D)\). Retract closure of \(\FVA\) therefore
gives \(\Vone(D)\in\FVA\).
\end{proof}

\begin{corollary}\label{cor:valuation-monads-restrict}
The endofunctors \(\Vsub\) and \(\Vone\) restrict to endofunctors on
\(\FVA\).  Their units and multiplications are morphisms of the full
subcategory \(\FVA\), so both valuation monads restrict to \(\FVA\).
\end{corollary}

\begin{proof}
By \cref{lem:valuation-kleisli-laws}, the Dirac unit and Kleisli extension satisfy the unit and associativity laws in \(\DCPO\). The mass calculation following that lemma shows that these constructions also restrict to probabilistic valuations. Since \(\FVA\) is full in \(\DCPO\), the Scott-continuous units and multiplications between its objects are morphisms of \(\FVA\), and the same monad laws hold in this subcategory.
\end{proof}

\subsection[Function spaces and Cartesian closedness]{Function Spaces and Cartesian Closedness}
\label{sec:function-spaces-ccc}

We now use the same finite kernels in a different way.  Instead of
integrating them against an input valuation, we use them to sample the
argument of a Scott-continuous function at finitely many labels and then
reconstruct the function by barycentric integration.  This is the essential
step in the proof of Cartesian closedness.

We now apply sampling and barycentric reconstruction to function spaces.

\begin{theorem}\label{thm:function-kernel-lifting}
Let \(X\) be an FS-domain.  Suppose that there are finite posets \(L_n\)
and Scott-continuous maps
\[
  p_n:X\longrightarrow\Vone(L_n),
  \qquad
  y_n:L_n\longrightarrow X
\]
such that
\[
  \kappa_n=\Vone(y_n)\circ p_n,
  \qquad
  \kappa_n\leq\kappa_{n+1}\leq\etaV_X,
  \qquad
  \sup_n\kappa_n=\etaV_X.
\]
Then, for every finite poset \(P\),
\[
  [X\to\Vsub(P)]\in\FVA.
\]
\end{theorem}

\begin{proof}
The assertion is immediate when \(P=\varnothing\), so assume that
\(P\neq\varnothing\), and put
\[
  Y=\Vsub(P).
\]
Both \(X\) and \(Y\) are FS-domains; hence their function space
\([X\to Y]\) is an FS-domain
\cite[Proposition~II-2.18]{GierzEtAl2003}.

For every \(n\), put
\[
  B_n=\M(L_n,P)=\Mon(L_n,Y).
\]
An element of \(B_n\) is a finite monotone table assigning a valuation in
\(Y\) to each state of \(L_n\).  By \cref{thm:finite-polytope},
\(
  B_n\in\FVA.
\)
Let
\[
  \bar p_n=\iota_{L_n}\circ p_n:
  X\longrightarrow\Vsub(L_n)
\]
and
\[
  \bar\kappa_n=\iota_X\circ\kappa_n:
  X\longrightarrow\Vsub(X).
\]
Define the sampling map
\[
  P_n:[X\to Y]\longrightarrow B_n,
  \qquad
  P_n(f)=f\circ y_n,
\]
and the reconstruction map
\[
  E_n:B_n\longrightarrow[X\to Y],
  \qquad
  E_n(v)=v^\dagger\circ\bar p_n.
\]
For \(x\in X\), formula~\eqref{eq:finite-kleisli-sum} reads
\begin{equation}\label{eq:function-reconstruction-sum}
  E_n(v)(x)
  =
  \sum_{\ell\in L_n}
  p_n(x)_\ell\,v(\ell).
\end{equation}
Thus \(P_n\) records finitely many values of \(f\), while \(E_n\) reconstructs
a function by the probability distribution \(p_n(x)\).

The map \(P_n\) is Scott-continuous because directed suprema in function
spaces are computed pointwise.  The map \(E_n\) is Scott-continuous by
\cref{lem:kleisli-local}, applied to directed families of kernels
\(L_n\to Y\).
Put
\[
  A_n=E_n\circ P_n.
\]
For \(f\in[X\to Y]\) and \(x\in X\), the Kleisli associativity law gives
\begin{align}
  A_n(f)(x)
  &=
  (f\circ y_n)^\dagger(\bar p_n(x))\notag\\
  &=
  f^\dagger\bigl(\Vsub(y_n)(\bar p_n(x))\bigr)\notag\\
  &=
  f^\dagger(\bar\kappa_n(x)).
  \label{eq:function-kernel-approximation}
\end{align}
Since
\[
  \bar\kappa_n
  \leq
  \bar\kappa_{n+1}
  \leq
  \etaV_X
\]
and \(f^\dagger\) is monotone, the unit law in
\cref{lem:valuation-kleisli-laws}\textup{(ii)} gives
\[
  A_n(f)(x)
  \leq
  A_{n+1}(f)(x)
  \leq
  f^\dagger(\etaV_X(x))
  =
  f(x).
\]
Hence
\[
  A_n\leq A_{n+1}\leq\id_{[X\to Y]}.
\]
Moreover, Scott continuity of \(f^\dagger\) and the unit law
(\cref{lem:valuation-kleisli-laws}\textup{(i)--(ii)}), together with \(\sup_n\bar\kappa_n=\etaV_X\), give
\[
\begin{aligned}
  \sup_nA_n(f)(x)
  &=
  f^\dagger\left(\sup_n\bar\kappa_n(x)\right)\\
  &=
  f^\dagger(\etaV_X(x))\\
  &=
  f(x).
\end{aligned}
\]
Thus \((A_n)_n\) is an increasing approximate identity on \([X\to Y]\).
Each \(A_n=E_nP_n\) factors through \(B_n\in\FVA\).  Applying
\cref{thm:FVA-property}\textup{(2)} proves the result.
\end{proof}

\begin{corollary}\label{cor:finite-generator-function}
For all finite posets \(P,Q\),
\(
  [\Vsub(Q)\to\Vsub(P)]\in\FVA.
\)
\end{corollary}

\begin{proof}
If \(Q=\varnothing\), then \(\Vsub(Q)\) is the terminal dcpo and
\[
  [\Vsub(Q)\to\Vsub(P)]
  \cong
  \Vsub(P)
  \in\FVA.
\]
Suppose \(Q\neq\varnothing\).  Since
\(
  \Vsub(Q)\cong\M(\mathbf1,Q),
\)
\cref{prop:finite-kernels} supplies a finite-kernel approximation of the
Dirac unit on \(\Vsub(Q)\).  Apply
\cref{thm:function-kernel-lifting}.
\end{proof}

\begin{corollary}\label{cor:finite-generators}
For all finite posets \(P,Q\),
\[
  \Vsub(P)\times \Vsub(Q)\in\FVA,
  \qquad
  [\Vsub(Q)\to\Vsub(P)]\in\FVA.
\]
\end{corollary}

\begin{proof}
Apply \cref{prop:C-product}, \cref{prop:generator-product} and 
\cref{cor:finite-generator-function}.
\end{proof}

Now we have reached the main result of this paper that  \(\FVA\) is a solution of the category-existence form of the Jung--Tix problem in the following sense.

\begin{theorem}\label{thm:main}
Let \(\FVA\) be the full subcategory of \(\DCPO\) defined in Section~\ref{subsec:FVA}. Then:
\begin{enumerate}[label=\textup{(\roman*)}]
\item every object of \(\FVA\) is a pointed countably based FS-domain;
\item \(\Vsub(P)\in\FVA\) for every finite poset \(P\), and \(\FVA\) is saturated, i.e., $\mathcal
F(\FVA)=\FVA$, and is closed under Scott-continuous
      retracts;
\item if \(X,Y\in\FVA\), then \(X\times Y, \ [X\to Y]\in\FVA\);
\item if \(D\in\FVA\), then
      \(\Vsub(D),\Vone(D)\in\FVA\).
\end{enumerate}
Hence \(\FVA\) is a full Cartesian closed subcategory of \(\FS\), and both
valuation monads restrict to \(\FVA\).
\end{theorem}

\begin{proof}
Assertions \textup{(i)} and \textup{(ii)} follow from
\cref{thm:FVA-property}.  Assertion
\textup{(iii)} follows from  \cref{prop:C-product}, \cref{prop:C-function} and \cref{cor:finite-generators}.  Assertion
\textup{(iv)} follows from \cref{prop:V-close} and \cref{thm:global-Vsub}.

Since \(\FVA\) is full and closed under finite products and the function-space objects inherited from \(\DCPO\), it follows that \(\FVA\) is Cartesian closed by \cref{thm:C-Cartesian closed}.
The final assertion follows from
\cref{cor:valuation-monads-restrict}.
\end{proof}

\begin{remark}\label{rem:directed-FVA}
The use of a sequence in the definition of $\omega\mathbf{FVA}$ is not
intrinsic to factorization approximation; it serves only to retain
countable basedness. Once the countability requirement is dropped, the
approximating sequence can be replaced by an arbitrary directed family.
The preceding continuity and stability results, together with the lifting
results for the constructions considered above, remain valid in this
directed setting. Indeed, their proofs use only directed convergence and
the existence of a common upper bound for every finite collection of
approximants.

More precisely, let $\mathbf{FVA}$ denote the full replete subcategory of
$\DCPO$ whose objects are those dcpos $D$ admitting a directed family
\[
    (a_i)_{i\in I},
    \qquad
    a_i=e_i\circ p_i\colon D\longrightarrow D,
\]
of finite-valuation factorization approximants such that
\[
    a_i\leq \id_D
    \qquad\text{and}\qquad
    \sup_{i\in I}a_i=\id_D
\]
pointwise. The directed version of the preceding continuity theorem shows
that every such $D$ is a continuous dcpo. Moreover,
\[
    \omega\mathbf{FVA}\subseteq\mathbf{FVA},
\]
while no countability condition is imposed on objects of $\mathbf{FVA}$.

The preceding proofs extend to $\mathbf{FVA}$ with only minor
modifications. In the square-refinement argument of
\cref{lem:square}, one retains the directed family of refining squares and
omits the final extraction of a countable cofinal sequence. The saturation
argument is unchanged: every comparison involves only finitely many
approximants, and directedness supplies a common upper bound for their
indices.

For products, if $(a_i)_{i\in I}$ and $(b_j)_{j\in J}$ approximate the
identities on $X$ and $Y$, respectively, one uses the product-directed
family
\[
    (a_i\times b_j)_{(i,j)\in I\times J}.
\]
For function spaces, one uses
\[
    A_{i,j}\colon [X\to Y]\longrightarrow [X\to Y],
    \qquad
    A_{i,j}(f)=b_j\circ f\circ a_i,
    \qquad (i,j)\in I\times J,
\]
so no diagonalization or countable-cofinality argument is required. For
the subprobabilistic powerdomain, the required approximation is given by the
directed family
\[
    \bigl(\Vsub(a_i)\bigr)_{i\in I},
\]
and the probability case follows from the same missing-mass retraction
\[
    \Vone(D)\triangleleft\Vsub(D).
\]

Consequently, after omitting only the countable-basis conclusion, the same
finite-generator, saturation, and transfer arguments show that
$\mathbf{FVA}$ is a full Cartesian closed subcategory of continuous domains
and that
\[
    D\in\mathbf{FVA}
    \quad\Longrightarrow\quad
    \Vsub(D),\Vone(D)\in\mathbf{FVA}.
\]
Hence the subprobabilistic and probabilistic valuation monads both restrict to
$\mathbf{FVA}$.
\end{remark}

\section{Comparison with bc-domains and RB-domains}
\label{sec:category-size}

Recall that a \emph{bc-domain} is a pointed domain in which every bounded subset has a supremum. We write $\BC$ for the category of bc-domains and $\omega\BC$ for the category of countably based bc-domains.
An \emph{RB-domain} is a Scott-continuous retract of a bifinite domain;
equivalently, its identity is the directed supremum of finite-image
Scott-continuous self-maps below the identity. We write $\RB$ for this full category. Let $\omega\mathbf{FS}_{\bot}$ (resp. $\omega\mathbf{RB}_{\bot}$ ) denote the full subcategory of pointed countably based FS-domains (resp. RB-domains). 

\subsection{Countably based bc-domains inside \(\FVA\)}
\label{subsec:size-bc}

\begin{lemma}\label{lem:finite-bc-fva}
Every finite bc-domain belongs to $\FVA$.
\end{lemma}

\begin{proof}
Let $B$ be a finite bc-domain, let $N=|B|$, and put
\[
  \tau=1-\frac{1}{N+1}.
\]
Since $B$ has a least element, \cref{prop:probability-positive,lem:generators}
gives
\[
  \Vone(B)\cong\Vsub(B\setminus\{\bot\})\in\FVA.
\]
For $\mu\in\Vone(B)$, define
\[
  S(\mu)=\{b\in B:\mu(\up b)>\tau\}.
\]
Viewing $\mu$ as a probability vector on the finite set $B$, one has
\[
\begin{aligned}
  \mu\left(\bigcap_{b\in S(\mu)}\up b\right)
  \geq
  1-\sum_{b\in S(\mu)}\bigl(1-\mu(\up b)\bigr)
  >
  1-\frac{|S(\mu)|}{N+1}
  \geq\frac{1}{N+1}>0.
\end{aligned}
\]
Thus $S(\mu)$ has a common upper bound. Since $B$ is bounded complete, the
join $\sup S(\mu)$ exists; for $S(\mu)=\varnothing$ it is understood to
be $\bot$. Define
\[
  r_B:\Vone(B)\longrightarrow B,
  \qquad
  r_B(\mu)=\sup S(\mu).
\]
If $\mu\stle\nu$, then $S(\mu)\subseteq S(\nu)$, and hence
$r_B(\mu)\leq r_B(\nu)$. Moreover, if $(\mu_i)_{i\in I}$ is directed with
supremum $\mu$, then, for every $b\in B$,
\[
  \mu(\up b)=\sup_i\mu_i(\up b),
\]
so
\[
  S(\mu)=\bigcup_i S(\mu_i).
\]
Consequently,
\[
  r_B(\mu)
  =\sup\bigcup_iS(\mu_i)
  =\sup_i r_B(\mu_i),
\]
and $r_B$ is Scott-continuous. Finally,
\[
  S(\Dirac_x)=\down x,
  \qquad
  r_B(\Dirac_x)=x
  \qquad(x\in B).
\]
Thus $r_B\circ\etaV_B=\id_B$, so $B$ is a Scott-continuous retract of
$\Vone(B)$. The conclusion follows from \cref{thm:FVA-property}\textup{(2)}.
\end{proof}

\begin{proposition}\label{prop:omega-bc-fva}
Every countably based bc-domain belongs to $\FVA$.
\end{proposition}

\begin{proof}
Let $D$ be a countably based bc-domain. By the standard inclusion
$\BC\subseteq\RB$, there is a directed
family $(f_i)_{i\in I}$ of finite-image Scott-continuous maps such that
\[
  f_i\leq\id_D,
  \qquad
  \sup_i f_i=\id_D.
\]
We first replace this family by an increasing sequence. Let $B_0$ be a
countable basis of $D$, and enumerate all pairs
\[
  (b,c)\in B_0\times B_0
  \qquad\text{with}\qquad
  b\waybelow c.
\]
For every such pair, the equality $c=\sup_i f_i(c)$ yields an index $i$ with
$b\leq f_i(c)$. Using directedness, choose recursively an increasing sequence
$(g_n)_n$ from the family $(f_i)_i$ that satisfies the first $n$ of these
requirements. If $y\waybelow x$, interpolation and the basis property give
$b,c\in B_0$ such that
\[
  y\leq b\waybelow c\leq x.
\]
For all sufficiently large $n$,
\[
  y\leq b\leq g_n(c)\leq g_n(x).
\]
It follows that
\[
  g_n\leq g_{n+1}\leq\id_D,
  \qquad
  \sup_n g_n=\id_D.
\]

Fix $n$ and let $F_n=g_n[D]$. A bc-domain has all nonempty infima: the
infimum of a nonempty set is the supremum of its set of lower bounds. Let
$C_n$ be the closure of the finite set $F_n$ under nonempty infima. Then
$C_n$ is finite, contains $\bot$, and is closed under finite nonempty meets.
Hence $C_n$ is a finite bc-domain: if $A\subseteq C_n$ has an upper bound in
$C_n$, then
\[
  \sup_{C_n}A
  =\bigwedge\{c\in C_n:a\leq c\text{ for every }a\in A\}.
\]
By \cref{lem:finite-bc-fva}, $C_n\in\FVA$.

Regard $g_n$ as a map $p_n:D\to C_n$, and let
$e_n:C_n\hookrightarrow D$ be the inclusion. To see that $p_n$ is
Scott-continuous, let $A\subseteq D$ be directed. The directed set
$g_n[A]\subseteq F_n$ is finite and therefore has a largest element $m$.
Since $g_n:D\to D$ is Scott-continuous,
\[
  p_n(\sup A)=g_n(\sup A)=\sup g_n[A]=m=\sup_{C_n}p_n[A].
\]
The inclusion $e_n$ is Scott-continuous because every directed subset of the
finite poset $C_n$ has a largest element, which is its supremum both in $C_n$
and in $D$. Moreover,
\[
  g_n=e_n\circ p_n.
\]
Thus the increasing approximate identity $(g_n)_n$ factors through objects
$C_n\in\FVA$. Applying \cref{thm:FVA-property}\textup{(2)} gives
$D\in\FVA$.
\end{proof}

\subsection[Incomparability with RB-domains]{Convex connectedness and incomparability with RB-domains}
\label{subsec:size-finite-retracts}\label{subsec:size-RB}

\begin{lemma}\label{lem:finite-FVA-retract}
If a finite domain $D$ belongs to $\FVA$, then $D$ is a
Scott-continuous retract of $\Vsub(P)$ for some finite poset $P$.
\end{lemma}

\begin{proof}
Let $(a_n)_n$ be a finite-valuation approximate identity on $D$, with
\[
  a_n=e_np_n,
  \qquad
  D\xrightarrow{p_n}\Vsub(P_n)\xrightarrow{e_n}D.
\]
For each $x\in D$, the increasing sequence $(a_n(x))_n$ has supremum $x$.
Since $D$ is finite, it is eventually equal to $x$. As $D$ has only finitely
many elements, there is one index $N$ such that
\[
  a_N(x)=x
  \qquad(x\in D).
\]
Thus $e_Np_N=\id_D$, and the displayed factorization at index $N$ is the
required Scott-continuous retraction.
\end{proof}

\begin{lemma}\label{lem:valuation-convex-connected}
Let $P$ be a finite poset. The Scott topology on $\Vsub(P)$ is coarser than
the relative Euclidean topology. Consequently, every convex subset of
$\Vsub(P)$ is connected in its relative Scott topology.
\end{lemma}

\begin{proof}
Let $O\subseteq\Vsub(P)$ be Scott open and let $\mu\in O$. For $n\geq1$,
put
\[
  \mu_n=(1-2^{-n})\mu.
\]
Then $\mu_n\uparrow\mu$, so $\mu_n\in O$ for some $n$. Since $O$ is an
upper set,
\[
  \up\mu_n\subseteq O.
\]
The stochastic order on $\Vsub(P)$ is determined by upper-set coordinates:
\[
  \nu\in\up\mu_n
  \quad\Longleftrightarrow\quad
  \nu(U)\geq\mu_n(U)
  \text{ for every }U\in\Up(P).
\]
If $\mu(U)>0$, then
\[
  \mu(U)>\mu_n(U),
\]
whereas if $\mu(U)=0$, the inequality
$\nu(U)\geq\mu_n(U)=0$ is automatic. Since $\Up(P)$ is finite and every map
$\nu\mapsto\nu(U)$ is linear, $\up\mu_n$ contains a relative Euclidean
neighborhood of $\mu$. Hence every Scott-open subset of $\Vsub(P)$ is
relatively Euclidean open.

A convex subset of a real vector space is Euclidean connected. Its relative
Scott topology is coarser than its relative Euclidean topology, and is
therefore connected as well.
\end{proof}

\begin{proposition}\label{prop:B5-not-FVA}
Let
\[
  B_5=\{\bot,a,b,c,d\},
\]
where
\[
  \bot<a<c,d,
  \qquad
  \bot<b<c,d,
  \qquad
  a\parallel b,
  \qquad
  c\parallel d,
\]
and there are no further comparabilities. Then
\[
  B_5\in\RB\setminus\FVA.
\]
\end{proposition}

\begin{proof}
The poset $B_5$ is finite and pointed. Hence its identity is an idempotent
finite-image deflation, and therefore $B_5\in\RB$.

Suppose, towards a contradiction, that $B_5\in\FVA$. By
\cref{lem:finite-FVA-retract}, there are a finite poset $P$ and
Scott-continuous maps
\[
  B_5\xrightarrow{i}\Vsub(P)\xrightarrow{r}B_5,
  \qquad
  r\circ i=\id_{B_5}.
\]
Consider the set of common upper bounds of $i(a)$ and $i(b)$,
\[
  H=
  \{\nu\in\Vsub(P):i(a)\leq\nu\text{ and }i(b)\leq\nu\}.
\]
The set $H$ is convex: each of its defining conditions is a finite family of
linear inequalities in the upper-set coordinates. It is therefore connected
in its relative Scott topology by
\cref{lem:valuation-convex-connected}.

For every $\nu\in H$, monotonicity of $r$ gives
\[
  a=r(i(a))\leq r(\nu),
  \qquad
  b=r(i(b))\leq r(\nu).
\]
The only common upper bounds of $a$ and $b$ in $B_5$ are $c$ and $d$.
Consequently,
\[
  r[H]\subseteq\{c,d\}.
\]
Conversely, $i(c),i(d)\in H$ and
\[
  r(i(c))=c,
  \qquad
  r(i(d))=d,
\]
so
\[
  r[H]=\{c,d\}.
\]
Since $c$ and $d$ are incomparable maximal elements, $\{c,d\}$ is discrete,
and hence disconnected, in its relative Scott topology. This contradicts
the fact that the continuous image of the connected space $H$ under
$r|_H$ must be connected. Therefore $B_5\notin\FVA$.
\end{proof}

\begin{theorem}\label{thm:category-size-count}
We have 
\(
  \omega\BC
  \subsetneq
  \FVA
  \subsetneq
  \omega\mathbf{FS}_{\bot}.
\)
Moreover,  $\FVA$ and $\omega\RB_{\bot}$  are
incomparable:
\[
  \FVA\nsubseteq\omega\RB_{\bot},
  \qquad
  \omega\RB_{\bot}\nsubseteq\FVA.
\]
\end{theorem}

\begin{proof}
The first inclusion is \cref{prop:omega-bc-fva}, and the second follows from
\cref{thm:main}\textup{(i)}.

Let
\[
  D_4=\{\bot,a,b,\top\},
  \qquad
  \bot<a,b<\top,
  \qquad
  a\parallel b,
\]
be the four-element diamond. By
\cref{prop:probability-positive,lem:generators},
\[
  \Vone(D_4)
  \cong
  \Vsub(D_4\setminus\{\bot\})
  \in\FVA.
\]
On the other hand, the finite-poset classification of
\cite{ChenKouLyu2026} gives
\[
  \Vone(D_4)\notin\RB,
\]
thus $\Vone(D_4)\notin\omega\RB_{\bot}$, because the undirected Hasse graph of $D_4$ is not a tree. Since every
bc-domain is an RB-domain, this also proves
\[
  \omega\BC\subsetneq\FVA
  \qquad\text{and}\qquad
  \FVA\nsubseteq\omega\RB_{\bot}.
\]

Finally, \cref{prop:B5-not-FVA} gives
\[
  B_5\in\omega\RB_{\bot}\setminus\FVA,
\]
and therefore $\omega\RB_{\bot}\nsubseteq\FVA$.
\end{proof}

Without countable-basis restriction, the following is also true.
\begin{corollary}
\(
  \BC
  \subsetneq
  {\bf FVA}
  \subsetneq
  \mathbf{FS}_{\bot}.
\)
Moreover,  ${\bf FVA}$ and $\RB$  are
incomparable:
\[
  {\bf FVA}\nsubseteq\RB,
  \qquad
  \RB\nsubseteq{\bf FVA}.
\]
\end{corollary}

\section{Conclusion and further research directions}
\label{sec:conclusion}

We introduced the class \(\FVA\) of finite-valuation approximable domains and
proved that it is a full Cartesian closed subcategory of \(\DCPO\), contained
in the pointed countably based FS-domains and closed under both \(\Vsub\) and
\(\Vone\).  The proof combines two kinds of finite approximation.  On finite
posets, the erosion maps provide explicit order-preserving FS approximants.
At the categorical level, factorization and saturation pass finite valuation
data to general domains.  The finite-dimensional polytope construction then
produces monotone stochastic kernels, which are lifted separately to
valuation powerdomains by Kleisli extension and to function spaces by finite
sampling and barycentric reconstruction.

We note that the Jung--Tix problem for FS-domains remains open: that is, whether the category  of FS-domains is   closed under  $\Vsub$.

The factorization viewpoint suggests several directions for further work.  A
first problem is to seek intrinsic order-theoretic characterizations of
finite-valuation approximability and to clarify more precisely its relation
to the established classes of FS-, RB-, and bc-domains.  A second direction
is to determine how far the finite-kernel and saturation methods extend to
other powerdomain constructions and to further semantic structures arising
in probabilistic and higher-order computation.  These questions concern the
scope of the method developed here and do not affect the closure results
proved above.

\clearpage
\appendix
\renewcommand{\thetheorem}{\Alph{section}.\arabic{theorem}}
\aliascntresetthe{proposition}
\aliascntresetthe{lemma}
\aliascntresetthe{corollary}
\aliascntresetthe{definition}
\aliascntresetthe{example}
\aliascntresetthe{remark}
\section{A concrete construction of the erosion map on \(B_5\)}
\label{app:erosion-B5}
\label{subsec:construction-B5}
\setcounter{figure}{0}
\renewcommand{\thefigure}{\Alph{section}.\arabic{figure}}

We give the calculation of the erosion map from Section~\ref{sec:construction} on
the five-element poset $B_5$ as a concrete example.

\begin{figure}[htbp]
\centering
\begin{tikzpicture}[scale=1.1,
  every node/.style={font=\small},
  point/.style={circle,fill=black,inner sep=1.2pt}
]
  \node[point,label=below:\(\bot\)] (bot) at (0,0) {};
  \node[point,label=left:\(a\)]    (a)   at (-1.2,1.4) {};
  \node[point,label=right:\(b\)]   (b)   at (1.2,1.4) {};
  \node[point,label=left:\(c\)]    (c)   at (-1.2,2.8) {};
  \node[point,label=right:\(d\)]   (d)   at (1.2,2.8) {};

  \draw (bot) -- (a);
  \draw (bot) -- (b);
  \draw (a) -- (c);
  \draw (a) -- (d);
  \draw (b) -- (c);
  \draw (b) -- (d);
\end{tikzpicture}
\caption{The poset \(B_5\).}
\label{fig:B5}
\end{figure}

\begin{example}\label{ex:erosion-B5}
Let
\(
  B_5=\{\bot,a,b,c,d\},
\)
where
\[
  \bot<a<c,d,\qquad
  \bot<b<c,d,\qquad
  a\parallel b,\qquad
  c\parallel d,
\]
and there are no further comparabilities.  Thus \(n=|B_5|=5\).

Consider the subprobabilistic valuation
\[
  p=\frac1{12}\delta_c+\frac16\delta_d+\frac1{12}\delta_a+\frac1{12}\delta_b.
\]
Its total mass is \(5/12\), so \(p\in\Delta_{\le 1}(B_5)\).  This example is
chosen so that the first frontier consists of two maximal points with unequal
masses, while the third frontier consists of two maximal points with equal
masses.

We compute the recursion for \Cref{eq:recursive-interval} explicitly.

\medskip
\noindent
\textbf{Step 0.}  Set
\[
  q_0=p,\qquad s_0=0.
\]
Since
\[
  \supp(q_0)=\{a,b,c,d\},
\]
its maximal elements are
\[
  A_0=A(q_0)=\{c,d\}.
\]
Now
\[
  \down A_0=\down\{c,d\}=B_5,
\]
so
\[
  c_0=c(A_0)=(5+1)^{\,5-|\down A_0|}
  =6^{5-5}=1.
\]
Moreover,
\[
  (q_0)_c=\frac1{12},\qquad (q_0)_d=\frac16,
\]
hence
\[
  \tau_0=\min_{x\in A_0}\frac{(q_0)_x}{c_0}
  =\min\Bigl\{\frac1{12},\frac16\Bigr\}
  =\frac1{12}.
\]
Therefore, for \(0\le u\le\tau_0=\frac1{12}\),
\[
  \phi_p(s_0+u)
  =q_0-u(\delta_c+\delta_d).
\]
Equivalently,
\[
  \phi_p(u)
  =\Bigl(\frac1{12}-u\Bigr)\delta_c
   +\Bigl(\frac16-u\Bigr)\delta_d
   +\frac1{12}\delta_a+\frac1{12}\delta_b,
  \qquad 0\le u\le \frac1{12}.
\]
At the end of the first step,
\[
  s_1=s_0+\tau_0=\frac1{12},
\]
and
\[
  q_1=\phi_p(s_1)
  =\frac1{12}\delta_d+\frac1{12}\delta_a+\frac1{12}\delta_b.
\]
Thus the point \(c\) has disappeared, while \(d\) remains.

\medskip
\noindent
\textbf{Step 1.}  Now
\[
  \supp(q_1)=\{a,b,d\},
\]
so the unique maximal element is
\[
  A_1=A(q_1)=\{d\}.
\]
Since
\[
  \down A_1=\down\{d\}=\{\bot,a,b,d\},
\]
we obtain
\[
  c_1=c(A_1)=(5+1)^{\,5-|\down A_1|}
  =6^{5-4}=6.
\]
Also,
\[
  (q_1)_d=\frac1{12},
\]
so
\[
  \tau_1=\min_{x\in A_1}\frac{(q_1)_x}{c_1}
  =\frac{1/12}{6}
  =\frac1{72}.
\]
Hence, for \(0\le u\le\tau_1=\frac1{72}\),
\[
  \phi_p(s_1+u)=q_1-6u\,\delta_d.
\]
That is,
\[
  \phi_p\Bigl(\frac1{12}+u\Bigr)
  =\Bigl(\frac1{12}-6u\Bigr)\delta_d
   +\frac1{12}\delta_a+\frac1{12}\delta_b,
  \qquad 0\le u\le \frac1{72}.
\]
At the end of the second step,
\[
  s_2=s_1+\tau_1=\frac1{12}+\frac1{72}=\frac7{72},
\]
and
\[
  q_2=\phi_p(s_2)
  =\frac1{12}\delta_a+\frac1{12}\delta_b.
\]
Thus the second maximal point \(d\) has now disappeared.

\medskip
\noindent
\textbf{Step 2.}  We now have
\[
  \supp(q_2)=\{a,b\},
\]
so
\[
  A_2=A(q_2)=\{a,b\}.
\]
Since
\[
  \down A_2=\down\{a,b\}=\{\bot,a,b\},
\]
it follows that
\[
  c_2=c(A_2)=(5+1)^{\,5-|\down A_2|}
  =6^{5-3}=36.
\]
Because
\[
  (q_2)_a=(q_2)_b=\frac1{12},
\]
we get
\[
  \tau_2=\min_{x\in A_2}\frac{(q_2)_x}{c_2}
  =\frac{1/12}{36}
  =\frac1{432}.
\]
Therefore, for \(0\le u\le\tau_2=\frac1{432}\),
\[
  \phi_p(s_2+u)=q_2-36u(\delta_a+\delta_b).
\]
Equivalently,
\[
  \phi_p\Bigl(\frac7{72}+u\Bigr)
  =\Bigl(\frac1{12}-36u\Bigr)\delta_a
   +\Bigl(\frac1{12}-36u\Bigr)\delta_b,
  \qquad 0\le u\le \frac1{432}.
\]
At the end of the third step,
\[
  s_3=s_2+\tau_2=\frac7{72}+\frac1{432}=\frac{43}{432},
\]
and
\[
  q_3=\phi_p(s_3)=0.
\]

\medskip
\noindent
Thus the recursive construction terminates after three nonzero steps:
\[
  q_0
  \longrightarrow q_1
  \longrightarrow q_2
  \longrightarrow q_3=0.
\]
The first frontier is \(\{c,d\}\), but since
\[
  (q_0)_c=\frac1{12}\neq\frac16=(q_0)_d,
\]
the two maximal points \(c\) and \(d\) are not removed simultaneously:
\(c\) disappears in the first step, and \(d\) disappears only in the second
step.  By contrast, in the third step the frontier is \(\{a,b\}\) and
\[
  (q_2)_a=(q_2)_b=\frac1{12},
\]
so the two remaining maximal points are removed simultaneously.

For convenience, the data of the recursion are summarized in the table
\[
\begin{array}{c|c|c|c|c}
j & q_j & A_j & c_j & \tau_j\\
\hline
0 &
\dfrac1{12}\delta_c+\dfrac16\delta_d+\dfrac1{12}\delta_a+\dfrac1{12}\delta_b
& \{c,d\} & 1 & \dfrac1{12}\\[1.2ex]
1 &
\dfrac1{12}\delta_d+\dfrac1{12}\delta_a+\dfrac1{12}\delta_b
& \{d\} & 6 & \dfrac1{72}\\[1.2ex]
2 &
\dfrac1{12}\delta_a+\dfrac1{12}\delta_b
& \{a,b\} & 36 & \dfrac1{432}\\[1.2ex]
3 & 0 & \varnothing & - & -
\end{array}
\]
and the corresponding break points are
\[
  s_0=0,\qquad s_1=\frac1{12},\qquad
  s_2=\frac7{72},\qquad s_3=\frac{43}{432}.
\]

Hence \(\phi_p(t)=\Phi_t(p)\) is given explicitly by
\[
  \phi_p(t)=
  \begin{cases}
  \left(\dfrac1{12}-t\right)\delta_c
  +\left(\dfrac16-t\right)\delta_d
  +\dfrac1{12}\delta_a+\dfrac1{12}\delta_b,
  & 0\le t\le \dfrac1{12},\\[1.2ex]
  \left(\dfrac1{12}-6\left(t-\dfrac1{12}\right)\right)\delta_d
  +\dfrac1{12}\delta_a+\dfrac1{12}\delta_b,
  & \dfrac1{12}\le t\le \dfrac7{72},\\[1.6ex]
  \left(\dfrac1{12}-36\left(t-\dfrac7{72}\right)\right)\delta_a
  +\left(\dfrac1{12}-36\left(t-\dfrac7{72}\right)\right)\delta_b,
  & \dfrac7{72}\le t\le \dfrac{43}{432},\\[1.6ex]
  0,
  & t\ge \dfrac{43}{432}.
  \end{cases}
\]
\end{example}

\section{\(\FVA\) is closed under the extended probabilistic powerdomain}
\label{app:extended-valuations}

For a dcpo $D$, the \emph{extended probabilistic powerdomain}
$\Vext(D)$ consists of all continuous valuations
$\nu:\sigma(D)\to[0,\infty]$, without any restriction on total mass,
ordered pointwise. It is a dcpo, with directed suprema computed
pointwise on Scott-open sets~\cite{GoubaultLarrecqJiaTheron2023}.
In particular, valuations of infinite total mass are allowed.
We use $0\cdot\infty=0$ in nonnegative scalar multiplication.
This appendix proves that
$\Vext(D)\in\FVA$ whenever $D\in\FVA$. For finite posets, the
approximating maps are defined directly on all extended valuations by a
finite minimum over upper sets. The general result then follows from the
finite kernels of \cref{prop:finite-kernels} and the saturation property
of \cref{thm:FVA-property}\textup{(2)}.

\subsection{Finite posets}
\label{subsec:extended-finite-posets}

For a finite poset $P$, its Scott-open sets are precisely its upper sets.
For $0<R<\infty$, write
\[
  \mathcal V_{\leq R}(P)
  =\{\nu\in\Vext(P):\nu(P)\leq R\}.
\]
Directed suprema in both $\Vext(P)$ and
$\mathcal V_{\leq R}(P)$ are computed pointwise on upper sets.
Positive finite scalar multiplication preserves these suprema; in
particular, multiplication by $R$ is an order isomorphism from
$\Vsub(P)$ onto $\mathcal V_{\leq R}(P)$.
We use the convention $0\cdot\infty=0$.

We shall use one elementary identity. If $J$ is a nonempty finite set,
$I$ is a directed poset, and $t_{i,j}\in[0,\infty]$ is nondecreasing in
$i$ for each $j\in J$, then
\begin{equation}\label{eq:extended-finite-min-sup}
  \sup_{i\in I}\min_{j\in J}t_{i,j}
  =\min_{j\in J}\sup_{i\in I}t_{i,j}.
\end{equation}
Indeed, the inequality from left to right follows pointwise. Let $L$
be the right-hand side. If $L=0$, the reverse inequality is immediate.
Otherwise, for any finite $r$ with $0\leq r<L$, choose $i_j$ such that
$t_{i_j,j}>r$ for each $j$. There is a common upper bound $i$ for the
finitely many $i_j$, so $\min_jt_{i,j}>r$. Taking the supremum over
such $r$ proves the reverse inequality, also when $L=\infty$.
Finite sums likewise commute with a common directed supremum in
$[0,\infty]$: separate indices may first be chosen for the finitely
many summands and then replaced by a common upper bound.

\Needspace{7\baselineskip}
\begin{theorem}\label{thm:finite-extended-FVA}
For every finite poset $P$, the extended probabilistic powerdomain
$\Vext(P)$ belongs to $\FVA$.
\end{theorem}

\begin{proof}
If $P=\varnothing$, then $\Vext(P)=\Vsub(P)$ is the one-point
dcpo. Its identity factors through $\Vsub(P)$, so the constant
sequence of identity maps satisfies \cref{def:FVA}.
Assume henceforth that $P\neq\varnothing$, and put
\[
  n=|P|,
  \qquad d(p)=|\down p| \quad(p\in P).
\]
Since $\down p$ is a proper subset of $\down q$ whenever $p<q$,
\begin{equation}\label{eq:extended-rank-growth}
  1\leq d(p)\leq n,
  \qquad p<q\ \Longrightarrow\ d(q)\geq d(p)+1.
\end{equation}
All subsequent indices $m$ are integers satisfying $m>n$.

\emph{The preliminary maps.}
Define
\begin{equation}\label{eq:extended-simple-preprocessing}
  S_m\nu
  =\left(1-\frac nm\right)\nu
   +\frac1m\sum_{p\in P}\nu(\up p)\delta_p.
\end{equation}
Each summand is a continuous valuation, including when
$\nu(\up p)=\infty$. Explicitly, for $U\in\sigma(P)$,
\begin{equation}\label{eq:extended-preprocessing-open}
  (S_m\nu)(U)
  =\left(1-\frac nm\right)\nu(U)
   +\frac1m\sum_{p\in U}\nu(\up p).
\end{equation}
All coefficients are nonnegative, and $1-n/m>0$.
The formula shows that $S_m$ is monotone. If $(\nu_i)_{i\in I}$ is
directed with supremum $\nu$, pointwise computation of suprema and
commutation with the finite sum give
\[
  (S_m\nu)(U)
  =\left(1-\frac nm\right)\sup_i\nu_i(U)
    +\frac1m\sum_{p\in U}\sup_i\nu_i(\up p)
  =\sup_i(S_m\nu_i)(U).
\]
Thus $S_m:\Vext(P)\to\Vext(P)$ is Scott-continuous.

For $p\in U$, upperness gives $\up p\subseteq U$, whence
\[
  \sum_{p\in U}\nu(\up p)\leq n\nu(U).
\]
It follows from \eqref{eq:extended-preprocessing-open} that
$S_m\nu\stle\nu$. Moreover, direct expansion gives
\begin{equation}\label{eq:extended-preprocessing-increasing}
  S_{m+1}\nu
  =\frac m{m+1}S_m\nu+\frac1{m+1}\nu.
\end{equation}
Using $S_m\nu\stle\nu$ in this equality yields
$S_m\nu\stle S_{m+1}\nu$. Finally,
\[
  \left(1-\frac nm\right)\nu(U)
  \leq(S_m\nu)(U)\leq\nu(U).
\]
If $\nu(U)<\infty$, the lower bound tends to $\nu(U)$ as
$m\to\infty$. If $\nu(U)=\infty$, that lower bound is already
infinite for every $m>n$. Consequently
\begin{equation}\label{eq:extended-preprocessing-limit}
  S_m\leq S_{m+1}\leq\id_{\Vext(P)},
  \qquad \sup_{m>n}S_m=\id_{\Vext(P)}.
\end{equation}

\emph{The finite-minimum formula.}
Put
\begin{equation}\label{eq:extended-simple-thresholds}
  c_m(p)=m^{d(p)},
  \qquad R_m=\sum_{p\in P}c_m(p).
\end{equation}
Thus $0<R_m<\infty$, and
\begin{equation}\label{eq:extended-simple-threshold-growth}
  p<q\ \Longrightarrow\ c_m(q)\geq m c_m(p).
\end{equation}
For $\nu\in\Vext(P)$ and $U\in\sigma(P)$, define the set
function
\begin{equation}\label{eq:extended-minimum-approximant}
  a_m(\nu)(U)
  =\min_{\substack{W\in\sigma(P)\\ W\subseteq U}}
    \left((S_m\nu)(W)+\sum_{p\in U\setminus W}c_m(p)\right).
\end{equation}
There are only finitely many candidates, and $W=\varnothing$ is
always allowed. Hence this expression is finite, even if $\nu$ has
infinite total mass. We next prove that it defines a valuation;
this does not follow merely from taking a finite minimum.

\emph{Atomic representations and the valuation property.}
Every $\nu\in\Vext(P)$ admits a representation
\begin{equation}\label{eq:extended-atomic-existence}
  \nu=\sum_{p\in P}r_p\delta_p,
  \qquad r_p\in[0,\infty].
\end{equation}
To see this without subtracting infinite values, let
\[
  O=\bigcup\{U\in\sigma(P):\nu(U)<\infty\}.
\]
There are finitely many upper sets. Modularity and nonnegativity
imply $\nu(U\cup W)\leq\nu(U)+\nu(W)$, so $\nu(O)<\infty$.
Every upper subset of $O$ is upper in $P$. Apply
\eqref{eq:atomic-representation} to the subprobabilistic valuation
obtained by dividing $\nu|_O$ by $\max\{1,\nu(O)\}$.
Multiplying back and extending the resulting atoms by zero outside
$O$ gives a finite valuation $\zeta$ on $P$ such that
\[
  \zeta(U)=\nu(U\cap O)
  \qquad(U\in\sigma(P)).
\]
This also covers $O=\varnothing$, with $\zeta=0$.
Let $M=\Max(P\setminus O)$. If $U\subseteq O$, then
$U\cap M=\varnothing$ and $\zeta(U)=\nu(U)$.
If $U\nsubseteq O$, then $\nu(U)=\infty$ by the definition of $O$.
Choose $x\in U\setminus O$ and a maximal element $z$ of the finite
poset $P\setminus O$ above $x$. Since $U$ is upper, $z\in U\cap M$.
It follows that
\[
  \nu=\zeta+\sum_{z\in M}\infty\,\delta_z,
\]
proving \eqref{eq:extended-atomic-existence}. Such a representation
need not be unique; only its existence will be used.

Fix one representation \eqref{eq:extended-atomic-existence}.
Then \eqref{eq:extended-simple-preprocessing} gives
\[
  S_m\nu=\sum_{p\in P}b_p\delta_p,
  \qquad
  b_p=\left(1-\frac nm\right)r_p
       +\frac1m\sum_{z\geq p}r_z.
\]
If $p<q$, then
\[
  b_p\geq\frac1m\sum_{z\geq q}r_z,
  \qquad
  b_q=\left(1-\frac{n-1}{m}\right)r_q
        +\frac1m\sum_{z>q}r_z
      \leq\sum_{z\geq q}r_z.
\]
Therefore
\begin{equation}\label{eq:extended-atomic-threshold-comparison}
  p<q\ \Longrightarrow\ b_q\leq m b_p.
\end{equation}
Combining \eqref{eq:extended-atomic-threshold-comparison} with
\eqref{eq:extended-simple-threshold-growth} shows that
\[
  T=\{p\in P:b_p<c_m(p)\}
\]
is upper. Indeed, if $p\in T$ and $p<q$, then
\[
  b_q\leq m b_p<m c_m(p)\leq c_m(q),
\]
so $q\in T$.

For every upper set $W\subseteq U$,
\[
  (S_m\nu)(W)+\sum_{p\in U\setminus W}c_m(p)
  =\sum_{p\in W}b_p+\sum_{p\in U\setminus W}c_m(p)
  \geq\sum_{p\in U}\min\{b_p,c_m(p)\}.
\]
Equality holds for $W=U\cap T$, which is an allowed upper set.
Thus \eqref{eq:extended-minimum-approximant} is exactly
\begin{equation}\label{eq:extended-minimum-atomic-form}
  a_m(\nu)(U)=\sum_{p\in U}\min\{b_p,c_m(p)\},
  \qquad
  a_m(\nu)=\sum_{p\in P}\min\{b_p,c_m(p)\}\delta_p.
\end{equation}
This proves strictness, modularity, monotonicity in $U$, and
continuity on directed unions: the right-hand side is a finite sum
of finite multiples of Dirac valuations. In particular,
$a_m(\nu)\in\mathcal V_{\leq R_m}(P)$.
Although the coefficients $r_p$ were chosen for this verification,
\eqref{eq:extended-minimum-approximant} depends only on $\nu$.
Consequently $a_m(\nu)$ is independent of that choice.

\emph{Scott continuity and order bounds.}
For fixed $m,U,W$, the candidate
\[
  \nu\longmapsto(S_m\nu)(W)+\sum_{p\in U\setminus W}c_m(p)
\]
is monotone and preserves directed suprema. Taking the finite
minimum in \eqref{eq:extended-minimum-approximant} preserves
monotonicity. If $(\nu_i)_{i\in I}$ is directed with supremum
$\nu$, \eqref{eq:extended-finite-min-sup} gives
\[
\begin{aligned}
  a_m(\nu)(U)
  &=\min_{\substack{W\in\sigma(P)\\W\subseteq U}}
      \sup_i\left((S_m\nu_i)(W)
                     +\sum_{p\in U\setminus W}c_m(p)\right)\\
  &=\sup_i\min_{\substack{W\in\sigma(P)\\W\subseteq U}}
      \left((S_m\nu_i)(W)
                     +\sum_{p\in U\setminus W}c_m(p)\right)\\
  &=\sup_i a_m(\nu_i)(U).
\end{aligned}
\]
The common bound $R_m$ is preserved by directed suprema. Hence
$a_m:\Vext(P)\to\mathcal V_{\leq R_m}(P)$ is
Scott-continuous.

Taking $W=U$ in \eqref{eq:extended-minimum-approximant} yields
$a_m(\nu)(U)\leq(S_m\nu)(U)\leq\nu(U)$.
For each fixed $U,W$, its candidate increases with $m$, by
\eqref{eq:extended-preprocessing-limit} and
$c_m(p)\leq c_{m+1}(p)$. Taking minima over the same finite set of
upper sets therefore gives
\begin{equation}\label{eq:extended-simple-approx-order}
  a_m\leq a_{m+1}\leq\id_{\Vext(P)},
  \qquad a_m(\nu)(P)\leq R_m.
\end{equation}
In these inequalities, the maps are regarded as self-maps of $\Vext(P)$.

\emph{The supremum of the approximants.}
Fix $\nu\in\Vext(P)$ and $U\in\sigma(P)$.
The candidate with $W=U$ has supremum $\nu(U)$, by
\eqref{eq:extended-preprocessing-limit}.
For $W\subsetneq U$, the set $U\setminus W$ is nonempty and
$d(p)\geq1$ for its elements, so
\[
  (S_m\nu)(W)+\sum_{p\in U\setminus W}m^{d(p)}\geq m.
\]
Every such candidate has supremum $\infty$.
Since each candidate increases with $m$, another application of
\eqref{eq:extended-finite-min-sup} gives
\[
\begin{aligned}
  \sup_{m>n}a_m(\nu)(U)
  &=\min_{\substack{W\in\sigma(P)\\W\subseteq U}}
     \sup_{m>n}\left((S_m\nu)(W)
                      +\sum_{p\in U\setminus W}m^{d(p)}\right)\\
  &=\nu(U).
\end{aligned}
\]
For $U=\varnothing$ there is just the single candidate $W=\varnothing$,
and the same equality holds with both sides zero. This proves
\begin{equation}\label{eq:extended-simple-approx-sup}
  \sup_{m>n}a_m=\id_{\Vext(P)}
\end{equation}
on the whole of $\Vext(P)$, including valuations of infinite
total mass.

\emph{Finite-valuation factorization.}
Define
\[
\begin{aligned}
  p_m:\Vext(P)&\longrightarrow\Vsub(P),
  &p_m(\nu)&=R_m^{-1}a_m(\nu),\\
  e_m:\Vsub(P)&\longrightarrow\Vext(P),
  &e_m(\lambda)&=R_m\lambda.
\end{aligned}
\]
The mass bound in \eqref{eq:extended-simple-approx-order} makes
$p_m$ well-defined. Both maps are Scott-continuous, since $R_m$ is
a fixed positive finite constant and scalar multiplication preserves
directed suprema. Their composite is $e_m p_m=a_m$.
Reindexing the integers $m>n$ by $\N$, the factorizations
\[
  \Vext(P)\xrightarrow{\ p_m\ }\Vsub(P)
     \xrightarrow{\ e_m\ }\Vext(P)
\]
together with \eqref{eq:extended-simple-approx-order} and
\eqref{eq:extended-simple-approx-sup} satisfy precisely
\cref{def:FVA}. Hence $\Vext(P)\in\FVA$.
\end{proof}

\subsection{Closure under extended valuations}
\label{subsec:extended-global-closure}

For the argument below, we use integration with respect to
arbitrary $\nu\in\Vext(X)$, not only subprobabilistic valuations.
For a Scott-continuous $h:X\to[0,\infty]$, the Choquet integral is
\[
  \int_X h\,d\nu
  =\int_0^\infty\nu\bigl(\{x:h(x)>t\}\bigr)\,dt,
\]
where the right-hand side is the nonnegative integral in the real
variable $t$, possibly infinite. This extends the subprobability
integral used in Section~\ref{subsec:prelim-valuations}.
It is additive, positively homogeneous and Scott-continuous in
each argument separately, and
\[
  \int_X\chi_U\,d\nu=\nu(U)
  \quad(U\in\sigma(X)),
  \qquad \int_X1\,d\nu=\nu(X);
\]
see \cite[arXiv:2106.16190v3, Section~2.2, p.~8]
{GoubaultLarrecqJiaTheron2023}. In particular, for a directed family
$(h_i)_{i\in I}$ of Scott-continuous maps $X\to[0,\infty]$,
\[
  \int_X\sup_{i\in I}h_i\,d\nu
  =\sup_{i\in I}\int_Xh_i\,d\nu,
\]
with no finite-total-mass assumption on $\nu$.

For a Scott-continuous kernel $k:X\to\Vext(Y)$, the extended
Kleisli map
\[
  k^\dagger:\Vext(X)\longrightarrow\Vext(Y),
  \qquad
  (k^\dagger\nu)(U)=\int_Xk(x)(U)\,d\nu(x)
  \quad(U\in\sigma(Y))
\]
is well-defined and Scott-continuous by
\cite[Lemma~3.1(ii)]
{GoubaultLarrecqJiaTheron2023}. Here $k^\dagger$ denotes the
extension on all continuous valuations, whereas in the main text
its domain and codomain are subprobabilistic or probabilistic
powerdomains. The Scott-continuous pushforward is
\[
  \Vext(f):\Vext(X)\longrightarrow\Vext(Y),
  \qquad
  (\Vext(f)\nu)(U)=\nu\bigl(f^{-1}(U)\bigr)
  \quad(U\in\sigma(Y)),
\]
for Scott-continuous $f:X\to Y$; see
\cite[Proposition~3.2]
{GoubaultLarrecqJiaTheron2023}.

\begin{lemma}\label{lem:extended-kernel-lifting}
Let $X$ be a dcpo. Suppose that there are finite posets $L_n$ and
Scott-continuous maps
\[
  p_n:X\longrightarrow\Vone(L_n),
  \qquad y_n:L_n\longrightarrow X
\]
such that
\[
  \kappa_n=\Vone(y_n)\circ p_n,
  \qquad
  \kappa_n\leq\kappa_{n+1}\leq\etaV_X,
  \qquad \sup_n\kappa_n=\etaV_X.
\]
Then $\Vext(X)\in\FVA$.
\end{lemma}

\begin{proof}
The canonical inclusion
$\jmath_n:\Vone(L_n)\hookrightarrow\Vext(L_n)$ is
Scott-continuous: directed suprema in both dcpos are computed
pointwise on upper sets, and a directed supremum of probabilistic
valuations still has total mass one. Define
\[
  Q_n=(\jmath_n\circ p_n)^\dagger:
     \Vext(X)\longrightarrow\Vext(L_n),
  \qquad
  E_n=\Vext(y_n):
     \Vext(L_n)\longrightarrow\Vext(X).
\]
By the extended integration and pushforward results just recalled,
both maps are well-defined and Scott-continuous. Explicitly,
\[
  (Q_n\nu)(A)=\int_Xp_n(x)(A)\,d\nu(x)
  \qquad(A\in\sigma(L_n)).
\]
Since $p_n(x)(L_n)=1$, their intermediate valuation has total mass
\[
  (Q_n\nu)(L_n)=\int_X1\,d\nu=\nu(X).
\]
Thus $Q_n$ takes values in the extended powerdomain
$\Vext(L_n)$; no finite-total-mass assumption is made on $\nu$.

Let $T_n=E_n\circ Q_n$. For $U\in\sigma(X)$, continuity of $y_n$
ensures that $y_n^{-1}(U)\in\sigma(L_n)$, and the definitions give
\[
\begin{aligned}
  (T_n\nu)(U)
  &=(Q_n\nu)(y_n^{-1}(U))\\
  &=\int_Xp_n(x)(y_n^{-1}(U))\,d\nu(x)\\
  &=\int_X\kappa_n(x)(U)\,d\nu(x).
\end{aligned}
\]
Evaluation of a valuation at a fixed open set is Scott-continuous.
Hence, for each $U$, the maps $x\mapsto\kappa_n(x)(U)$ are
Scott-continuous and increase pointwise to
$x\mapsto\delta_x(U)=\chi_U(x)$. Monotonicity of the integral gives
\[
  (T_n\nu)(U)\leq(T_{n+1}\nu)(U)
      \leq\int_X\chi_U\,d\nu=\nu(U).
\]
Its Scott continuity in the integrand then yields
\[
\begin{aligned}
  \sup_n(T_n\nu)(U)
  &=\int_X\sup_n\kappa_n(x)(U)\,d\nu(x)\\
  &=\int_X\chi_U\,d\nu
   =\nu(U).
\end{aligned}
\]
Consequently
\[
  T_n\leq T_{n+1}\leq\id_{\Vext(X)},
  \qquad \sup_nT_n=\id_{\Vext(X)}.
\]
Each $T_n$ factors through $\Vext(L_n)$, and
$\Vext(L_n)\in\FVA$ by \cref{thm:finite-extended-FVA}.
These factorizations therefore witness
$\Vext(X)\in\mathcal F(\FVA)$ in the sense of
\cref{def:C-factorization}. Applying
$\mathcal F(\FVA)=\FVA$ from
\cref{thm:FVA-property}\textup{(2)} proves the claim.
\end{proof}

\begin{theorem}\label{thm:extended-V-closure}
For every $D\in\FVA$, the extended probabilistic powerdomain
$\Vext(D)$ belongs to $\FVA$.
\end{theorem}

\begin{proof}
We first verify closure on the generating objects:
\begin{equation}\label{eq:extended-on-finite-generator}
  \Vext(\Vsub(P))\in\FVA
  \qquad\text{for every finite poset }P.
\end{equation}
If $P=\varnothing$, then $\Vsub(P)$ is the one-point dcpo, so
\eqref{eq:extended-on-finite-generator} follows from
\cref{thm:finite-extended-FVA} applied to the one-point poset.
If $P\neq\varnothing$, evaluation at the unique point of
$\mathbf 1$ identifies $\Vsub(P)$ with $\M(\mathbf 1,P)$.
Apply \cref{prop:finite-kernels} with this choice of source and
target. It supplies finite posets $L_n$ and maps $p_n,y_n$ whose
kernels satisfy all the hypotheses of
\cref{lem:extended-kernel-lifting}. That lemma gives
\eqref{eq:extended-on-finite-generator}.

Now fix $D\in\FVA$. By \cref{def:FVA}, choose Scott-continuous
maps
\[
  p_n:D\longrightarrow\Vsub(P_n),
  \qquad e_n:\Vsub(P_n)\longrightarrow D
\]
for finite posets $P_n$, such that
\[
  a_n=e_n\circ p_n,
  \qquad a_n\leq a_{n+1}\leq\id_D,
  \qquad\sup_n a_n=\id_D.
\]
Functoriality of pushforward gives Scott-continuous factorizations
\[
  \Vext(a_n)=\Vext(e_n)\circ\Vext(p_n):
  \Vext(D)\longrightarrow
  \Vext(\Vsub(P_n))\longrightarrow\Vext(D).
\]
We verify that these form an increasing approximate identity without
any mass restriction. Fix $U\in\sigma(D)$. Since $U$ is upper and
$a_n\leq a_{n+1}\leq\id_D$,
\[
  a_n^{-1}(U)\subseteq a_{n+1}^{-1}(U)\subseteq U.
\]
If $x\in U$, then $x=\sup_n a_n(x)$, and Scott openness of $U$
implies $a_n(x)\in U$ for some $n$. Thus
\[
  \bigcup_n a_n^{-1}(U)=U.
\]
For any $\nu\in\Vext(D)$, monotonicity and preservation of
directed unions now give
\[
\begin{aligned}
  (\Vext(a_n)\nu)(U)
  &\leq(\Vext(a_{n+1})\nu)(U)\leq\nu(U),\\
  \sup_n(\Vext(a_n)\nu)(U)
  &=\sup_n\nu(a_n^{-1}(U))\\
  &=\nu\left(\bigcup_n a_n^{-1}(U)\right)
   =\nu(U).
\end{aligned}
\]
These equalities also apply when $\nu(U)=\infty$. Hence
\[
  \Vext(a_n)\leq\Vext(a_{n+1})
       \leq\id_{\Vext(D)},
  \qquad\sup_n\Vext(a_n)=\id_{\Vext(D)}.
\]
By \eqref{eq:extended-on-finite-generator}, every intermediate
object $\Vext(\Vsub(P_n))$ belongs to $\FVA$.
Therefore $\Vext(D)\in\mathcal F(\FVA)$, and a second
application of \cref{thm:FVA-property}\textup{(2)} gives
$\Vext(D)\in\FVA$.
\end{proof}

\section*{Acknowledgements}
During the preparation of this manuscript, the authors used AI-assisted tools for language polishing and grammar checking. The authors carefully reviewed and verified the final manuscript and take full responsibility for its content, including the correctness of all mathematical statements, proofs, and references. 

\end{document}